\documentclass[11pt,letterpaper,onecolumn]{article}

\newif\ifdraft
\draftfalse

\usepackage[letterpaper,margin=1in]{geometry}
\usepackage{microtype}
\usepackage{lmodern}
\usepackage{setspace}
\usepackage{amsmath,amsfonts,amssymb,amsthm}
\usepackage{thmtools,thm-restate}
\usepackage{mathtools}
\usepackage{nicefrac}

\usepackage{graphicx}
\usepackage{subcaption}
\usepackage{xcolor}
\usepackage{colortbl}
\usepackage{booktabs}
\usepackage{tabularx}
\usepackage{tcolorbox}
\usepackage{enumitem}
\usepackage{comment}

\usepackage{tikz}
\usetikzlibrary{calc,arrows.meta,decorations.markings}
\usepackage{pgfplots}
\pgfplotsset{compat=newest}
\usepackage{pgfplotstable}

\pgfplotsset{
  discard if not/.style 2 args={
    x filter/.code={
      \edef\tempa{\thisrow{#1}}
      \edef\tempb{#2}
      \ifx\tempa\tempb\else
        
      \fi
    }
  }
}
\pgfmathdeclarefunction{lg2}{1}{\pgfmathparse{ln(#1)/ln(2)}}
\pgfmathdeclarefunction{lg10}{1}{\pgfmathparse{ln(#1)/ln(10)}}

\usepackage[ruled,vlined]{algorithm2e}
\DontPrintSemicolon

\usepackage[square,numbers]{natbib}
\usepackage[pagebackref,colorlinks=true,
  linkcolor=red!40!black,
  citecolor=green!40!black,
  urlcolor=blue!40!black,
  filecolor=magenta!40!black,
  pdfdisplaydoctitle=true
]{hyperref}
\providecommand{\doi}[1]{\href{https://doi.org/\detokenize{#1}}{doi: \nolinkurl{#1}}}
\usepackage[nameinlink,sort&compress,capitalize]{cleveref}

\providecommand*{\backref}[1]{}
\renewcommand*{\backref}[1]{}
\providecommand*{\backrefalt}[4]{}
\renewcommand*{\backrefalt}[4]{%
  \ifcase #1%
  \or [p.~#2.]%
  \else [pp.~#2.]%
  \fi%
}

\ifdraft
  \usepackage[mathlines]{lineno}
  \usepackage[backgroundcolor=gray!10,textsize=footnotesize]{todonotes}
\else
  \newcommand{\todo}[2][]{}
\fi

\newtheorem{theorem}{Theorem}[section]
\newtheorem{lemma}[theorem]{Lemma}

\theoremstyle{definition}

\crefname{rrule}{Rule}{Rules}

\newcommand{\boxproblem}[4]{%
  \begin{center}
    \begin{tcolorbox}[
        width=0.96\linewidth,
        colback=white,
        colframe=black!35,
        boxrule=0.45pt,
        arc=4pt,
        outer arc=4pt,
        boxsep=0pt,
        left=7pt,
        right=7pt,
        top=6pt,
        bottom=6pt]
      \centering
      {\normalsize\scshape #1\par}
      \vspace{3pt}
      {\color{black!35}\hrule height 0.45pt}
      \vspace{5pt}
      \setlength{\tabcolsep}{4pt}%
      \renewcommand{\arraystretch}{1.12}%
      \normalfont
      \begin{tabularx}{\linewidth}{@{}>{\normalfont\normalsize\bfseries}l@{\hspace{0.7em}}>{\normalfont\normalsize}X@{}}
        Input:      & #2 \\
        Task:       & #3 \\
        Parameter:  & #4
      \end{tabularx}
    \end{tcolorbox}
  \end{center}
}

\newcommand{\boxproblemnoparam}[3]{%
  \begin{center}
    \begin{tcolorbox}[
        width=0.96\linewidth,
        colback=white,
        colframe=black!35,
        boxrule=0.45pt,
        arc=4pt,
        outer arc=4pt,
        boxsep=0pt,
        left=7pt,
        right=7pt,
        top=6pt,
        bottom=6pt]
      \centering
      {\normalsize\scshape #1\par}
      \vspace{3pt}
      {\color{black!35}\hrule height 0.45pt}
      \vspace{5pt}
      \setlength{\tabcolsep}{4pt}%
      \renewcommand{\arraystretch}{1.12}%
      \normalfont
      \begin{tabularx}{\linewidth}{@{}>{\normalfont\normalsize\bfseries}l@{\hspace{0.7em}}>{\normalfont\normalsize}X@{}}
        Input:  & #2 \\
        Task:   & #3
      \end{tabularx}
    \end{tcolorbox}
  \end{center}
}

\newcommand{\ZZ}{\mathbb{Z}}

\newcommand{\poly}{\textnormal{poly}}

\title{A Single-Exponential FPT Algorithm \\ for 2-Vertex-Connectivity Augmentation}
\author{%
  Tomohiro Koana\thanks{Graduate School of Information Science and Technology, The University of Tokyo, Japan. Email: \href{mailto:tomohiro.koana@gmail.com}{\nolinkurl{tomohiro.koana@gmail.com}}}
  \and
  Soh Kumabe\thanks{CyberAgent, Inc., Tokyo, Japan. Email: \href{mailto:kumabe_soh@cyberagent.co.jp}{\nolinkurl{kumabe_soh@cyberagent.co.jp}}}
}
\date{}

\begin{document}
\ifdraft\linenumbers\fi
\maketitle

\begin{abstract}
We study restricted-link augmentation to $2$-vertex-connectivity.
An instance consists of a graph $G$, possibly disconnected, a set $L$ of admissible links on its vertices, integer link costs in $\{1,\dots,W\}$, and an integer $k$; the task is to add at most $k$ links of minimum total cost so that the resulting multigraph is $2$-vertex-connected.
Recent work gives $O^*(k^{O(k)})$-time algorithms for unweighted $\lambda$-vertex-connectivity augmentation for every $\lambda\leq 4$ [Carmesin and Ramanujan, SODA 2026], and an $O^*((k+\lambda)^{O(k)})$-time algorithm for arbitrary $\lambda$ [Korhonen and Thorup, FOCS 2026].
We give a deterministic algorithm with running time $O^*(36^kW)$.
Thus, for $\lambda=2$, the unweighted running time improves from $O^*(k^{O(k)})$ to $O^*(36^k)$, and the algorithm also handles link costs with pseudo-polynomial dependence on $W$.

We reduce the problem to a boundary-pair variant of \(2\)-vertex-connected spanning subgraph, where each vertex is assigned a pair of incident edges with an associated pair cost.
We solve this variant using a cancellation identity, inspired by Cut\&Count~[Cygan et al., TALG 2022], obtained by applying M\"obius inversion to decompositions along cut vertices: the identity cancels every connected spanning graph with more than one block and keeps exactly the \(2\)-vertex-connected spanning graphs.
\end{abstract}

\newpage

\section{Introduction}

Connectivity augmentation asks which additional edges, called links, should be added to a graph to meet a prescribed edge- or vertex-connectivity requirement.
In this paper, the available links are part of the input: besides a graph $G=(V,E)$, we are given a set $L$ of links on $V$ with costs, and the objective is to add at most $k$ links from $L$ so as to achieve the required connectivity at minimum total cost.

For $\lambda=2$, Eswaran and Tarjan, and independently Ple\v{s}n{\'i}k, gave polynomial-time algorithms in the unrestricted unit-cost model, where every missing edge is available as a link~\citep{DBLP:journals/siamcomp/EswaranT76,Plesnik76}.
Later work gave polynomial-time algorithms for higher edge-connectivity requirements in unrestricted augmentation models~\citep{DBLP:conf/focs/Frank90,DBLP:journals/jct/FrankJ95,DBLP:journals/jcss/WatanabeN87}.
For vertex-connectivity in the unrestricted model, Jackson and Jord{\'a}n showed fixed-parameter tractability with respect to~$\lambda$, and V{\'e}gh gave a polynomial-time algorithm for the augment-by-one case in which the input graph is already $(\lambda-1)$-vertex-connected~\citep{DBLP:journals/jct/JacksonJ05a,DBLP:journals/siamdm/Vegh11}.
It remains open whether unrestricted vertex-connectivity augmentation is polynomial-time solvable or NP-hard when $\lambda$ is part of the input~\citep{DBLP:journals/corr/abs-2605-11757}.

By contrast, the restricted-link model already contains a classical hard problem.
Indeed, if $G$ is the empty graph on $n$ vertices and $L=E(H)$ for an input graph $H$, then deciding whether $G$ can be augmented to $2$-edge- or $2$-vertex-connectivity using at most $n$ links is exactly the \textsc{Hamiltonian Cycle} problem in~$H$.
Thus the problem is NP-hard already for $\lambda=2$.
This hardness has motivated extensive work on approximation algorithms~\citep{DBLP:journals/mp/AngelidakisHS23,DBLP:journals/siamcomp/ByrkaGA23,DBLP:journals/mp/CecchettoTZ24,DBLP:journals/siamcomp/FredericksonJ81,DBLP:journals/combinatorica/Jain01,DBLP:journals/jal/KhullerT93,DBLP:journals/siamcomp/KortsarzKL04,DBLP:journals/tcs/Nutov24,DBLP:conf/focs/TraubZ21,DBLP:conf/stoc/TraubZ23} and on parameterized algorithms with respect to the number $k$ of added links~\citep{DBLP:conf/icalp/BasavarajuFGMRS14,DBLP:conf/soda/Carmesin026,DBLP:journals/networks/GuoU10,DBLP:journals/corr/abs-2605-11757,DBLP:journals/talg/MarxV15,DBLP:journals/dam/Nagamochi03,DBLP:conf/esa/Nutov24}.


Most parameterized results concern augment-by-one instances, where the input is already one connectivity level below the target.
For edge connectivity, augmentation from \(\lambda-1\) to \(\lambda\) amounts to selecting links that cross every minimum cut.
Nagamochi gave an \(O^*(k^{O(k)})\)-time algorithm for the minimum-cardinality case of tree augmentation~\citep{DBLP:journals/dam/Nagamochi03}.\footnote{$O^*(\cdot)$ suppresses factors polynomial in the input size.}
Guo and Uhlmann obtained \(O(k^2)\)-size kernels for both edge- and vertex-connectivity augmentation~\citep{DBLP:journals/networks/GuoU10}.
Marx and V{\'e}gh gave an \(O^*(k^{O(k)})\)-time algorithm and a polynomial kernel for minimum-cost edge-connectivity augment-by-one~\citep{DBLP:journals/talg/MarxV15}.

The single-exponential algorithms known for augment-by-one problems use reductions to \textsc{Steiner Tree} or \textsc{Group Steiner Tree}.
Basavaraju et al.\ reduced minimum-cost edge-connectivity augment-by-one to \textsc{Steiner Tree} and obtained an \(O^*(9^k)\)-time algorithm~\citep{DBLP:conf/icalp/BasavarajuFGMRS14}.
Nutov used \textsc{Group Steiner Tree} to obtain the same bound for vertex-connectivity augmentation from \(\lambda-1\) to \(\lambda\) for every \(\lambda\leq 3\)~\citep{DBLP:conf/esa/Nutov24}.
The cactus representation of minimum cuts reduces the relevant edge-connectivity cases to augmentation from \(1\) to \(2\) or from \(2\) to \(3\)~\citep{DinitzKL76}, and \Cref{sec:preliminaries} recalls the standard reduction from edge-connectivity augmentation to vertex-connectivity augmentation.
Hence Nutov's result also covers the corresponding single-exponential edge-connectivity case.

Beyond augment-by-one, Marx and V{\'e}gh handled the edge-connectivity case of raising connectivity from $0$ to $2$ in time $O^*(k^{O(k)})$~\citep{DBLP:journals/talg/MarxV15}.
For vertex connectivity, fixed-parameter tractability on arbitrary inputs remained open until the recent unweighted results summarized below.

Carmesin and Ramanujan proved that arbitrary-input \(\lambda\)-vertex-connectivity augmentation is solvable in time \(O^*(k^{O(k)})\) for every \(\lambda\leq 4\)~\citep{DBLP:conf/soda/Carmesin026}.
Korhonen and Thorup subsequently obtained an \(O^*(k^{O(k)})\)-time algorithm for arbitrary-input \(\lambda\)-edge-connectivity augmentation and an \(O^*((k+\lambda)^{O(k)})\)-time algorithm for arbitrary-input \(\lambda\)-vertex-connectivity augmentation~\citep{DBLP:journals/corr/abs-2605-11757}.
Both results are unweighted.

These algorithms are fixed-parameter tractable, but the dependence on $k$ is not single-exponential.
The known single-exponential algorithms for augment-by-one problems rely on reductions to \textsc{Steiner Tree}, and no corresponding reduction is known for arbitrary inputs.
We give a single-exponential algorithm for \(\lambda=2\).\footnote{After contracting the connected components of \(G\), the case \(\lambda=1\) is a minimum-spanning-tree problem on the admissible links.}

\begin{table}[t]
  \centering
  \small
  \setlength{\tabcolsep}{4pt}
  \renewcommand{\arraystretch}{1.12}
  \begin{tabularx}{\linewidth}{@{\hspace{5pt}}>{\raggedright\arraybackslash}p{0.14\linewidth}>{\raggedright\arraybackslash}X>{\raggedright\arraybackslash}p{0.22\linewidth}>{\raggedright\arraybackslash}p{0.10\linewidth}>{\raggedright\arraybackslash}p{0.13\linewidth}@{}}
    \toprule
    Connectivity & Augmentation & Running time & Costs & Ref. \\
    \midrule
    Edge
      & Augment-by-one from 1 to 2.
      & $O^*(k^{O(k)})$
      & No
      & \citep{DBLP:journals/dam/Nagamochi03} \\
    Edge
      & Arbitrary input, $\lambda = 2$.
      & $O^*(k^{O(k)})$
      & Yes
      & \citep{DBLP:journals/talg/MarxV15} \\
    Edge
      & Augment-by-one.
      & $O^*(9^k)$
      & Yes
      & \citep{DBLP:conf/icalp/BasavarajuFGMRS14} \\
    Vertex
      & Augment-by-one from $2$ to $3$.
      & $O^*(9^k)$
      & Yes
      & \citep{DBLP:conf/esa/Nutov24} \\
    Vertex
      & Arbitrary input, for every $\lambda\le 4$.
      & $O^*(k^{O(k)})$
      & No
      & \citep{DBLP:conf/soda/Carmesin026} \\
    Edge
      & Arbitrary input.
      & $O^*(k^{O(k)})$
      & No
      & \citep{DBLP:journals/corr/abs-2605-11757} \\
	    Vertex
	      & Arbitrary input.
	      & $O^*((k + \lambda)^{O(k)})$
	      & No
	      & \citep{DBLP:journals/corr/abs-2605-11757} \\
	    \arrayrulecolor{black!35}\midrule\arrayrulecolor{black}
	    Vertex
	      & Arbitrary input, $\lambda=2$.
	      & $O^*(36^kW)$
	      & Yes
	      & This work \\
	    \bottomrule
  \end{tabularx}
  \caption{Known parameterized algorithms for restricted-link connectivity augmentation. Here, \(k\) is the number of added links, \(\lambda\) is the target connectivity, and \(W\) is the maximum link cost. The weighted bound in the last row has pseudo-polynomial dependence on \(W\); its unweighted specialization is FPT parameterized by \(k\).}
  \label{tab:known-fpt-algorithms}
\end{table}

Formally, we study the following problem.

\boxproblem{2-Vertex-Connectivity Augmentation (2VCA)}
{A graph $G=(V,E)$, a set $L$ of links on $V$, an integer $k$, an integer $W\geq 1$, and link costs $c\colon L\to \{1,\dots,W\}$.}
{Compute a minimum-cost set $A\subseteq L$ such that $|A|\leq k$ and the augmented multigraph $(V,E\cup A)$ is $2$-vertex-connected, or report that no such set exists.}
{$k$}

We prove the following.

\begin{restatable}{theorem}{twovcamainthm}\label{thm:main_runtime}
\textsc{2VCA} with positive integer link costs bounded by \(W\) can be solved deterministically in \(O^*(36^kW)\) time.
\end{restatable}

By the standard cost-preserving reduction from edge-connectivity augmentation to vertex-connectivity augmentation, the same bounds hold for 2-edge-connectivity augmentation.

The proof first transforms a \textsc{2VCA} instance into an equivalent instance of \textsc{Path Forest Augmentation} (\textsc{PFA}), which is the same augmentation problem where the input graph is restricted to a path forest with at most \(2k\) paths.
After contracting each path, we solve \textsc{PFA} by subset dynamic programming on the contracted graph.
The main transition of this dynamic programming is an auxiliary problem, \textsc{\(2\)-Vertex-Connected Spanning Subgraph with Boundary-Pair Costs} (\textsc{2VCSS-BP}), which asks for a \(2\)-vertex-connected spanning subgraph with additional local constraints.
The key step is a Cut\&Count-inspired~\citep{CyganNPPVW11} M\"obius-inversion filter over decompositions along cut vertices: in the resulting signed sum, every graph with a cut vertex cancels, while exactly the graphs consisting of one block survive.
\Cref{sec:technical-overview} gives the details.

Two questions remain.
Can the pseudo-polynomial dependence on \(W\) be replaced by dependence polynomial in the bit length of the costs?
Can the \(2^{O(k)}\) dependence be extended to \(\lambda=3\)?

\textbf{Organization.}
Section 2 gives a technical overview of the algorithm.
Section 3 gives the graph-theoretic notation and algebraic tools.
Section 4 reduces \textsc{2VCA} to \textsc{PFA}, Section 5 reduces \textsc{PFA} to \textsc{2VCSS-BP}, and Section 6 defines the polynomial that encodes feasible \textsc{2VCSS-BP} solutions.
Section 7 evaluates the resulting polynomial and proves \Cref{thm:main_runtime}.

\section{Technical Overview}\label{sec:technical-overview}

We outline the algorithm behind \Cref{thm:main_runtime}, focusing on where the single-exponential dependence on \(k\) comes from.
For simplicity, this overview treats the unweighted case, where every link has unit cost and the task is simply to minimize the number of added links.
The weighted algorithm follows the same dynamic programming structure, with additional care for the reduction and cost bookkeeping.

\subsection{Reduction to Path Forests}

\paragraph*{From graphs to path forests.}
By modifying the graph appropriately and treating some original edges as zero-cost links, we reduce to an equivalent \textsc{PFA} instance, which is the same augmentation problem where the input graph is restricted to a path forest with at most \(2k\) paths.
This is proved in \Cref{lem:twovca_to_twovca_pf}; see \Cref{fig:twovca-to-twovca-pf-reduction} for an illustration.
We write \(G^{\mathrm{pf}}\) for this path forest and \(\mathcal P\) for its set of paths.

\paragraph*{Contracting paths.}
Contract each \(P\in\mathcal P\) to a vertex \(v_P\), and map every link to the corresponding edge between contracted vertices.
Let \(H\) be the resulting multigraph.
Since \(|V(H)|=|\mathcal P|\leq 2k\), an algorithm with running time \(O^*(2^{O(|V(H)|)})\) is single-exponential in \(k\).
Contraction does not, however, preserve all information relevant to vertex connectivity.
The two graphs in \Cref{fig:overview-contraction-cover}(a) and \Cref{fig:overview-contraction-cover}(b) contract to the same graph, although only the first is 2-vertex-connected.

For a candidate edge set \(A\subseteq E(H)\), write \(H_A=(V(H),A)\).
Let \(P=(u_1,\ldots,u_\ell)\), and let \(B\) be a block of \(H_A\), that is, a maximal connected subgraph with no cut vertex, containing \(v_P\).
We say that \(B\) \emph{covers} index \(t\) at \(P\) if two incidences of \(B\) at \(v_P\) come from links whose endpoints on \(P\) lie on opposite sides of \(u_t\).
Equivalently, the two links meet \(P\) at \(u_a\) and \(u_b\) for some \(a<t<b\).
In \Cref{fig:overview-contraction-cover}(a), the block containing \(v_{Q_1}\) covers \(2\) and \(3\) at \(P\), while the block containing \(v_{Q_2}\) covers \(3\) and \(4\) at \(P\).
In \Cref{fig:overview-contraction-cover}(b), no block covers index \(3\), and \(u_3\) is a cut vertex.
More generally, an uncovered internal vertex of a path becomes a cut vertex in \(G^{\mathrm{pf}}+A\); hence every internal index of every path must be covered by some block of \(H_A\).
Together with bridgelessness of \(H_A\) and a small endpoint condition at each contracted path, this gives the condition for the uncontracted graph \(G^{\mathrm{pf}}+A\) to be \(2\)-vertex-connected; see \Cref{lem:vertex_cond_on_contracted_graph}.

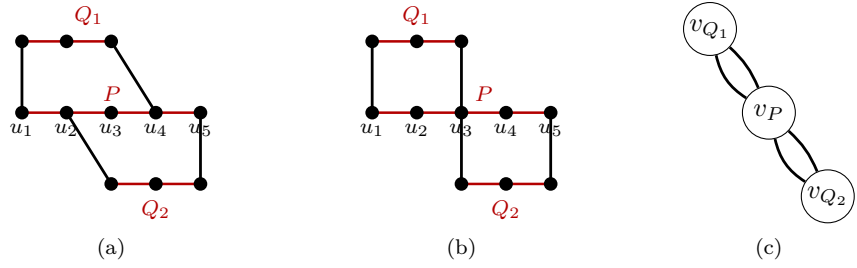
\begin{figure}[t]
\centering
\begin{tikzpicture}[
  x=0.82cm,y=0.82cm,
  vertex/.style={circle,fill=black,inner sep=0pt,minimum size=1.8mm},
  hvertex/.style={circle,draw=black,fill=white,inner sep=1pt,minimum size=7mm,font=\small},
  path edge/.style={line width=0.95pt,red!70!black},
  blue link/.style={line width=1.05pt,black},
  orange link/.style={line width=1.05pt,black},
  label/.style={font=\scriptsize,inner sep=1pt},
  path label/.style={font=\scriptsize,inner sep=1pt,text=red!70!black},
]
\begin{scope}[shift={(-6.9,0)}]
  \foreach \i/\x in {1/0,2/0.72,3/1.44,4/2.16,5/2.88}
    \coordinate (p\i) at (\x,0);
  \coordinate (q1) at (0,1.15);
  \coordinate (q2) at (0.72,1.15);
  \coordinate (q3) at (1.44,1.15);
  \coordinate (r1) at (1.44,-1.15);
  \coordinate (r2) at (2.16,-1.15);
  \coordinate (r3) at (2.88,-1.15);
  \draw[path edge] (p1)--(p2)--(p3)--(p4)--(p5);
  \draw[path edge] (q1)--(q2)--(q3);
  \draw[path edge] (r1)--(r2)--(r3);
  \draw[blue link] (p1)--(q1) (p4)--(q3);
  \draw[orange link] (p2)--(r1) (p5)--(r3);
  \foreach \i in {1,...,5} {
    \node[vertex] at (p\i) {};
    \node[label] at ($(p\i)+(0,-0.28)$) {\(u_{\i}\)};
  }
  \foreach \x in {q1,q2,q3,r1,r2,r3}
    \node[vertex] at (\x) {};
  \node[path label] at (1.08,1.56) {$Q_1$};
  \node[path label] at (1.45,0.30) {$P$};
  \node[path label] at (2.16,-1.56) {$Q_2$};  \node[label] at (1.44,-2.20) {(a)};
\end{scope}

\begin{scope}[shift={(-1.25,0)}]
  \foreach \i/\x in {1/0,2/0.72,3/1.44,4/2.16,5/2.88}
    \coordinate (p\i) at (\x,0);
  \coordinate (q1) at (0,1.15);
  \coordinate (q2) at (0.72,1.15);
  \coordinate (q3) at (1.44,1.15);
  \coordinate (r1) at (1.44,-1.15);
  \coordinate (r2) at (2.16,-1.15);
  \coordinate (r3) at (2.88,-1.15);
  \draw[path edge] (p1)--(p2)--(p3)--(p4)--(p5);
  \draw[path edge] (q1)--(q2)--(q3);
  \draw[path edge] (r1)--(r2)--(r3);
  \draw[blue link] (p1)--(q1) (p3)--(q3);
  \draw[orange link] (p3)--(r1) (p5)--(r3);
  \foreach \i in {1,...,5} {
    \node[vertex] at (p\i) {};
    \node[label] at ($(p\i)+(0,-0.28)$) {\(u_{\i}\)};
  }
  \foreach \x in {q1,q2,q3,r1,r2,r3}
    \node[vertex] at (\x) {};
  \node[path label] at (0.72,1.56) {$Q_1$};
  \node[path label] at (1.80,0.30) {$P$};
  \node[path label] at (2.16,-1.56) {$Q_2$};
  \node[label] at (1.44,-2.20) {(b)};
\end{scope}

\begin{scope}[shift={(5.15,0)}]
  \node[hvertex] (vp) at (0,0) {$v_P$};
  \node[hvertex] (vq1) at (-0.95,1.35) {$v_{Q_1}$};
  \node[hvertex] (vq2) at (0.95,-1.35) {$v_{Q_2}$};
  \draw[blue link] (vp) to[bend left=22] (vq1);
  \draw[blue link] (vp) to[bend right=12] (vq1);
  \draw[orange link] (vp) to[bend right=22] (vq2);
  \draw[orange link] (vp) to[bend left=12] (vq2);
  \node[label] at (0,-2.20) {(c)};
\end{scope}
\end{tikzpicture}
\caption{Contraction can forget \(2\)-vertex-connectivity. Both (a) and (b) contract to the same graph (c), but (a) is \(2\)-vertex-connected whereas, in (b), the third vertex of \(P\) is a cut vertex.}
\label{fig:overview-contraction-cover}
\end{figure}

\subsection{Dynamic programming on the contracted graph}\label{sec:overview-contracted-dp}

We now seek a minimum-cost edge set \(A\subseteq E(H)\) satisfying the characterization from \Cref{lem:vertex_cond_on_contracted_graph}.
The dynamic programming is indexed by pairs \((X,v_P)\) with \(X\subseteq V(H)\) and \(v_P\in X\), and follows the block-cut tree of the unknown graph \(H_A\), where the block-cut tree is a tree representing the containment relation between vertices and blocks; see \Cref{fig:overview-block-cut-tree}.

\begin{figure}[t]
\centering
\def\overviewblockcutx{1.05}
\def\overviewblockcutcoordinates{%
  \coordinate (p1)  at ({-sqrt(3)/2*\overviewblockcutx},{ 0.5*\overviewblockcutx});
  \coordinate (p2)  at ({-sqrt(3)/2*\overviewblockcutx},{-0.5*\overviewblockcutx});
  \coordinate (p3)  at (0,0);
  \coordinate (p4)  at (0,{-\overviewblockcutx});
  \coordinate (p5)  at ({\overviewblockcutx},0);
  \coordinate (p6)  at ({\overviewblockcutx+\overviewblockcutx/sqrt(2)},{ \overviewblockcutx/sqrt(2)});
  \coordinate (p7)  at ({\overviewblockcutx+\overviewblockcutx/sqrt(2)},{-\overviewblockcutx/sqrt(2)});
  \coordinate (p8)  at ({\overviewblockcutx+sqrt(2)*\overviewblockcutx},0);
  \coordinate (p9)  at ({\overviewblockcutx+sqrt(2)*\overviewblockcutx},{-\overviewblockcutx});
  \coordinate (p10) at ({\overviewblockcutx+sqrt(2)*\overviewblockcutx+\overviewblockcutx/sqrt(2)},{ \overviewblockcutx/sqrt(2)});
  \coordinate (p11) at ({\overviewblockcutx+sqrt(2)*\overviewblockcutx+\overviewblockcutx/sqrt(2)},{-\overviewblockcutx/sqrt(2)});
  \coordinate (p12) at ({\overviewblockcutx+2*sqrt(2)*\overviewblockcutx},0);
  \coordinate (p13) at ({\overviewblockcutx+2*sqrt(2)*\overviewblockcutx+cos(80)*\overviewblockcutx},{ sin(80)*\overviewblockcutx});
  \coordinate (p14) at ({\overviewblockcutx+2*sqrt(2)*\overviewblockcutx+cos(20)*\overviewblockcutx},{ sin(20)*\overviewblockcutx});
  \coordinate (p15) at ({\overviewblockcutx+2*sqrt(2)*\overviewblockcutx+cos(20)*\overviewblockcutx},{-sin(20)*\overviewblockcutx});
  \coordinate (p16) at ({\overviewblockcutx+2*sqrt(2)*\overviewblockcutx+cos(80)*\overviewblockcutx},{-sin(80)*\overviewblockcutx});
  \coordinate (p17) at ({\overviewblockcutx/sqrt(2)},{-\overviewblockcutx/sqrt(2)});
}
\begin{subfigure}[t]{0.48\linewidth}
\centering
\begin{tikzpicture}[
  baseline=(current bounding box.center),
  x=1cm,y=1cm,
  vertex/.style={circle,fill=black,inner sep=0pt,minimum size=1.8mm},
  edge/.style={line width=0.55pt},
  every label/.style={font=\scriptsize,inner sep=1pt}
]
\overviewblockcutcoordinates

\draw[edge]
  (p1) -- (p2)
  (p1) -- (p3)
  (p2) -- (p3)
  (p3) -- (p4)
  (p3) -- (p5)
  (p5) -- (p6)
  (p6) -- (p8)
  (p8) -- (p7)
  (p7) -- (p5)
  (p7) -- (p17)
  (p8) -- (p9)
  (p8) -- (p10)
  (p8) -- (p11)
  (p10) -- (p11)
  (p10) -- (p12)
  (p11) -- (p12)
  (p12) -- (p13)
  (p13) -- (p14)
  (p14) -- (p12)
  (p12) -- (p15)
  (p15) -- (p16)
  (p16) -- (p12);

\foreach \i/\lab/\pos in {
  1/$a$/above left,
  2/$b$/below left,
  3/$c$/above,
  4/$d$/below,
  5/$e$/above,
  17/$f$/below left,
  6/$g$/above,
  7/$h$/below,
  8/$i$/above,
  9/$j$/below,
  10/$k$/above,
  11/$\ell$/below,
  12/$m$/right,
  13/$n$/above,
  14/$o$/right,
  15/$p$/right,
  16/$q$/below}
  \node[vertex,label=\pos:{\lab}] at (p\i) {};
\end{tikzpicture}
\caption{}
\label{fig:overview-block-cut-tree-graph}
\end{subfigure}\hfill
\begin{subfigure}[t]{0.48\linewidth}
\centering
\begin{tikzpicture}[
  baseline=(current bounding box.center),
  x=1cm,y=1cm,
  vertex/.style={circle,fill=black,inner sep=0pt,minimum size=1.8mm},
  bag/.style={circle,draw=black,fill=white,line width=0.55pt,inner sep=0pt,minimum size=7mm,text width=6.5mm,font=\scriptsize,align=center},
  edge/.style={line width=0.55pt},
  every label/.style={font=\scriptsize,inner sep=1pt}
]
\overviewblockcutcoordinates

\node[bag] (b123)      at ({-sqrt(3)/3*\overviewblockcutx},0) {$abc$};
\node[bag] (b34)       at (0,{-0.5*\overviewblockcutx}) {$cd$};
\node[bag] (b35)       at ({0.5*\overviewblockcutx},0) {$ce$};
\node[bag] (bfh)       at ($(p17)!0.5!(p7)$) {$fh$};
\node[bag] (b5678)     at ({\overviewblockcutx+\overviewblockcutx/sqrt(2)},0) {$eghi$};
\node[bag] (b89)       at ({\overviewblockcutx+sqrt(2)*\overviewblockcutx},{-0.5*\overviewblockcutx}) {$ij$};
\node[bag] (b8101112)  at ({\overviewblockcutx+sqrt(2)*\overviewblockcutx+\overviewblockcutx/sqrt(2)},0) {$ik\ell m$};
\node[bag] (b121314)   at ({\overviewblockcutx+2*sqrt(2)*\overviewblockcutx+(cos(80)+cos(20))*\overviewblockcutx/3},{(sin(80)+sin(20))*\overviewblockcutx/3}) {$mno$};
\node[bag] (b121516)   at ({\overviewblockcutx+2*sqrt(2)*\overviewblockcutx+(cos(80)+cos(20))*\overviewblockcutx/3},{-(sin(80)+sin(20))*\overviewblockcutx/3}) {$mpq$};

\draw[edge]
  (b123) -- (p1)
  (b123) -- (p2)
  (b123) -- (p3)
  (b34) -- (p3)
  (b34) -- (p4)
  (b35) -- (p3)
  (b35) -- (p5)
  (bfh) -- (p17)
  (bfh) -- (p7)
  (b5678) -- (p5)
  (b5678) -- (p6)
  (b5678) -- (p7)
  (b5678) -- (p8)
  (b89) -- (p8)
  (b89) -- (p9)
  (b8101112) -- (p8)
  (b8101112) -- (p10)
  (b8101112) -- (p11)
  (b8101112) -- (p12)
  (b121314) -- (p12)
  (b121314) -- (p13)
  (b121314) -- (p14)
  (b121516) -- (p12)
  (b121516) -- (p15)
  (b121516) -- (p16);

\foreach \i/\lab/\pos in {
  1/$a$/above left,
  2/$b$/below left,
  3/$c$/above,
  4/$d$/below,
  5/$e$/above,
  17/$f$/below left,
  6/$g$/above,
  7/$h$/below,
  8/$i$/above,
  9/$j$/below,
  10/$k$/above,
  11/$\ell$/below,
  12/$m$/right,
  13/$n$/above,
  14/$o$/right,
  15/$p$/right,
  16/$q$/below}
  \node[vertex,label=\pos:{\lab}] at (p\i) {};
\end{tikzpicture}
\caption{}
\label{fig:overview-block-cut-tree-tree}
\end{subfigure}
\caption{A graph and its block-cut tree. Black vertices are original vertices, while white vertices are block vertices.}
\label{fig:overview-block-cut-tree}
\end{figure}
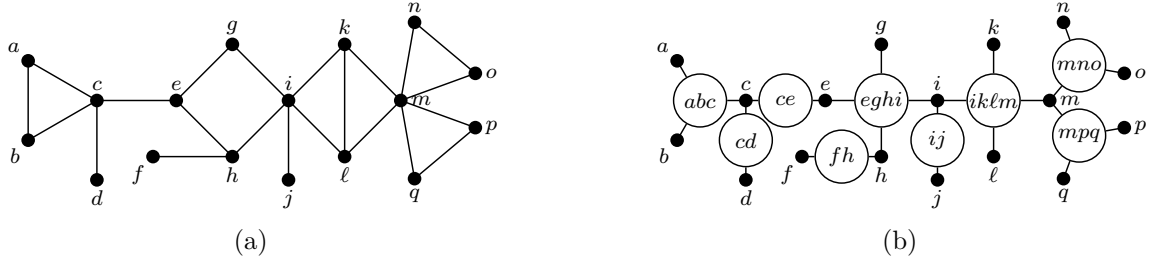

Fix a path \(P\) and a set \(X\subseteq V(H)\) with \(v_P\in X\).
For this state, we choose an edge set \(A_X\subseteq E(H[X])\) that represents the part of the solution living inside \(X\).
The graph \((X,A_X)\) may meet the rest of the solution only through \(v_P\).
Therefore the only unfinished information is which indices of \(P\) still have to be covered from outside \(X\).

We use two tables, denoted later by \(a_{X,v_P,i,j}\) and \(b_{X,v_P,i,j}\).
An \(a\)-entry fixes a unique block containing \(v_P\), and the interval \(\{i,\ldots,j\}\) is covered by this block.
A \(b\)-entry allows several blocks containing \(v_P\), and every internal index outside \(\{i,\ldots,j\}\) is already covered.
Thus, only indices in \(\{i,\ldots,j\}\) may still remain uncovered, and their covering is left to the remaining part of the solution.

The transition from \(a\) to \(b\) is straightforward: several \(a\)-pieces are glued at \(v_P\), and the interval of uncovered indices is updated.
This is formalized in \Cref{lem:vertex_b_from_a}.
The converse transition, from \(b\) to \(a\), is the harder part.
Assume that all relevant \(b\)-values on proper subsets of \(X\) are known.
To compute \(a_{X,v_P,i,j}\), we guess the vertex set \(Z\subseteq X\) of the unique block containing \(v_P\), attach already computed \(b\)-subproblems below vertices of \(Z\setminus\{v_P\}\), and optimize over the partition of \(X\setminus Z\).
At the root block, for each \(v_Q\in Z\), we choose the \emph{boundary pair} consisting of the two incident edges in the root block that witness the covering condition on \(Q\).
Informally, a boundary pair \(F\) at \(v_Q\in Z\) specifies where the root block enters and leaves the path \(Q\).
If the corresponding links touch \(Q\) at \(u_{Q,\ell}\) and \(u_{Q,r}\), with \(\ell<r\), then the root block covers \(\{\ell+1,\ldots,r-1\}\) on \(Q\).
At \(v_P\), this interval must contain \(\{i,\ldots,j\}\); for \(Q\neq P\), it specifies the attached \(b\)-entry, and the value of that entry is the cost of choosing the boundary pair.
Extracting this core part gives the following simplified auxiliary problem.\footnote{A closely related \(2\)-vertex-connected spanning subgraph subproblem also appears if the same dynamic programming strategy is applied to \(2\)-edge-connectivity augmentation: the contracted graph is then governed by a different characterization, but realizing a guessed block still requires enforcing \emph{vertex connectivity} inside the block.}

\boxproblemnoparam{$2$-Vertex-Connected Spanning Subgraph with Boundary-Pair Costs (overview form)}{
    A loopless multigraph \(G\), edge costs \(c_e\) for \(e\in E(G)\), and boundary-pair costs \(s_{v,F}\) for \(v\in V(G)\) and \(F\in \binom{\delta_G(v)}{2}\).
}{
    Select \(F_v\in \binom{\delta_G(v)}{2}\) for every \(v\in V(G)\), so that \((V(G),\bigcup_{v\in V(G)}F_v)\) is \(2\)-vertex-connected, while minimizing \(\sum_{e\in \bigcup_{v\in V(G)}F_v} c_e+\sum_{v\in V(G)}s_{v,F_v}\).
}

The formal version is slightly richer.
In the simplified problem, the selected edges are only those in \(\bigcup_{v\in V(G)}F_v\), while the actual root block may also use edges not selected by any boundary pair; the formal problem therefore adds an ordinary edge set \(A\) and charges the corresponding edge weights.
The formal problem also lets each vertex of \(Z\) choose an attachment set, so that the parts hanging below \(Z\) use disjoint subsets of \(X\setminus Z\); see \Cref{lem:vertex_a_from_b}.

\subsection{Solving the Boundary-Pair Subproblem}\label{sec:overview-boundary-dp}

The root-block transition requires a single-exponential algorithm for \textsc{2VCSS-BP}.
Without boundary-pair costs, adapting the ear-decomposition subset dynamic programming of Ameli et al.~\citep{JabalAmeliKNW26} gives an \(O^*(3^{|V(G)|})\)-time algorithm for ordinary \textsc{2VCSS}.
In that dynamic programming, a state is indexed only by a vertex subset \(X\subseteq V(G)\), and stores the best way to build a \(2\)-vertex-connected subgraph on \(X\).
The transition adds one ear whose endpoints are already present and whose internal vertices are new.
Unfortunately, the same idea does not directly extend to \textsc{2VCSS-BP}.
In an ear-based subset dynamic programming, a subproblem on \(X\subseteq V(G)\) would have to remember the boundary-pair information at vertices of \(X\) that may later become endpoints of new ears.
Storing all such boundary-pair data explicitly, however, would destroy the desired single-exponential state space.

\subsubsection{Algebraic Cancellation}\label{sec:overview-boundary-cancellation}

The auxiliary problem still contains a global \(2\)-vertex-connectivity condition.
Our goal is to replace this condition by a signed algebraic sum in which every candidate that is not \(2\)-vertex-connected cancels.
Recall Cut\&Count~\citep{CyganNPPVW11}.
It tests connectivity by counting consistent cuts modulo \(2\): with a fixed root, a graph with \(c\) connected components has \(2^{c-1}\) consistent cuts, so the parity is nonzero exactly when \(c=1\).
Over the integers, the same role is played by the M\"obius coefficients of the partition lattice.
For a graph \(K\), we have the identity
\[
    \sum_{\pi} (-1)^{|\pi|-1}(|\pi|-1)!
    =
    \begin{cases}
        1 & \text{if }K\text{ is connected},\\
        0 & \text{otherwise}.
    \end{cases}
\]
The sum is over the partitions of \(V(K)\) for which every edge has both endpoints in one part.
Thus partition-lattice inversion cancels all disconnected choices and keeps exactly the connected ones.

For \(2\)-vertex-connectivity, once connectedness has been enforced, the remaining obstruction is the presence of more than one block.
The object replacing a partition of connected components is a tree decomposition whose bags are glued by \(1\)-sums at original vertices.
More explicitly, for a connected graph \(K\), let \(\mathcal B_K(v)\) be the set of blocks containing \(v\).
A \emph{1-sum tree} \(T\) on \(V(K)\) is a bipartite tree whose nodes are the original vertices in \(V(K)\) and a family of bag nodes.
Each bag node has degree at least two and is identified with its neighborhood in \(V(K)\), called a bag.
We say that \(T\) is consistent with \(K\) if each bag is the vertex union of a nonempty collection of blocks whose union is connected in \(K\), and these collections partition the set of blocks of \(K\).
For a 1-sum tree \(T\), let \(d_T(v)\) be the number of bags containing \(v\), equivalently the degree of \(v\) in this incidence tree.
\Cref{fig:overview-one-sum-tree-examples} shows one example.

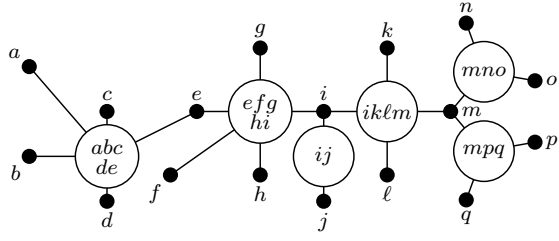
\begin{figure}[t]
\centering
\def\overviewonesumx{1.25}
\def\overviewonesumcoordinates{%
  \coordinate (p1)  at ({-sqrt(3)/2*\overviewonesumx},{ 0.5*\overviewonesumx});
  \coordinate (p2)  at ({-sqrt(3)/2*\overviewonesumx},{-0.5*\overviewonesumx});
  \coordinate (p3)  at (0,0);
  \coordinate (p4)  at (0,{-\overviewonesumx});
  \coordinate (p5)  at ({\overviewonesumx},0);
  \coordinate (p6)  at ({\overviewonesumx+\overviewonesumx/sqrt(2)},{ \overviewonesumx/sqrt(2)});
  \coordinate (p7)  at ({\overviewonesumx+\overviewonesumx/sqrt(2)},{-\overviewonesumx/sqrt(2)});
  \coordinate (p8)  at ({\overviewonesumx+sqrt(2)*\overviewonesumx},0);
  \coordinate (p9)  at ({\overviewonesumx+sqrt(2)*\overviewonesumx},{-\overviewonesumx});
  \coordinate (p10) at ({\overviewonesumx+sqrt(2)*\overviewonesumx+\overviewonesumx/sqrt(2)},{ \overviewonesumx/sqrt(2)});
  \coordinate (p11) at ({\overviewonesumx+sqrt(2)*\overviewonesumx+\overviewonesumx/sqrt(2)},{-\overviewonesumx/sqrt(2)});
  \coordinate (p12) at ({\overviewonesumx+2*sqrt(2)*\overviewonesumx},0);
  \coordinate (p13) at ({\overviewonesumx+2*sqrt(2)*\overviewonesumx+cos(80)*\overviewonesumx},{ sin(80)*\overviewonesumx});
  \coordinate (p14) at ({\overviewonesumx+2*sqrt(2)*\overviewonesumx+cos(20)*\overviewonesumx},{ sin(20)*\overviewonesumx});
  \coordinate (p15) at ({\overviewonesumx+2*sqrt(2)*\overviewonesumx+cos(20)*\overviewonesumx},{-sin(20)*\overviewonesumx});
  \coordinate (p16) at ({\overviewonesumx+2*sqrt(2)*\overviewonesumx+cos(80)*\overviewonesumx},{-sin(80)*\overviewonesumx});
  \coordinate (p17) at ({\overviewonesumx/sqrt(2)},{-\overviewonesumx/sqrt(2)});
}
\def\overviewonesumverticesleft{%
  \foreach \i/\lab/\pos in {
    1/$a$/above left,
    2/$b$/below left,
    3/$c$/above,
    4/$d$/below,
    5/$e$/above,
    17/$f$/below left,
    6/$g$/above,
    7/$h$/below,
    8/$i$/above,
    9/$j$/below,
    10/$k$/above,
    11/$\ell$/below,
    12/$m$/right,
    13/$n$/above,
    14/$o$/right,
    15/$p$/right,
    16/$q$/below}
    \node[vertex,label=\pos:{\lab}] at (p\i) {};
}
\begin{tikzpicture}[
  x=0.95cm,y=0.95cm,
  vertex/.style={circle,fill=black,inner sep=0pt,minimum size=2mm},
  bag/.style={circle,draw=black,fill=white,line width=0.55pt,inner sep=0pt,minimum size=8mm,text width=7mm,font=\scriptsize,align=center},
  edge/.style={line width=0.55pt},
  every label/.style={font=\scriptsize,inner sep=1pt}
]
\overviewonesumcoordinates
\coordinate (A) at ($(p3)!0.5!(p4)$);
\coordinate (B) at ($(p5)!0.5!(p8)$);
\coordinate (C) at ($(p8)!0.5!(p9)$);
\coordinate (D) at ($(p8)!0.5!(p12)$);
\coordinate (E) at ({\overviewonesumx+2*sqrt(2)*\overviewonesumx+(cos(80)+cos(20))*\overviewonesumx/3},{(sin(80)+sin(20))*\overviewonesumx/3});
\coordinate (F) at ({\overviewonesumx+2*sqrt(2)*\overviewonesumx+(cos(80)+cos(20))*\overviewonesumx/3},{-(sin(80)+sin(20))*\overviewonesumx/3});

\draw[edge]
  (A)--(p1) (A)--(p2) (A)--(p3) (A)--(p4) (A)--(p5)
  (B)--(p5) (B)--(p17) (B)--(p6) (B)--(p7) (B)--(p8)
  (C)--(p8) (C)--(p9)
  (D)--(p8) (D)--(p10) (D)--(p11) (D)--(p12)
  (E)--(p12) (E)--(p13) (E)--(p14)
  (F)--(p12) (F)--(p15) (F)--(p16);

\node[bag] at (A) {$abc$\\[-2pt]$de$};
\node[bag] at (B) {$efg$\\[-2pt]$hi$};
\node[bag] at (C) {$ij$};
\node[bag] at (D) {$ik\ell m$};
\node[bag] at (E) {$mno$};
\node[bag] at (F) {$mpq$};
\overviewonesumverticesleft
\end{tikzpicture}
\caption{A 1-sum tree consistent with the graph in \Cref{fig:overview-block-cut-tree}(a). Black nodes are original vertices, while white nodes are bags; the labels inside white nodes are the corresponding bags.}
\label{fig:overview-one-sum-tree-examples}
\end{figure}

\begin{lemma}\label{lem:overview-bag-tree-mobius}
Let \(K\) be a connected graph with no self-loops, and let \(\mathcal T(K)\) be the family of 1-sum trees consistent with \(K\).
Then,
\[
    \sum_{T\in\mathcal T(K)}
    \prod_{v\in V(K)} (-1)^{d_T(v)-1}(d_T(v)-1)!
    =
    \begin{cases}
        1, & \text{if }K\text{ is }2\text{-vertex-connected},\\
        0, & \text{otherwise.}
    \end{cases}
\]
\end{lemma}
\begin{proof}[Proof sketch]
We construct a bijection between 1-sum trees consistent with \(K\) and the following data on the block-cut tree of \(K\): for each original vertex \(v\), a partition \(\pi_v\) of \(\mathcal B_K(v)\), the set of blocks containing \(v\).
Start with the block-cut tree of \(K\).
For every part of \(\pi_v\), identify the corresponding block nodes adjacent to \(v\), and delete the parallel incidences created at \(v\).
Doing this for all vertices gives a 1-sum tree consistent with \(K\).
Conversely, a 1-sum tree consistent with \(K\) recovers \(\pi_v\) by grouping the blocks of \(\mathcal B_K(v)\) that lie in the same bag incident with \(v\).

Thus the sum over 1-sum trees consistent with \(K\) is the product, over \(v\in V(K)\), of the partition-lattice M\"obius sums on \(\mathcal B_K(v)\).
Indeed, \(d_T(v)=|\pi_v|\), so the factor at \(v\) is exactly \((-1)^{|\pi_v|-1}(|\pi_v|-1)!\).
This vertex-wise M\"obius sum is \(1\) when \(|\mathcal B_K(v)|=1\), and \(0\) otherwise.
Thus the product is nonzero exactly when every vertex of \(K\) belongs to one block.
\end{proof}

Let \(Q(\zeta)\) be the objective polynomial for the overview form:
\[
    Q(\zeta):=
    \sum_{\substack{\mathbf F=(F_v)_{v\in V(G)}\\
        (V(G),\bigcup_{v\in V(G)}F_v)\text{ is }2\text{-vertex-connected}}}
    \zeta^{\operatorname{cost}(\mathbf F)},
\]
where \(\operatorname{cost}(\mathbf F):=\sum_{e\in\bigcup_vF_v}c_e+\sum_v s_{v,F_v}\).
The least exponent with a nonzero coefficient is the optimum cost, and interpolation recovers the coefficients from evaluations of \(Q\).

After connectedness has been imposed, \Cref{lem:overview-bag-tree-mobius} lets us evaluate \(Q(\zeta)\) by summing over pairs consisting of a boundary-pair choice \(\mathbf F\) and a 1-sum tree \(T\) consistent with \(K_{\mathbf F}:=(V(G),\bigcup_{v\in V(G)}F_v)\).
Such a pair contributes the monomial \(\zeta^{\operatorname{cost}(\mathbf F)}\) with coefficient \(\prod_{v\in V(G)}(-1)^{d_T(v)-1}(d_T(v)-1)!\).
The lemma cancels all boundary-pair choices whose graph \(K_{\mathbf F}\) is not \(2\)-vertex-connected.

\subsubsection{Dynamic Programming}

There may be \(|V(G)|^{O(|V(G)|)}\) 1-sum trees, so they cannot be enumerated within the desired running time.
Instead, as in \Cref{sec:overview-contracted-dp}, we apply dynamic programming indexed by subsets of \(V(G)\), following the 1-sum tree of the unknown graph selected by the solution.
Fix a vertex \(u\) and a set \(X\subseteq V(G)\) with \(u\in X\).
The important transition guesses the top bag \(Z\subseteq X\) containing \(u\).
Thus the remaining task is to evaluate the contribution of one fixed top bag.

For a fixed top bag \(Z\), we choose a pair \(F_v\) of incident edges at each \(v\in Z\).
The difficulty is that these choices are not independent: the same edge may be chosen at both endpoints, but its edge cost \(c_e\) must be charged only once.

The decisive observation is that the dependencies have a very restricted shape.
Since every \(F_v\) has size two, each vertex contributes exactly two selected incidences.
Start with one length-two walk for each \(v\in Z\), using the two edges of \(F_v\) as the two end edges incident with \(v\).
The only interaction between two local choices occurs when an edge \(e=vw\) lies in both \(F_v\) and \(F_w\); then we glue the two current walks along the two occurrences of \(e\), and if both occurrences already lie in the same walk, this operation closes a cycle.
After all such gluing operations, every inner vertex has exactly two selected incidences, so the resulting objects are only walks and cycles.
No branching object can appear.
We call these walks and cycles pseudo-ears; the formal definition and an example appear in \Cref{sec:2VCSSBP}.\footnote{Pseudo-ears are used only for this grouping argument; they are unrelated to ears in an ear decomposition.}

For each \(I\subseteq Z\), the contribution of one pseudo-ear with inner-vertex set \(I\) can be computed by a Held--Karp-style dynamic programming.
The table records the set of already used inner vertices and the current end of the partial walk.
Cycle pseudo-ears are handled by the analogous computation after fixing one canonical starting point and direction.

For each partition of \(Z\), we multiply the single-pseudo-ear contributions indexed by its parts, and then sum these products over all partitions of \(Z\).
This partition sum is computed by standard subset dynamic programming in single-exponential time.
The partition-lattice M\"obius inversion explained at the beginning of \Cref{sec:overview-boundary-cancellation} then imposes ordinary connectedness on \(Z\), again in single-exponential time.

\subsection{Runtime Analysis}

Let \(p\) be the number of paths in the \textsc{PFA} instance, so \(|V(H)|=p\).
We can observe that all computations described above can be done in \(O^*(2^{O(p)})\) time.
Since the reduction gives \(p\leq 2k\), this already yields a single-exponential algorithm for \textsc{2VCA}.

The formal algorithm uses several speed-ups based on subset convolution.
In particular, for a fixed top bag \(Z\), the transition to all states \(X\supseteq Z\) can be computed together in \(O^*(2^{|V(H)|})\) time.
This batching is used both in the dynamic programming from \Cref{sec:overview-contracted-dp} and in the dynamic programming from \Cref{sec:overview-boundary-dp}, and subset convolution also appears in the other partitioning steps.
A precise analysis gives an \(O^*(6^p)\)-time algorithm for unweighted \textsc{PFA}.
Together with \(p\leq 2k\), this gives \(O^*(36^k)\) time for unweighted \textsc{2VCA}.
For weighted instances, the same analysis adds a linear factor \(W\), where \(W\) is the maximum link cost.

\section{Preliminaries}\label{sec:preliminaries}

\subsection{Graph-Theoretic Notions}

All graphs are undirected, not necessarily connected, and may have parallel edges and self-loops.
We also use the name \emph{multigraph} to emphasize that multiple edges may be present.
For a graph $G$, we write $V(G)$ and $E(G)$ for its vertex and edge sets, $G[X]$ for the subgraph induced by $X\subseteq V(G)$, and $\delta_G(X)$ for the set of edges with exactly one endpoint in $X$.
For a vertex $v$, we abbreviate $\delta_G(\{v\})$ to $\delta_G(v)$.
For a walk $P$, let $V_{\mathrm{in}}(P)$ be the set of non-endpoint vertices of $P$.
For a cycle $C$, we use the convention $V_{\mathrm{in}}(C):=V(C)$.
For an integer $\lambda\geq 1$, a graph $G$ is \emph{$\lambda$-edge-connected} if $G-F$ is connected for every $F\subseteq E(G)$ with $|F|<\lambda$.
A graph $G$ is \emph{$\lambda$-vertex-connected} if $G-X$ is connected for every $X\subseteq V(G)$ with $|X|<\lambda$.
In particular, we regard $K_\lambda$ as $\lambda$-vertex-connected.

An edge $e$ of $G$ is a \emph{bridge} if $G-e$ has more connected components than $G$.
A graph is \emph{bridgeless} if it is connected and has no bridge; this convention includes the one-vertex graph.
For graphs with at least two vertices, bridgelessness, $2$-edge-connectivity, and every edge lying on a cycle are equivalent.
A vertex $v$ of $G$ is a \emph{cut vertex} if $G-v$ has more connected components than $G$.
A graph with at least three vertices is $2$-vertex-connected if and only if it is connected and has no cut vertex.
For blocks, we use the following convention, which also covers multigraphs with self-loops.
A \emph{non-self-loop block} is a maximal subgraph $B$ with no self-loop and with at least two vertices such that $B$ is $2$-vertex-connected; in particular, $K_2$ is a block.
A single self-loop is a block by itself, called a \emph{self-loop block}.
We call both non-self-loop blocks and self-loop blocks simply \emph{blocks}.
For graphs without self-loops, this agrees with the usual block notion with $K_2$ counted as a block.
For a connected graph \(G\), let \(\mathcal B_G\) be the family of vertex sets of the blocks of \(G\).
The \emph{block-cut tree} of \(G\) is the bipartite graph with node set \(V(G)\cup\mathcal B_G\), where the elements of \(V(G)\) are vertex nodes and the elements of \(\mathcal B_G\) are block nodes.
A vertex node \(v\in V(G)\) is adjacent to a block node \(B\in\mathcal B_G\) exactly when \(v\in B\).
We use the standard fact that this incidence graph is a tree; see \Cref{fig:overview-block-cut-tree}.

\subsection{Reducing Edge-Connectivity Augmentation to Vertex-Connectivity Augmentation}

We focus on vertex-connectivity augmentation because the corresponding edge-connectivity augmentation problem reduces to it while preserving the parameter $k$.
The standard reduction is the following.
For target $\lambda$, replace each vertex $v$ by a clique $K_v$ containing $\lambda$ dummy vertices and one vertex $p_{v,e}$ for every original edge or link $e$ incident to $v$.
For every original edge $e=uv$, add the fixed edge $p_{u,e}p_{v,e}$; for every link $\ell=uv$, add the link $p_{u,\ell}p_{v,\ell}$ with the same cost.
Then, for any selected link set $A$, the graph $G+A$ is $\lambda$-edge-connected if and only if the graph constructed above, with the links corresponding to $A$ added, is $\lambda$-vertex-connected.
Moreover, for $\lambda\geq 2$, this reduction preserves the augment-by-one assumption: if the original graph is already $(\lambda-1)$-edge-connected, then the constructed graph is already $(\lambda-1)$-vertex-connected.

\subsection{Algebraic Tools}

We introduce three algebraic tools: interpolation, subset transforms and convolution, and M\"obius inversion.

\paragraph*{Interpolation of polynomials.}

\begin{lemma}[Interpolation {\cite[Chapter~5]{GathenGerhard2013}}]\label{lem:interpolation}
Let $p(z)\in \mathbb{F}[z]$ be a univariate polynomial of degree at most $d$.
Given values $p(z_i)=p_i$ at distinct points $z_0,\dots,z_d\in\mathbb{F}$, the polynomial is
\[
    p(z)=\sum_{i=0}^{d} p_i
    \prod_{\substack{0\le j\le d\\j\neq i}}\frac{z-z_j}{z_i-z_j}.
\]
Thus every coefficient of $p(z)$ can be computed from these $d+1$ evaluations in polynomial time.
\end{lemma}

\paragraph*{Subset transforms and convolution.}
\begin{lemma}[Fast zeta and M\"obius transforms {\cite{DBLP:conf/stoc/BjorklundHKK07}}]\label{lem:fast-zeta-mobius}
Let $U$ be a finite set and let $R$ be a commutative ring.
Given values $f(T)\in R$ for all $T\subseteq U$, all zeta-transform values
\(\sum_{T\subseteq S}f(T)\) and all M\"obius-transform values
\(\sum_{T\subseteq S}(-1)^{|S|-|T|}f(T)\), over all \(S\subseteq U\), can be computed
in $O^*(2^{|U|})$ arithmetic operations in $R$.
\end{lemma}

\begin{lemma}[Subset convolution {\cite{DBLP:conf/stoc/BjorklundHKK07}}]\label{lem:subset-convolution}
Let $U$ be a finite set.
For two functions $f,g$ on $2^U$ over a ring, their subset convolution is the function $f*g$ defined by
\[
    (f*g)(X):=\sum_{Y\subseteq X} f(Y)g(X\setminus Y).
\]
All values of $f*g$ over $2^U$ can be computed in $O^*(2^{|U|})$ arithmetic operations.
\end{lemma}

\begin{lemma}[Min-plus subset convolution {\cite{DBLP:conf/stoc/BjorklundHKK07}}]\label{lem:min-plus-subset-convolution}
Let $U$ be a finite set and let $\overline{W}$ be an integer.
For two functions $f,g\colon 2^U\to \{0,\dots,\overline{W}\}\cup\{\infty\}$, their min-plus subset convolution is the function $f*_{\min}g$ defined by
\[
    (f*_{\min}g)(X):=\min_{Y\subseteq X}\bigl(f(Y)+g(X\setminus Y)\bigr).
\]
Then all values of $f*_{\min}g$ can be computed in $O^*(2^{|U|}\overline{W})$ arithmetic operations.
\end{lemma}

For a finite set $U$, let $\Pi_i(U)$ denote the set of unordered partitions of $U$ into $i$ nonempty parts.
Let $\Pi(U):=\bigcup_i\Pi_i(U)$.
We use the convention $\Pi_0(\emptyset)=\{\emptyset\}$ and $\Pi_0(U)=\emptyset$ for $U\neq\emptyset$.

\begin{lemma}\label{lem:app-partition-convolution}
Let $R$ be a commutative ring.
Suppose that a value $x_C\in R$ is given for every $C\subseteq U$, with $x_\emptyset=0$.
Then, for every $i\in\{0,\dots,|U|\}$, the values of $p_i(S):=\sum_{\{C_1,\dots,C_i\}\in\Pi_i(S)}\prod_{j=1}^{i}x_{C_j}$ for all $S\subseteq U$ can be computed together in $O^*(2^{|U|})$ arithmetic operations in $R$.
\end{lemma}
\begin{proof}
Fix a linear order on $U$.
Set $p_0(\emptyset)=1$ and $p_0(S)=0$ for $S\neq\emptyset$.
Assume that all values $p_{i-1}(S)$ are known.
For $a\in U$, let $U_a:=\{b\in U:a<b\}$, and define functions on $2^{U_a}$ by $f_a(T):=x_{T\cup\{a\}}$ and $g_a(T):=p_{i-1}(T)$.
Compute $h_a:=f_a*g_a$ over the universe $U_a$ using \Cref{lem:subset-convolution}.
For every nonempty $S\subseteq U$ with minimum element $a$, set $p_i(S):=h_a(S\setminus\{a\})$.
This recurrence chooses the unique part of a partition that contains the minimum element of $S$.
After this part is chosen, the remaining $i-1$ parts form a partition of the remaining vertices.
Thus each unordered partition of $S$ into $i$ nonempty parts is counted exactly once.
For fixed $i$, the total time is $\sum_{a\in U}O^*(2^{|U_a|})=O^*(2^{|U|})$.
Iterating over all $i$ only changes the suppressed polynomial factor.
\end{proof}

\paragraph*{M\"obius inversion.}
Let $(\mathcal L,\preceq)$ be a finite lattice with maximum element $\hat{1}$.
For $x\preceq y$, the \emph{M\"obius function} $\mu_{\mathcal L}$ is defined by $\mu_{\mathcal L}(x,x)=1$ and $\mu_{\mathcal L}(x,y)=-\sum_{x\preceq z\prec y}\mu_{\mathcal L}(x,z)$ for $x\prec y$, where $x\prec y$ means $x\preceq y$ and $x\neq y$.

\begin{lemma}[M\"obius inversion {\cite[Proposition~3.7.1]{Stanley2011}}]\label{lem:mobius-inversion-top}
Let $R$ be a commutative ring.
Let $a_x,b_x\in R$ be indexed by $x\in\mathcal L$, and suppose that $b_y=\sum_{x\preceq y}a_x$ for all $y\in\mathcal L$.
Then $a_y=\sum_{x\preceq y}\mu_{\mathcal L}(x,y)b_x$ for every $y\in\mathcal L$.
Such $a_y$ is called the \emph{M\"obius inversion} of $b$ at $y$.
\end{lemma}

In this paper, we use two specific lattices: subset lattices and partition lattices.
For a finite set $S$, let $2^S$ be the lattice whose elements are the subsets of $S$, where $T\preceq T'$ means $T\subseteq T'$.
For a nonempty finite set $S$, let $\Pi(S)$ be the lattice whose elements are the partitions of $S$, where $\pi\preceq \pi'$ means that every part of $\pi$ is contained in a part of $\pi'$.
We record their M\"obius coefficients below.

\begin{lemma}[Subset lattice {\cite[Example~3.8.3]{Stanley2011}}]\label{lem:app-all-subuniverse-ie}
Let $U$ be a finite set.
For $S\subseteq U$, the M\"obius function of $2^S$ satisfies $\mu_{2^S}(T,S)=(-1)^{|S|-|T|}$ for every $T\subseteq S$.
Moreover, given values $x_T$ in a commutative ring for all $T\subseteq U$, the values $\sum_{T\subseteq S}(-1)^{|S|-|T|}x_T$ can be computed in $O^*(2^{|U|})$ arithmetic operations for all $S\subseteq U$.
\end{lemma}

\begin{lemma}[Partition lattice {\cite[Example~3.10.4]{Stanley2011}}]\label{lem:partition-lattice-mobius}
Let $U$ be a finite set.
For nonempty $S\subseteq U$, the M\"obius function of $\Pi(S)$ satisfies $\mu_{\Pi(S)}(\pi,\{S\})=(-1)^{|\pi|-1}(|\pi|-1)!$ for every $\pi\in\Pi(S)$.
Moreover, given values $x_C$ in a commutative ring for all nonempty subsets $C\subseteq U$, the values $\sum_{\pi\in\Pi(S)}(-1)^{|\pi|-1}(|\pi|-1)!\prod_{C\in\pi}x_C$ can be computed in $O^*(2^{|U|})$ arithmetic operations for all nonempty $S\subseteq U$.
\end{lemma}
\begin{proof}
The formula for the M\"obius coefficient is the cited one.
For the computation, set $x_\emptyset=0$.
For each $i$, \Cref{lem:app-partition-convolution} computes all unordered partition sums $p_i(S):=\sum_{\pi\in\Pi_i(S)}\prod_{C\in\pi}x_C$ in $O^*(2^{|U|})$ time.
For each nonempty $S\subseteq U$, the required M\"obius inversion is $\sum_i(-1)^{i-1}(i-1)!p_i(S)$.
Summing over all $i$ adds only a polynomial factor.
\end{proof}

\section{From \textsc{2VCA} to \textsc{PFA}}

The first step is to reduce general \textsc{2VCA} instances to the following path-forest form.

\boxproblem{Path Forest Augmentation (PFA)}
{A path forest $G=(V(G),E(G))$ that is a union of at most $p$ vertex-disjoint paths, a set of links $L$ on $V(G)$ whose endpoint sets belong to $\binom{V(G)}{2}$, an integer $k$, link costs $c\colon L\to \mathbb Z_{\geq 0}$, link counts $\chi\colon L\to\{0,1\}$, and an integer $W\geq 0$ such that $c_e\leq W$ for every $e\in L$ and $c_e=0$ whenever $\chi_e=0$.}
{Find a minimum-cost set $A\subseteq L$ such that $\chi(A)\leq k$ and the augmented multigraph $(V(G),E(G)\cup A)$ is $2$-vertex-connected.}
{$p$}

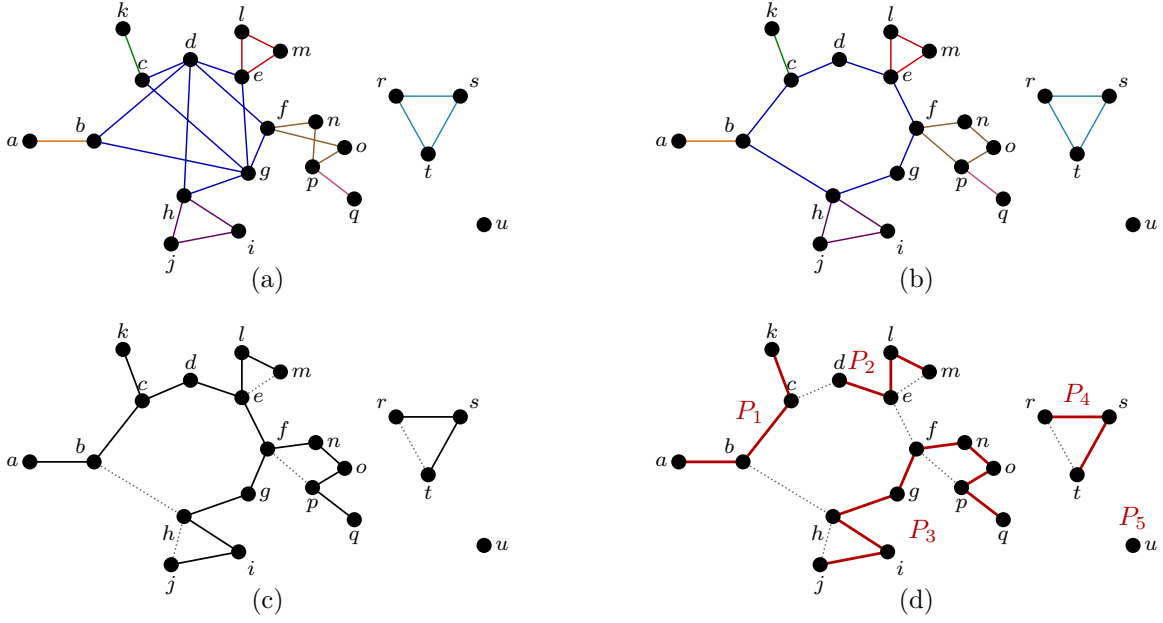
\begin{figure}[t]
\centering
\def\twovcacoordinates{%
  \coordinate (va) at (-3.35, 0.05);
  \coordinate (vb) at (-2.35, 0.05);
  \coordinate (vl) at (-3.05, 0.65);
  \coordinate (vc) at (-1.60, 1.00);
  \coordinate (vm) at (-0.85, 1.32);
  \coordinate (vd) at (-0.05, 1.05);
  \coordinate (ve) at ( 0.35, 0.25);
  \coordinate (vn) at ( 0.05,-0.45);
  \coordinate (vf) at (-0.95,-0.80);
  \coordinate (vg) at (-0.10,-1.35);
  \coordinate (vh) at (-1.15,-1.55);
  \coordinate (vp) at (-1.90, 1.80);
  \coordinate (vo) at (-2.35, 1.35);
  \coordinate (vq) at (-0.05, 1.75);
  \coordinate (vt) at ( 0.55, 1.45);
  \coordinate (vr) at ( 1.10, 0.35);
  \coordinate (vw) at ( 1.55,-0.05);
  \coordinate (vu) at ( 1.05,-0.35);
  \coordinate (vx) at ( 1.70,-0.85);
  \coordinate (vi) at ( 2.35, 0.75);
  \coordinate (vj) at ( 3.35, 0.75);
  \coordinate (vk) at ( 2.85,-0.15);
  \coordinate (vs) at ( 3.72,-1.25);
}
\def\twovcavertices{%
  \foreach \name/\lab/\pos in {
    a/$a$/left,
    b/$b$/above left,
    c/$c$/above,
    m/$d$/above,
    d/$e$/right,
    e/$f$/above right,
    n/$g$/right,
    f/$h$/below left,
    g/$i$/below right,
    h/$j$/below,
    p/$k$/above,
    q/$l$/above,
    t/$m$/right,
    r/$n$/right,
    w/$o$/right,
    u/$p$/below,
    x/$q$/below,
    i/$r$/above left,
    j/$s$/above right,
    k/$t$/below,
    s/$u$/right}
    \node[vertex,label=\pos:{\lab}] at (v\name) {};
}
\begin{subfigure}[t]{0.48\linewidth}
\centering
\begin{tikzpicture}[
  x=0.85cm,y=0.85cm,
  vertex/.style={circle,fill=black,inner sep=0pt,minimum size=2mm},
  edge/.style={line width=0.55pt},
  block label/.style={font=\scriptsize,blue!60!black,inner sep=1pt},
  every label/.style={font=\scriptsize,inner sep=1pt}
]
\twovcacoordinates
\draw[edge,orange!85!black]
  (va) -- (vb);
\draw[edge,blue!70!black]
  (vb) -- (vm)
  (vc) -- (vm)
  (vc) -- (vn)
  (vd) -- (vn)
  (vd) -- (vm)
  (vm) -- (ve)
  (ve) -- (vn)
  (vf) -- (vn)
  (vf) -- (vm)
  (vb) -- (vn);
\draw[edge,green!50!black]
  (vc) -- (vp);
\draw[edge,red!75!black]
  (vd) -- (vq)
  (vq) -- (vt)
  (vt) -- (vd);
\draw[edge,brown!75!black]
  (ve) -- (vw)
  (vw) -- (vu)
  (vu) -- (vr)
  (vr) -- (ve);
\draw[edge,magenta!75!black]
  (vu) -- (vx);
\draw[edge,violet!75!black]
  (vf) -- (vg)
  (vg) -- (vh)
  (vh) -- (vf);
\draw[edge,cyan!60!black]
  (vi) -- (vj)
  (vj) -- (vk)
  (vk) -- (vi);
\twovcavertices
\node[font=\small] at (0.35,-2.12) {(a)};
\end{tikzpicture}
\end{subfigure}\hfill
\begin{subfigure}[t]{0.48\linewidth}
\centering
\begin{tikzpicture}[
  x=0.85cm,y=0.85cm,
  vertex/.style={circle,fill=black,inner sep=0pt,minimum size=2mm},
  edge/.style={line width=0.55pt},
  every label/.style={font=\scriptsize,inner sep=1pt}
]
\twovcacoordinates
\draw[edge,orange!85!black]
  (va) -- (vb);
\draw[edge,blue!70!black]
  (vb) -- (vf)
  (vf) -- (vn)
  (vn) -- (ve)
  (vd) -- (ve)
  (vd) -- (vm)
  (vm) -- (vc)
  (vc) -- (vb);
\draw[edge,green!50!black]
  (vc) -- (vp);
\draw[edge,red!75!black]
  (vd) -- (vq)
  (vq) -- (vt)
  (vt) -- (vd);
\draw[edge,brown!75!black]
  (ve) -- (vu)
  (vu) -- (vw)
  (vw) -- (vr)
  (vr) -- (ve);
\draw[edge,magenta!75!black]
  (vu) -- (vx);
\draw[edge,violet!75!black]
  (vf) -- (vg)
  (vg) -- (vh)
  (vh) -- (vf);
\draw[edge,cyan!60!black]
  (vi) -- (vj)
  (vj) -- (vk)
  (vk) -- (vi);
\twovcavertices
\node[font=\small] at (0.35,-2.12) {(b)};
\end{tikzpicture}
\end{subfigure}

\vspace{0.5em}
\begin{subfigure}[t]{0.48\linewidth}
\centering
\begin{tikzpicture}[
  x=0.85cm,y=0.85cm,
  vertex/.style={circle,fill=black,inner sep=0pt,minimum size=2mm},
  forest edge/.style={line width=0.65pt},
  omitted edge/.style={line width=0.55pt,densely dotted,gray!75!black},
  every label/.style={font=\scriptsize,inner sep=1pt}
]
\twovcacoordinates
\draw[forest edge]
  (va) -- (vb)
  (vb) -- (vc)
  (vc) -- (vp)
  (vc) -- (vm)
  (vm) -- (vd)
  (vd) -- (vq)
  (vq) -- (vt)
  (vd) -- (ve)
  (ve) -- (vr)
  (vr) -- (vw)
  (vu) -- (vx)
  (vu) -- (vw)
  (ve) -- (vn)
  (vn) -- (vf)
  (vf) -- (vg)
  (vg) -- (vh)
  (vi) -- (vj)
  (vj) -- (vk);
\draw[omitted edge]
  (vt) -- (vd)
  (vb) -- (vf)
  (ve) -- (vu)
  (vf) -- (vh)
  (vk) -- (vi);
\twovcavertices
\node[font=\small] at (0.35,-2.12) {(c)};
\end{tikzpicture}
\end{subfigure}\hfill
\begin{subfigure}[t]{0.48\linewidth}
\centering
\begin{tikzpicture}[
  x=0.85cm,y=0.85cm,
  vertex/.style={circle,fill=black,inner sep=0pt,minimum size=2mm},
  path edge/.style={line width=1.05pt,red!70!black},
  omitted edge/.style={line width=0.55pt,densely dotted,gray!75!black},
  path label/.style={font=\small,red!70!black},
  every label/.style={font=\scriptsize,inner sep=1pt}
]
\twovcacoordinates
\draw[path edge]
  (va) -- (vb)
  (vp) -- (vc)
  (vc) -- (vb)
  (vt) -- (vq)
  (vq) -- (vd)
  (vd) -- (vm)
  (ve) -- (vr)
  (vr) -- (vw)
  (vu) -- (vw)
  (vu) -- (vx)
  (ve) -- (vn)
  (vn) -- (vf)
  (vf) -- (vg)
  (vg) -- (vh)
  (vi) -- (vj)
  (vj) -- (vk);
\draw[omitted edge]
  (vt) -- (vd)
  (vb) -- (vf)
  (vc) -- (vm)
  (vd) -- (ve)
  (ve) -- (vu)
  (vf) -- (vh)
  (vk) -- (vi);
\node[path label] at (-2.25,0.82) {$P_1$};
\node[path label] at (-0.50,1.60) {$P_2$};
\node[path label] at (0.45,-1.02) {$P_3$};
\node[path label] at (2.85,1.10) {$P_4$};
\node[path label] at (3.72,-0.82) {$P_5$};
\twovcavertices
\node[font=\small] at (0.35,-2.12) {(d)};
\end{tikzpicture}
\end{subfigure}
\caption{Reduction from \textsc{2VCA} to \textsc{PFA}. (a) The input graph $G$; edge colors indicate its blocks. (b) The graph $G'$ obtained by replacing each block with at least three vertices by a cycle, chosen so that specified cut vertices become consecutive; for example, the cycle in (b) replacing the blue block on $\{b,c,d,e,f,g,h\}$ makes $b$ and $h$ consecutive, and the cycle replacing the brown block on $\{f,n,o,p\}$ makes $f$ and $p$ consecutive. (c) A spanning forest $S$ of $G'$ whose leaves, in each component with at least two blocks, lie in the parts corresponding to leaf blocks. (d) The red edges form the path forest $F$; dotted edges are omitted edges of $G'$ and are added to the \textsc{PFA} instance as zero-cost links.}
\label{fig:twovca-to-twovca-pf-reduction}
\end{figure}

The following lemma reduces \textsc{2VCA} to \textsc{PFA} on at most $2k$ paths.

\begin{lemma}\label{lem:twovca_to_twovca_pf}
Given an instance $(G,L,k,W,c)$ of \textsc{2VCA}, one can in polynomial time either solve the instance directly, correctly conclude that no feasible solution exists, or construct an equivalent instance of \textsc{PFA} whose path forest has at most $2k$ paths, whose link budget is $k$, and whose maximum link cost is $W$.
\end{lemma}
\begin{proof}
Instances on at most two vertices are handled directly, so assume $|V(G)|\geq 3$.
Without loss of generality, we may also assume that $G$ has no self-loops and no parallel edges: self-loops can be deleted, and parallel edges in \(G\) can be replaced by one copy, without changing which link sets make the graph \(2\)-vertex-connected.
If $G$ is already $2$-vertex-connected, then the empty set is optimal.
\Cref{fig:twovca-to-twovca-pf-reduction}(a) shows an example input graph.

For each connected component of $G$, consider its block-cut tree as defined in the preliminaries.
A block $B$ is a \emph{leaf block} if at most one vertex $v\in V(B)$ is a cut vertex of $G$.
Let $\mathcal{L}$ be the set consisting of all leaf blocks and all isolated vertices of $G$, and let $\lambda:=|\mathcal{L}|$.
For $B\in\mathcal{L}$, define $X_B$ as follows.
If $B=\{v\}$ is an isolated vertex, let $X_B:=\{v\}$.
If $B$ is a leaf block in a component with at least two blocks and $a_B$ is the unique cut vertex of $G$ in $B$, then let $X_B:=V(B)\setminus\{a_B\}$.
In the remaining case, $B$ is the unique block of its component, and we let $X_B:=V(B)$.

Since $G$ is not already $2$-vertex-connected, every feasible solution must contain a link with an endpoint in $X_B$ for every $B\in\mathcal{L}$.
The sets $X_B$ are pairwise disjoint, so one link can have endpoints in $X_B$ for at most two members $B\in\mathcal{L}$.
Since a feasible \textsc{2VCA} solution uses at most $k$ links, no feasible solution exists when $\lambda>2k$.

Assume $\lambda\leq 2k$.
Construct a graph $G'$ on $V(G)$ as follows.
For every block $B$ of $G$ with $|V(B)|=2$, keep one edge between the two vertices of $B$.
For every block $B$ with $|V(B)|\geq 3$, replace $B$ by a cycle on $V(B)$.
If $B$ contains at least two cut vertices of $G$, choose the cycle so that two such cut vertices are consecutive on it.
Isolated vertices of $G$ remain isolated in $G'$.
For the example in \Cref{fig:twovca-to-twovca-pf-reduction}(a), this produces the graph shown in \Cref{fig:twovca-to-twovca-pf-reduction}(b).
For every $A\subseteq L$, the graph $(V(G),E(G)\cup A)$ is $2$-vertex-connected if and only if $(V(G),E(G')\cup A)$ is $2$-vertex-connected.
To see this, fix a vertex $x$.
For every block $B$ of $G$, both $G[V(B)]-x$ and $G'[V(B)]-x$ are connected on the vertex set $V(B)\setminus\{x\}$; this is immediate for $K_2$ blocks, and for blocks with at least three vertices because $G[V(B)]$ is $2$-vertex-connected and $G'[V(B)]$ is a cycle.
Moreover, the links of $A$ have the same endpoints in the two graphs.
Therefore any path in $(V(G),E(G)\cup A)-x$ can be converted into a path in $(V(G),E(G')\cup A)-x$ by replacing each maximal subpath inside a block $B$ by a path in $G'[V(B)]-x$ with the same endpoints, and the converse conversion is identical with $G[V(B)]-x$ in place of $G'[V(B)]-x$.
Thus $(V(G),E(G)\cup A)-x$ is connected if and only if $(V(G),E(G')\cup A)-x$ is connected.

We next choose a spanning forest $S$ of $G'$.
For the running example, this is the forest shown in \Cref{fig:twovca-to-twovca-pf-reduction}(c).
Inside a $K_2$ block, take its single edge.
Inside a block with at least three vertices and at least two cut vertices, take the path obtained from its chosen cycle by deleting the edge between the two consecutive cut vertices.
Inside a leaf block with at least three vertices in a component with at least two blocks, take a spanning path whose one endpoint is its unique cut vertex.
Inside a component consisting of a single block with at least three vertices, take any spanning path of its chosen cycle.
An isolated vertex is kept as a one-vertex tree.
Gluing these paths at the cut vertices gives a spanning forest $S$ of $G'$.
In every component with at least two blocks, the resulting tree has one leaf in $X_B$ for each leaf block $B$ and no other leaves.
Indeed, in a non-leaf block, the chosen path has both endpoints at cut vertices of $G$, and each of these endpoints also belongs to another block, so an edge from that other block is incident with it after gluing.
In a leaf block, the endpoint at its unique cut vertex is similarly incident with an edge from another block, while the other endpoint lies in $X_B$ and remains a leaf.

By the standard greedy leaf-stripping argument, a tree with $q$ leaves can be partitioned into at most $q$ vertex-disjoint paths.
Applying this to each component with at least two blocks, and using one path for each single-block component and isolated vertex, partition $V(G)$ into at most $\lambda\leq 2k$ vertex-disjoint paths of $S$.
Let $F$ be the resulting path forest.
For the running example, the red edges in \Cref{fig:twovca-to-twovca-pf-reduction}(d) show this path forest.
Let $L_0$ be the set containing, for every edge $f\in E(G')\setminus E(F)$, one zero-cost link indexed by $f$ and with the same endpoints as $f$.
Extend $c$ to $L_0\uplus L$ by setting $c_f=0$ for each $f\in L_0$.
Define $\chi_f:=0$ for $f\in L_0$ and $\chi_e:=1$ for $e\in L$.
The \textsc{PFA} instance is $(F,L_0\uplus L,k,c,\chi,W)$, which has at most $2k$ paths, link budget $k$, and maximum link cost $W$.
It also satisfies $c_e=0$ whenever $\chi_e=0$.

If $A\subseteq L$ is feasible for \textsc{2VCA}, then $A':=L_0\cup A$ satisfies $(V(F),E(F)\cup A')=(V(G),E(G')\cup A)$.
By the equivalence between $G$ and $G'$, the graph $(V(G),E(G')\cup A)$ is $2$-vertex-connected, and hence $A'$ is feasible for \textsc{PFA}, has cost $c(A)$, and satisfies $\chi(A')=|A|\leq k$.
Conversely, let $A'\subseteq L_0\uplus L$ be feasible for \textsc{PFA}, and set $A:=A'\cap L$.
The graph $(V(G),E(G')\cup A)$ is obtained from $(V(F),E(F)\cup A')$ by adding the links of $L_0$ that are not in $A'$.
Since adding edges preserves $2$-vertex-connectivity, $(V(G),E(G')\cup A)$ is $2$-vertex-connected.
By the equivalence between $G$ and $G'$, the graph $(V(G),E(G)\cup A)$ is $2$-vertex-connected.
Also $c(A)=c(A')$ because all links in $L_0$ have cost $0$.
Moreover, $|A|=\chi(A)\leq\chi(A')\leq k$.
Thus the two instances are equivalent.
\end{proof}

\section{Reducing \textsc{PFA} to \textsc{2VCSS-BP}}\label{subsec:pf-to-bp}

In this section, we reduce \textsc{PFA} to the auxiliary problem \textsc{2VCSS-BP}, which is defined below and solved in \Cref{sec:2VCSSBP,sec:2VCSSBP-evaluation}.
We first contract every path and prove \Cref{lem:vertex_cond_on_contracted_graph}, which characterizes the $2$-vertex-connectivity condition by conditions on the contracted multigraph.
We then optimize over edge sets satisfying the conditions in \Cref{lem:vertex_cond_on_contracted_graph} by dynamic programming over subsets of the contracted vertices.
The main transition of this dynamic programming is formulated as an instance of \textsc{2VCSS-BP}.

\paragraph*{Contracting the graph.}
Let $(G,L,k,c,\chi,W)$ be a \textsc{PFA} instance, and let $\mathcal{P}$ be the set of paths of $G$.
For each $P\in\mathcal{P}$, write $P=(u_{P,1},\dots,u_{P,|P|})$; if $P$ has one vertex, then $P=(u_{P,1})$.
We form a multigraph $H$ by contracting each path of $G$ to a single vertex; see \Cref{fig:twovca-auxiliary-graph-contraction} for an example.
Precisely, $H$ has vertex set $\{v_P\mid P\in\mathcal{P}\}$.
For a link $e\in L$ whose endpoints lie on paths $P$ and $Q$, define $\pi(e)$ to be the edge of $H$ joining $v_P$ and $v_Q$; this is a self-loop when $P=Q$.
We keep parallel edges and regard each edge of $H$ as indexed by the corresponding link of $L$, so $\pi^{-1}(f)$ is defined for every edge $f\in E(H)$.
For $f\in E(H)$, write $c_f:=c_{\pi^{-1}(f)}$ and $\chi_f:=\chi_{\pi^{-1}(f)}$.

\begin{figure}[t]
\centering
\begin{subfigure}[t]{0.49\linewidth}
\centering
\begin{tikzpicture}[
  x=1.06cm,y=1.06cm,
  vertex/.style={circle,fill=black,inner sep=0pt,minimum size=2mm},
  path edge/.style={line width=1.05pt,red!70!black},
  link edge/.style={line width=0.6pt,densely dotted,black},
  link label/.style={fill=white,inner sep=1pt,font=\small},
  path label/.style={font=\small,red!70!black},
]
\coordinate (p11) at (-2.4, 1.25);
\coordinate (p12) at (-1.35,1.25);
\coordinate (p13) at (-0.3, 1.25);
\coordinate (p21) at (-2.4,-0.15);
\coordinate (p22) at (-1.35,-0.15);
\coordinate (p23) at (-0.3,-0.15);
\coordinate (p24) at ( 0.75,-0.15);
\coordinate (p31) at (-2.4,-1.55);
\coordinate (p32) at (-1.35,-1.55);
\coordinate (p33) at (-0.3,-1.55);

\draw[path edge] (p11) -- (p12) -- (p13);
\draw[path edge] (p21) -- (p22) -- (p23) -- (p24);
\draw[path edge] (p31) -- (p32) -- (p33);

\draw[link edge] (p11) to[out=115,in=65,looseness=1.6] node[link label,above] {$e_1$} (p13);
\draw[link edge] (p11) to[out=210,in=150] node[link label,left] {$e_2$} (p21);
\draw[link edge] (p13) to[out=-35,in=35] node[link label,right] {$e_3$} (p23);
\draw[link edge] (p24) to[out=-45,in=25] node[link label,right] {$e_4$} (p33);
\draw[link edge] (p31) to[out=-115,in=-65,looseness=1.6] node[link label,below] {$e_5$} (p33);
\draw[link edge] (p11) to[out=210,in=150,looseness=1.35] node[link label,left] {$e_6$} (p31);

\foreach \p in {p11,p12,p13,p21,p22,p23,p24,p31,p32,p33}
  \node[vertex] at (\p) {};
\node[path label] at (-1.35,1.55) {$P_1$};
\node[path label] at (-0.82,0.15) {$P_2$};
\node[path label] at (-1.35,-1.25) {$P_3$};
\end{tikzpicture}
\caption{}
\label{fig:twovca-auxiliary-graph-before-contraction}
\end{subfigure}\hfill
\begin{subfigure}[t]{0.45\linewidth}
\centering
\begin{tikzpicture}[
  x=1cm,y=1cm,
  hvertex/.style={circle,draw=black,fill=white,inner sep=1pt,minimum size=8mm,font=\small},
  link edge/.style={line width=0.6pt,densely dotted,black},
  link label/.style={fill=white,inner sep=1pt,font=\small},
]
\node[hvertex] (h1) at (0, 1.25) {$v_{P_1}$};
\node[hvertex] (h2) at (0,-0.15) {$v_{P_2}$};
\node[hvertex] (h3) at (0,-1.55) {$v_{P_3}$};

\draw[link edge] (h1) to[out=130,in=50,looseness=7] node[link label,above] {$e_1$} (h1);
\draw[link edge] (h1) to[out=-105,in=105] node[link label,left] {$e_2$} (h2);
\draw[link edge] (h1) to[out=-75,in=75] node[link label,right] {$e_3$} (h2);
\draw[link edge] (h2) -- node[link label,right] {$e_4$} (h3);
\draw[link edge] (h3) to[out=-130,in=-50,looseness=7] node[link label,below] {$e_5$} (h3);
\draw[link edge] (h1) to[out=-160,in=160,looseness=1.45] node[link label,left] {$e_6$} (h3);
\end{tikzpicture}
\caption{}
\label{fig:twovca-auxiliary-graph-after-contraction}
\end{subfigure}
\caption{Construction of the auxiliary multigraph $H$ for a \textsc{PFA} instance. (a) A path forest $G$ with links: red edges are path edges, and black dotted edges are links. (b) Contracting each path $P_i$ to $v_{P_i}$ maps each link $e_j$ to the edge $\pi(e_j)$ of $H$; in the figure, we abuse notation and write $e_j$ also for $\pi(e_j)$. The edges $e_1$ and $e_5$ are self-loops, and $e_2$ and $e_3$ are parallel edges.}
\label{fig:twovca-auxiliary-graph-contraction}
\end{figure}
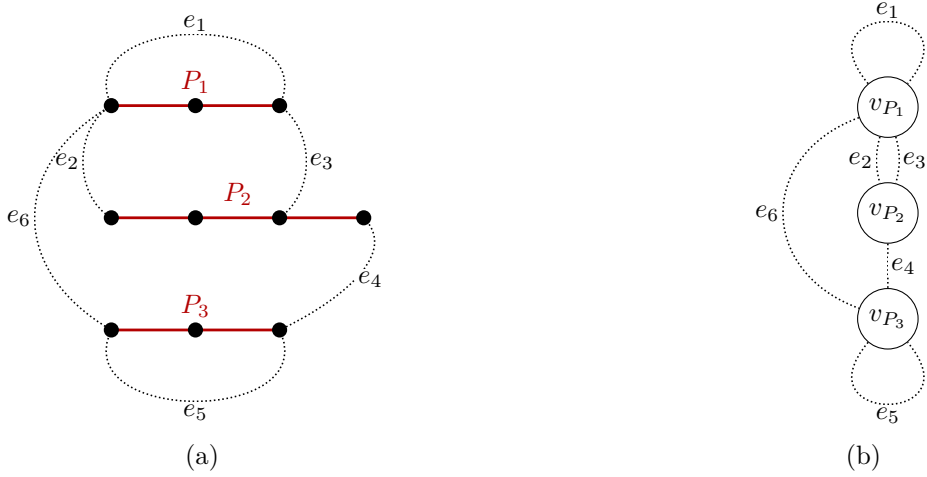

\paragraph*{Rephrasing the $2$-vertex-connectivity constraint.}
We next express the condition that $(V(G),E(G)\cup A)$ is $2$-vertex-connected as conditions on the multigraph $H$.
Let $A'\subseteq E(H)$, and let $B$ be a block of $(V(H),A')$ with $v_P\in V(B)$.
For $t\in\{2,\dots,|P|-1\}$, we say that $B$ \emph{covers} $t$ \emph{at} $P$ if there are edges $f,f'\in E(B)$ incident with $v_P$, possibly $f=f'$, such that $\pi^{-1}(f)$ has an endpoint $u_{P,i}$ with $i<t$ and $\pi^{-1}(f')$ has an endpoint $u_{P,j}$ with $t<j$.

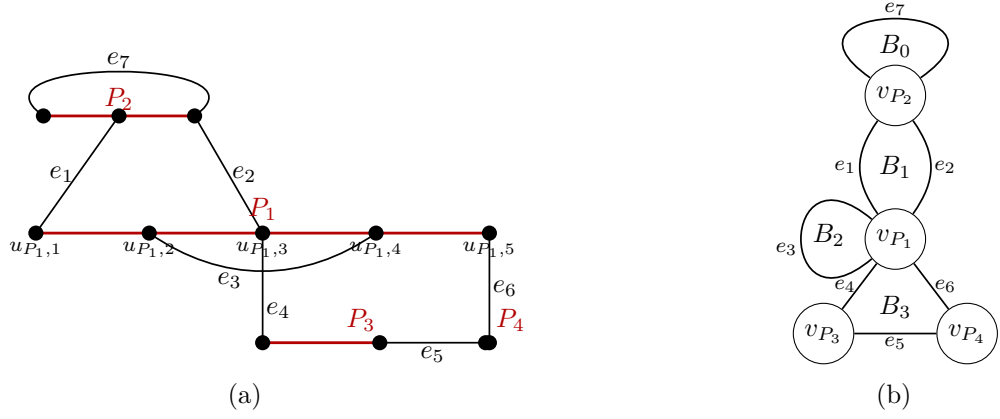
\begin{figure}[t]
\centering
\begin{subfigure}[t]{0.60\linewidth}
\centering
\begin{tikzpicture}[
  x=1cm,y=1cm,
  vertex/.style={circle,fill=black,inner sep=0pt,minimum size=2mm},
  path edge/.style={line width=1.05pt,red!70!black},
  selected edge/.style={line width=0.65pt,black},
  link label/.style={inner sep=1pt,font=\small},
  path label/.style={font=\small,red!70!black},
  index label/.style={font=\scriptsize,inner sep=1pt},
]
\coordinate (p11) at (-3.0,0);
\coordinate (p12) at (-1.5,0);
\coordinate (p13) at ( 0.0,0);
\coordinate (p14) at ( 1.5,0);
\coordinate (p15) at ( 3.0,0);
\coordinate (p21) at (-2.9,1.55);
\coordinate (p22) at (-1.9,1.55);
\coordinate (p23) at (-0.9,1.55);
\coordinate (p41) at ( 0.0,-1.45);
\coordinate (p42) at ( 1.55,-1.45);
\coordinate (p51) at ( 2.95,-1.45);
\coordinate (p52) at ( 3.0,-1.45);

\draw[path edge] (p11) -- (p12) -- (p13) -- (p14) -- (p15);
\draw[path edge] (p21) -- (p22) -- (p23);
\draw[path edge] (p41) -- (p42);
\draw[path edge] (p51) -- (p52);

\draw[selected edge] (p11) -- node[link label,left] {$e_1$} (p22);
\draw[selected edge] (p23) -- node[link label,right] {$e_2$} (p13);
\draw[selected edge] (p12) to[out=-35,in=-145] node[link label,below,pos=0.38,xshift=-2pt] {$e_3$} (p14);
\draw[selected edge] (p13) -- node[link label,right,pos=0.70] {$e_4$} (p41);
\draw[selected edge] (p42) -- node[link label,below] {$e_5$} (p51);
\draw[selected edge] (p52) -- node[link label,right] {$e_6$} (p15);
\draw[selected edge] (p21) to[out=145,in=35,looseness=1.75] node[link label,above] {$e_7$} (p23);

\foreach \p in {p11,p12,p13,p14,p15,p21,p22,p23,p41,p42,p51,p52}
  \node[vertex] at (\p) {};
\foreach \i/\p in {1/p11,2/p12,3/p13,4/p14,5/p15}
\node[index label,below=2pt] at (\p) {$u_{P_1,\i}$};
\node[path label] at (0,0.32) {$P_1$};
\node[path label] at (-1.9,1.78) {$P_2$};
\node[path label] at (1.3,-1.15) {$P_3$};
\node[path label] at (3.28,-1.15) {$P_4$};
\end{tikzpicture}
\caption{}
\label{fig:vertex-covering-by-several-blocks-original}
\end{subfigure}\hfill
\begin{subfigure}[t]{0.36\linewidth}
\centering
\begin{tikzpicture}[
  x=1cm,y=1cm,
  hvertex/.style={circle,draw=black,fill=white,inner sep=1pt,minimum size=8mm,font=\small},
  selected edge/.style={line width=0.65pt,black},
  edge label/.style={inner sep=1pt,font=\scriptsize},
  block label/.style={font=\small},
]
\node[hvertex] (p1) at (0,0) {$v_{P_1}$};
\node[hvertex] (p2) at (0,1.9) {$v_{P_2}$};
\node[hvertex] (p4) at (-0.95,-1.25) {$v_{P_3}$};
\node[hvertex] (p5) at (0.95,-1.25) {$v_{P_4}$};

\draw[selected edge] (p1) to[out=125,in=-125,looseness=1.15] node[edge label,left] {$e_1$} (p2);
\draw[selected edge] (p1) to[out=55,in=-55,looseness=1.15] node[edge label,right] {$e_2$} (p2);
\draw[selected edge] (p1) to[out=220,in=140,looseness=8] node[edge label,below left] {$e_3$} (p1);
\draw[selected edge] (p1) -- node[edge label,left] {$e_4$} (p4);
\draw[selected edge] (p4) -- node[edge label,below] {$e_5$} (p5);
\draw[selected edge] (p5) -- node[edge label,right] {$e_6$} (p1);
\draw[selected edge] (p2) to[out=145,in=35,looseness=7] node[edge label,above] {$e_7$} (p2);
\node[block label] at (0,0.95) {$B_1$};
\node[block label] at (-0.88,0.05) {$B_2$};
\node[block label] at (0,-0.9) {$B_3$};
\node[block label] at (0,2.55) {$B_0$};
\end{tikzpicture}
\caption{}
\label{fig:vertex-covering-by-several-blocks-auxiliary}
\end{subfigure}
\caption{A selected link set satisfying the conditions of \Cref{lem:vertex_cond_on_contracted_graph}, shown before and after contracting the paths. (a) The path forest before contraction: red edges are path edges, and black edges are the selected links $A=\{e_1,\ldots,e_7\}$. (b) The contracted graph $H_A=(V(H),\pi(A))$. At $P_1$, the blocks $B_1$, $B_2$, and $B_3$ cover the indices $2$, $3$, and $4$, respectively, meaning the vertices $u_{P_1,2}$, $u_{P_1,3}$, and $u_{P_1,4}$. The self-loop block $B_0$ covers index $2$ at $P_2$.}
\label{fig:vertex-covering-by-several-blocks}
\end{figure}
The condition in the following lemma is illustrated in \Cref{fig:vertex-covering-by-several-blocks}.

\begin{lemma}\label{lem:vertex_cond_on_contracted_graph}
Assume $|V(G)|\geq 3$.
For $A\subseteq L$, let $H_A:=(V(H),\pi(A))$.
Then $(V(G),E(G)\cup A)$ is $2$-vertex-connected if and only if the following conditions hold:
\begin{enumerate}
    \item[\rm{(V1)}] $H_A$ is bridgeless,
    \item[\rm{(V2)}] for each $P\in\mathcal{P}$ and each $t\in\{2,\dots,|P|-1\}$, some block $B$ of $H_A$ with $v_P\in V(B)$ covers $t$ at $P$, and
    \item[\rm{(V3)}] for each $P\in\mathcal{P}$, if $|P|=1$, then at most one block of $H_A$ contains $v_P$, and if $|P|\geq 2$, then every block $B$ of $H_A$ with $v_P\in V(B)$ has edges $f,f'\in E(B)$ incident with $v_P$, possibly $f=f'$, such that $\pi^{-1}(f)$ and $\pi^{-1}(f')$ have endpoints at distinct vertices of $P$.
\end{enumerate}
\end{lemma}
\begin{proof}
Assume first that $(V(G),E(G)\cup A)$ is $2$-vertex-connected.
Since this graph has at least three vertices, it is $2$-edge-connected.
Contracting connected paths preserves connectedness and cannot create a bridge, so \(H_A\) is bridgeless.
Thus (V1) holds.

Fix $P$.
We first prove (V2) at $P$.
If $|P|\leq 2$, there is nothing to prove.
Assume $|P|\geq 3$, and take any $t\in\{2,\dots,|P|-1\}$.
Since $(V(G),E(G)\cup A)-u_{P,t}$ is connected, there is a path $R$ in $(V(G),E(G)\cup A)-u_{P,t}$ from $u_{P,t-1}$ to $u_{P,t+1}$.
Along $R$, take the segment $R_0$ from the last visit to $\{u_{P,1},\dots,u_{P,t-1}\}$ before the first subsequent visit to $\{u_{P,t+1},\dots,u_{P,|P|}\}$.
Then the internal vertices of $R_0$ do not belong to $P$.
Let $e^-$ and $e^+$ be the first and last links of $R_0$, possibly the same link.
Their endpoints on $P$ lie in $\{u_{P,1},\dots,u_{P,t-1}\}$ and $\{u_{P,t+1},\dots,u_{P,|P|}\}$, respectively.
After contracting the paths, $R_0$ gives a closed walk from $v_P$ to $v_P$ with no internal occurrence of $v_P$, so $\pi(e^-)$ and $\pi(e^+)$ belong to one block of $H_A$ containing $v_P$.
This block covers $t$ at $P$.

We next prove (V3) at $P$.
If $|P|=1$, then more than one block containing $v_P$ would make $(V(G),E(G)\cup A)-u_{P,1}$ disconnected, a contradiction.
Thus at most one block contains $v_P$.
Now assume $|P|>1$.
A self-loop block at $v_P$ comes from a link whose two endpoints are distinct vertices of $P$.
Thus it satisfies (V3).
Consider a non-self-loop block $B$ with $v_P\in V(B)$.
Suppose that every edge of $E(B)$ incident with $v_P$ comes from a link incident with the same vertex $u_{P,s}$.
Then deleting $u_{P,s}$ separates $B-v_P$ from the rest of $(V(G),E(G)\cup A)$, contradicting the connectedness of $(V(G),E(G)\cup A)-u_{P,s}$.
Hence $B$ has two edges whose preimages have endpoints at distinct vertices of $P$.
Thus (V3) holds at $P$.

Conversely, assume (V1), (V2), and (V3).
By (V1), $H_A$ is connected, and therefore $(V(G),E(G)\cup A)$ is connected.
Fix a vertex $u_{P,t}$ of $(V(G),E(G)\cup A)$.
Every component of $H_A-v_P$ is attached to $v_P$ through a block containing $v_P$ and at least one vertex of that component.
If $|P|=1$, (V3) says that there is at most one such block, so $H_A-v_P$ is connected, and hence $(V(G),E(G)\cup A)-u_{P,1}$ is connected.
Assume $|P|>1$.
Contract every path $Q\neq P$ and each of the at most two maximal nonempty subpaths of $P-u_{P,t}$.
We show that the resulting graph is connected.
By (V3), every block containing $v_P$ and a vertex of a component of $H_A-v_P$ has incident edges whose preimages have endpoints at distinct vertices of $P$.
Thus, after deleting $u_{P,t}$, each component of $H_A-v_P$ is adjacent to one of these contracted subpaths.
If $t=1$ or $t=|P|$, there is only one contracted subpath, so all components of $H_A-v_P$ are connected to it.
If $2\leq t\leq |P|-1$, by (V2), some block covers $t$ at $P$.
The two edges certifying this cover have endpoints in the two contracted subpaths, and the block gives a connection between those edges avoiding $v_P$.
If the covering block is a self-loop block, then this self-loop itself gives the connection.
Hence, in all cases, the contracted subpaths and all components of $H_A-v_P$ lie in one connected component.
Thus the contracted graph of $(V(G),E(G)\cup A)-u_{P,t}$ is connected, and expanding the contracted connected subpaths shows that $(V(G),E(G)\cup A)-u_{P,t}$ is connected.
Thus deleting any vertex of $(V(G),E(G)\cup A)$ leaves a connected graph.
Therefore $(V(G),E(G)\cup A)$ is $2$-vertex-connected.
\end{proof}

\paragraph*{States of dynamic programming.}

From this point on, we solve \textsc{PFA} as a problem on the contracted graph $H$.
By \Cref{lem:vertex_cond_on_contracted_graph}, it suffices to find an edge set of $H$ satisfying (V1)--(V3) under the two weight bounds.
The graph $H$ has one vertex for each path of the input path forest, so $|V(H)|=|\mathcal{P}|\leq p$ for \textsc{PFA}.
We use dynamic programming over subsets of $V(H)$ and intervals on a path.

For $e\in E(H)$, write $\mathbf w_e:=(\chi_e,c_e)$, and for $A\subseteq E(H)$, write $\mathbf w(A):=\sum_{e\in A}\mathbf w_e$.
Let $\Omega:=\{0,\dots,k\}\times\{0,\dots,kW\}$, ordered componentwise, that is, $(\alpha,\beta)\leq(\alpha',\beta')$ means $\alpha\leq\alpha'$ and $\beta\leq\beta'$.
Here the first coordinate $\alpha$ and the second coordinate $\beta$ record the two weight-coordinate thresholds.
The cap $kW$ is sufficient because links with $\chi_e=0$ have $c_e=0$.
All dynamic programming entries below are subsets of $\Omega$.
Thus $(\alpha,\beta)$ belongs to an entry exactly when the corresponding object can be realized by an edge set $A$ with $\mathbf w(A)\leq(\alpha,\beta)$.
In particular, every entry is upward closed: if $(\alpha,\beta)$ belongs to an entry and $(\alpha,\beta)\leq(\alpha',\beta')\in\Omega$, then $(\alpha',\beta')$ also belongs to the entry.
For two entries $R,S\subseteq\Omega$, define
\[
    R\boxplus S:=
    \{\mathbf q\in\Omega:\text{there are }\mathbf q_1\in R,\ \mathbf q_2\in S
    \text{ with }\mathbf q_1+\mathbf q_2\leq\mathbf q\}.
\]

Let $X\subseteq V(H)$ and $v_P\in X$.
Assume $|P|>1$.
For $2\leq i\leq j\leq |P|-1$, define $a_{X,v_P,i,j}$ as follows.
The entry $a_{X,v_P,i,j}$ consists of all pairs $(\alpha,\beta)\in\Omega$ for which there is an edge subset $A\subseteq E(H[X])$ with $\mathbf w(A)\leq(\alpha,\beta)$ such that
\begin{itemize}
    \item $(X,A)$ is bridgeless,
    \item for each $v_Q\in X\setminus\{v_P\}$, the conditions in (V2) and (V3) hold at $Q$ in $(X,A)$,
    \item there is a unique block $B$ of $(X,A)$ with $v_P\in V(B)$, and
    \item the condition in (V3) holds at $P$ for $B$, and $B$ covers every index $t$ with $i\leq t\leq j$ at $P$.
\end{itemize}
We also use $a_{X,v_P,1,1}$ for the same definition with no covering requirement at $P$; the unique block containing $v_P$ is still required to satisfy (V3) at $P$.
See \Cref{fig:vertex-a-state-example} for an example of a solution counted by an $a$-value.

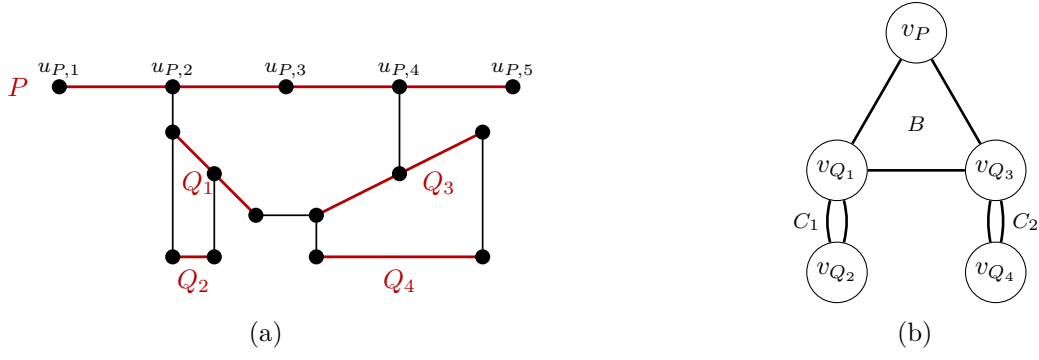
\begin{figure}[t]
\centering
\begin{subfigure}[t]{0.60\linewidth}
\centering
\begin{tikzpicture}[
  x=1cm,y=1cm,
  vertex/.style={circle,fill=black,inner sep=0pt,minimum size=2mm},
  path edge/.style={line width=1.05pt,red!70!black},
  selected edge/.style={line width=0.65pt,black},
  path label/.style={font=\small,red!70!black},
  index label/.style={font=\scriptsize,inner sep=1pt},
  block label/.style={font=\scriptsize,inner sep=1pt},
]
\coordinate (p1) at (-3.0, 2.0);
\coordinate (p2) at (-1.5, 2.0);
\coordinate (p3) at ( 0.0, 2.0);
\coordinate (p4) at ( 1.5, 2.0);
\coordinate (p5) at ( 3.0, 2.0);
\coordinate (q11) at (-1.50,1.40);
\coordinate (q12) at (-0.95,0.85);
\coordinate (q13) at (-0.40,0.30);
\coordinate (q21) at (-1.50,-0.25);
\coordinate (q22) at (-0.95,-0.25);
\coordinate (q31) at ( 0.40,0.30);
\coordinate (q32) at ( 1.50,0.85);
\coordinate (q33) at ( 2.60,1.40);
\coordinate (q41) at ( 0.40,-0.25);
\coordinate (q42) at ( 2.60,-0.25);

\draw[path edge] (p1) -- (p2) -- (p3) -- (p4) -- (p5);
\draw[path edge] (q11) -- (q12) -- (q13);
\draw[path edge] (q21) -- (q22);
\draw[path edge] (q31) -- (q32) -- (q33);
\draw[path edge] (q41) -- (q42);

\draw[selected edge] (p2) -- (q11);
\draw[selected edge] (q13) -- (q31);
\draw[selected edge] (q32) -- (p4);
\draw[selected edge] (q11) -- (q21);
\draw[selected edge] (q12) -- (q22);
\draw[selected edge] (q31) -- (q41);
\draw[selected edge] (q33) -- (q42);

\foreach \v in {p1,p2,p3,p4,p5,q11,q12,q13,q21,q22,q31,q32,q33,q41,q42}
  \node[vertex] at (\v) {};
\foreach \i/\v in {1/p1,2/p2,3/p3,4/p4,5/p5}
  \node[index label,above=2pt] at (\v) {$u_{P,\i}$};
\node[path label] at (-3.55,2.0) {$P$};
\node[path label] at (-1.16,0.72) {$Q_1$};
\node[path label] at ($(q21)!0.5!(q22)+(0,-0.33)$) {$Q_2$};
\node[path label] at (2.02,0.72) {$Q_3$};
\node[path label] at (1.50,-0.58) {$Q_4$};
\end{tikzpicture}
\caption{}
\label{fig:vertex-a-state-example-original}
\end{subfigure}\hfill
\begin{subfigure}[t]{0.36\linewidth}
\centering
\begin{tikzpicture}[
  x=1cm,y=1cm,
  hvertex/.style={circle,draw=black,fill=white,inner sep=1pt,minimum size=8mm,font=\small},
  root block/.style={black,line width=1pt},
  side block one/.style={black,line width=1pt},
  side block two/.style={black,line width=1pt},
  side block three/.style={green!55!black,line width=1pt},
  block label/.style={font=\scriptsize},
]
\node[hvertex] (p) at (0,1.82) {$v_P$};
\node[hvertex] (q1) at (-1.05,0.00) {$v_{Q_1}$};
\node[hvertex] (q2) at (-1.05,-1.35) {$v_{Q_2}$};
\node[hvertex] (q3) at (1.05,0.00) {$v_{Q_3}$};
\node[hvertex] (q4) at (1.05,-1.35) {$v_{Q_4}$};

\draw[root block] (p)--(q1)--(q3)--(p);
\draw[side block one] (q1) to[bend left=13] (q2);
\draw[side block one] (q1) to[bend right=13] (q2);
\draw[side block two] (q3) to[bend left=10] (q4);
\draw[side block two] (q3) to[bend right=10] (q4);
\node[block label] at (0.00,0.62) {$B$};
\node[block label] at (-1.45,-0.68) {$C_1$};
\node[block label] at (1.45,-0.68) {$C_2$};
\end{tikzpicture}
\caption{}
\label{fig:vertex-a-state-example-contracted}
\end{subfigure}
\caption{An example of a solution counted by an $a$-value for \textsc{PFA}. In (a), the selected links on the path forest form a solution counted by $a_{X,v_P,3,3}$. After contraction, (b) has a unique block $B$ containing the root vertex $v_P$; the other blocks hang below $B$. The block $B$ covers index $3$ at $P$, but it does not cover every internal index of $P$. At $Q_1$, the block $B$ itself covers the internal index, while the lower block $C_1$ does not. At $Q_3$, the block $B$ does not cover the internal index, while $C_2$ does.}
\label{fig:vertex-a-state-example}
\end{figure}

The following lemma gives the initial values for the dynamic programming.
\begin{lemma}\label{lem:vertex_a_singlevertex}
Assume $|P|>1$.
For $2\leq i\leq j\leq |P|-1$, the entry $a_{\{v_P\},v_P,i,j}$ is the union, over all links $e$ with endpoints $u_{P,\ell}$ and $u_{P,r}$ for some $\ell,r$ with $\ell<i\leq j<r$, of the upward-closed sets $\{(\alpha,\beta)\in\Omega:\mathbf w_e\leq(\alpha,\beta)\}$.
The entry $a_{\{v_P\},v_P,1,1}$ is defined in the same way, but the union ranges over all links $e$ with endpoints $u_{P,\ell}$ and $u_{P,r}$ for some $\ell<r$.
\end{lemma}
\begin{proof}
Since $H[\{v_P\}]$ has no vertex other than $v_P$, the unique block containing $v_P$ must be a self-loop block.
For $2\leq i\leq j\leq |P|-1$, such a block satisfies the last two conditions in the definition exactly when its edge is $\pi(e)$ for a link $e$ with endpoints $u_{P,\ell}$ and $u_{P,r}$ for some $\ell,r$ with $\ell<i\leq j<r$.
For $a_{\{v_P\},v_P,1,1}$, only the requirement $\ell<r$ remains.
Conversely, every edge described in the statement forms a self-loop block, and the resulting one-vertex graph is bridgeless and satisfies the relevant conditions; it contributes exactly the set of pairs $(\alpha,\beta)\in\Omega$ with $\mathbf w_e\leq(\alpha,\beta)$.
\end{proof}

Assume $|P|>1$.
For $2\leq i\leq j\leq |P|-1$, define $b_{X,v_P,i,j}\subseteq\Omega$ as the set of pairs $(\alpha,\beta)$ for which there is an edge subset $A\subseteq E(H[X])$ with $\mathbf w(A)\leq(\alpha,\beta)$ such that
\begin{itemize}
    \item $(X,A)$ is bridgeless,
    \item for each $v_Q\in X\setminus\{v_P\}$, the conditions in (V2) and (V3) hold at $Q$ in $(X,A)$,
    \item the condition in (V3) holds at $P$ in $(X,A)$, and
    \item for every $t\in\{2,\dots,|P|-1\}\setminus\{i,\dots,j\}$, some block $B$ of $(X,A)$ with $v_P\in V(B)$ covers $t$ at $P$.
\end{itemize}
We define $b_{X,v_P,1,1}$ by requiring every index in $\{2,\dots,|P|-1\}$ to be covered at $P$.
For notational convenience in the recurrences below, when $i>j$ we use $b_{X,v_P,i,j}$ as shorthand for $b_{X,v_P,1,1}$.
This shorthand corresponds to leaving no uncovered internal index of $P$.

\paragraph*{Algorithm output.}

Before giving the dynamic programming recurrences, we explain how to obtain a solution to \textsc{PFA} from the defined $b$-values.
If some path $P$ satisfies $|P|>1$, then an edge set attaining $b_{V(H),v_P,1,1}$ satisfies (V1), (V2), and (V3) in \Cref{lem:vertex_cond_on_contracted_graph}.
Hence the optimum value is the minimum $\beta$ such that $(k,\beta)\in b_{V(H),v_P,1,1}$.
Assume all paths are singleton paths.
We guess a link $e$ in an optimum solution, delete it from $L$, and regard it as the only edge of a new two-vertex path.
Solving the resulting instance with budget $k-\chi_e$ and adding $c_e$ gives the optimum among solutions containing $e$.
Taking the minimum over all links costs only a polynomial factor and reduces the case to one with a path on two vertices.

\begin{figure}[t]
\centering
\begin{subfigure}[t]{0.60\linewidth}
\centering
\begin{tikzpicture}[
  x=0.65cm,y=0.65cm,
  vertex/.style={circle,fill=black,inner sep=0pt,minimum size=2mm},
  path edge/.style={line width=1.05pt,red!70!black},
  selected edge/.style={line width=0.65pt,black},
  path label/.style={font=\small,red!70!black},
  index label/.style={font=\scriptsize,inner sep=1pt},
]
\foreach \i/\x in {1/0,2/1.27,3/2.54,4/3.81,5/5.08,6/6.35,7/7.62,8/8.89,9/10.16}
  \coordinate (p\i) at (\x,3.85);
\coordinate (q11) at (0.00,2.25);
\coordinate (q12) at (1.45,2.25);
\coordinate (q13) at (2.90,2.25);
\coordinate (q21) at (1.45,0.85);
\coordinate (q22) at (2.90,0.85);
\coordinate (q31) at (3.45,2.25);
\coordinate (q32) at (4.35,2.25);
\coordinate (q41) at (5.08,2.25);
\coordinate (q42) at (7.62,2.25);
\coordinate (q51) at (8.89,2.25);
\coordinate (q52) at (10.16,2.25);

\draw[path edge] (p1)--(p2)--(p3)--(p4)--(p5)--(p6)--(p7)--(p8)--(p9);
\draw[path edge] (q11)--(q12)--(q13);
\draw[path edge] (q21)--(q22);
\draw[path edge] (q31)--(q32);
\draw[path edge] (q41)--(q42);
\draw[path edge] (q51)--(q52);
\draw[selected edge] (p1)--(q11);
\draw[selected edge] (q12)--(p3);
\draw[selected edge] (q11)--(q21);
\draw[selected edge] (q13)--(q22);
\draw[selected edge] (p2)--(q31);
\draw[selected edge] (q32)--(p4);
\draw[selected edge] (p5)--(q41);
\draw[selected edge] (q42)--(p7);
\draw[selected edge] (p7)--(q51);
\draw[selected edge] (q52)--(p9);

\foreach \v in {p1,p2,p3,p4,p5,p6,p7,p8,p9,q11,q12,q13,q21,q22,q31,q32,q41,q42,q51,q52}
  \node[vertex] at (\v) {};
\foreach \i in {1,...,9}
  \node[index label,above=2pt] at (p\i) {$u_{P,\i}$};
\node[path label] at (3.18,3.54) {$P$};
\node[path label] at (1.18,1.92) {$Q_1$};
\node[path label] at (2.18,0.58) {$Q_2$};
\node[path label] at (3.90,1.92) {$Q_3$};
\node[path label] at (6.35,1.92) {$Q_4$};
\node[path label] at (9.52,1.92) {$Q_5$};
\end{tikzpicture}
\caption{}
\label{fig:vertex-b-state-example-original}
\end{subfigure}\hfill
\begin{subfigure}[t]{0.36\linewidth}
\centering
\begin{tikzpicture}[
  x=0.82cm,y=0.82cm,
  hvertex/.style={circle,draw=black,fill=white,inner sep=1pt,minimum size=8mm,font=\small},
  left block one/.style={black,line width=1pt},
  left block two/.style={black,line width=1pt},
  right block/.style={black,line width=1pt},
  redundant block/.style={black,line width=1pt},
  note/.style={font=\scriptsize},
]
\node[hvertex] (p) at (0,1.75) {$v_P$};
\node[hvertex] (q1) at (-2.40,0.20) {$v_{Q_1}$};
\node[hvertex] (q2) at (-2.40,-1.10) {$v_{Q_2}$};
\node[hvertex] (q3) at (-0.80,0.20) {$v_{Q_3}$};
\node[hvertex] (q4) at (0.80,0.20) {$v_{Q_4}$};
\node[hvertex] (q5) at (2.40,0.20) {$v_{Q_5}$};

\draw[left block one] (p) to[bend left=12] (q1);
\draw[left block one] (p) to[bend right=12] (q1);
\draw[left block one] (q1) to[bend left=12] (q2);
\draw[left block one] (q1) to[bend right=12] (q2);
\draw[left block two] (p) to[bend left=12] (q3);
\draw[left block two] (p) to[bend right=12] (q3);
\draw[redundant block] (p) to[bend left=12] (q4);
\draw[redundant block] (p) to[bend right=12] (q4);
\draw[right block] (p) to[bend left=12] (q5);
\draw[right block] (p) to[bend right=12] (q5);
\node[note] at (-1.70,1.34) {$B_1$};
\node[note] at (-2.82,-0.45) {$C$};
\node[note] at (-0.42,0.92) {$B_3$};
\node[note] at (0.42,0.92) {$B_4$};
\node[note] at (1.70,1.34) {$B_5$};
\end{tikzpicture}
\caption{}
\label{fig:vertex-b-state-example-contracted}
\end{subfigure}
\caption{An example of a solution counted by a $b$-value for \textsc{PFA}. In (a), the selected links form a solution counted by $b_{X,v_P,4,7}$. It is not counted by $b_{X,v_P,4,5}$, because index $7$ is not covered at $P$. After contraction, (b) has several blocks containing the root vertex $v_P$: the blocks $B_1$ and $B_3$ cover the required left-side indices $2$ and $3$ at $P$, while $B_5$ covers the required right-side index $8$. The block $B_4$ covers the middle index $6$, but it does not contribute to the indices required by this $b$-value. For suitable vertex sets $Y_1,Y_3,Y_4,Y_5$, the subgraphs attached to $v_P$ through $B_1,B_3,B_4,B_5$ are counted by $a_{Y_1,v_P,2,2}$, $a_{Y_3,v_P,3,3}$, $a_{Y_4,v_P,6,6}$, and $a_{Y_5,v_P,8,8}$, respectively; the $b$-value is obtained by gluing such subgraphs at $v_P$.}
\label{fig:vertex-b-state-example}
\end{figure}
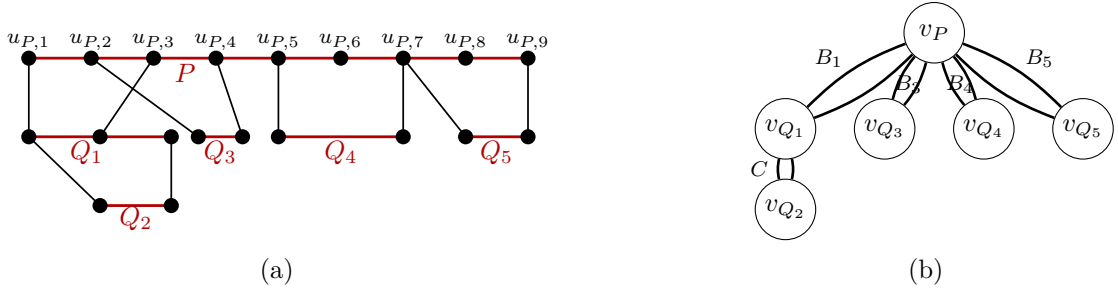

\paragraph*{Computing $b$ from $a$.}

Fix $v_P$ with $|P|>1$.
Fix $X\subseteq V(H)$ with $v_P\in X$.
We compute $b_{X,v_P,\cdot,\cdot}$ by gluing $a$-entries at the common vertex $v_P$.
Each glued entry is indexed by a subset $Y\subseteq X$ with $v_P\in Y$; an entry $a_{Y,v_P,\ell,r}$ with $2\leq \ell\leq r\leq |P|-1$ represents a subgraph whose unique block containing $v_P$ covers every index $t$ with $\ell\leq t\leq r$ at $P$.

First define $b^{\mathrm{free}}_{X,v_P}\subseteq\Omega$ by families of entries $a_{Y,v_P,1,1}$, which impose no covering requirement at $P$, for which the sets $Y\setminus\{v_P\}$ over the chosen entries form a partition of $X\setminus\{v_P\}$.
It is computed by
\[
    b^{\mathrm{free}}_{\{v_P\},v_P}=\Omega,
\]
and, for $X\neq\{v_P\}$,
\[
    b^{\mathrm{free}}_{X,v_P}
    =
    \bigcup_{\substack{Y\subsetneq X\\ v_P\in Y}}
    \left(
        b^{\mathrm{free}}_{Y,v_P}
        \boxplus
        a_{(X\setminus Y)\cup\{v_P\},v_P,1,1}
    \right)
\]
For $1\leq t\leq |P|-1$, define $b^{\mathrm{left}}_{X,v_P,t}\subseteq\Omega$ by families of $a$-entries that cover every index $s\in\{2,\dots,|P|-1\}$ with $s\leq t$.
For $2\leq t\leq |P|$, define $b^{\mathrm{right}}_{X,v_P,t}$ symmetrically, requiring every index $s\in\{2,\dots,|P|-1\}$ with $t\leq s$ to be covered.
Clearly,
\[
    b^{\mathrm{left}}_{X,v_P,1}
    =
    b^{\mathrm{right}}_{X,v_P,|P|}
    =
    b^{\mathrm{free}}_{X,v_P}.
\]
For $2\leq t\leq |P|-1$, the recurrences are
\[
\begin{aligned}
    b^{\mathrm{left}}_{X,v_P,t}
    &=
    \bigcup_{\substack{Y\subseteq X\\ v_P\in Y}}
    \bigcup_{\substack{2\leq \ell\leq t\leq r\leq |P|-1}}
        \left(
            b^{\mathrm{left}}_{Y,v_P,\ell-1}
            \boxplus
            a_{(X\setminus Y)\cup\{v_P\},v_P,\ell,r}
        \right),\\
    b^{\mathrm{right}}_{X,v_P,t}
    &=
    \bigcup_{\substack{Y\subseteq X\\ v_P\in Y}}
    \bigcup_{\substack{2\leq \ell\leq t\leq r\leq |P|-1}}
        \left(
            b^{\mathrm{right}}_{Y,v_P,r+1}
            \boxplus
            a_{(X\setminus Y)\cup\{v_P\},v_P,\ell,r}
        \right).
\end{aligned}
\]
Finally, combine the left and right tables.
First, set
\[
    b_{X,v_P,1,1}
    =
    b^{\mathrm{left}}_{X,v_P,|P|-1}.
\]
For $2\leq i\leq j\leq |P|-1$, set
\[
    b_{X,v_P,i,j}
    =
    b_{X,v_P,1,1}
    \cup
    \bigcup_{\substack{Y\subseteq X\\ v_P\in Y}}
    \left(
        b^{\mathrm{left}}_{Y,v_P,i-1}
        \boxplus
        b^{\mathrm{right}}_{(X\setminus Y)\cup\{v_P\},v_P,j+1}
    \right)
\]

\begin{lemma}\label{lem:vertex_b_from_a}
Assume $|P|>1$.
For each $X\subseteq V(H)$ with $v_P\in X$, the above update correctly computes $b_{X,v_P,\cdot,\cdot}$ from the entries $a_{Y,v_P,\ell,r}$ with $Y\subseteq X$ in $O^*(2^{|X|}W)$ time, suppressing polynomial factors in $k$.
\end{lemma}
\begin{proof}
The formula for $b^{\mathrm{free}}$ is the standard decomposition of the blocks glued at $v_P$.
The recurrences may count terms in which the same singleton self-loop $a$-entry is used more than once.
This is harmless because the entries are upward closed: replacing the repeated uses by one copy preserves the covered indices and gives a componentwise no larger weight vector.
Every used $a$-entry already requires its unique block containing $v_P$ to satisfy (V3) at $P$.

For $b^{\mathrm{left}}_{X,v_P,t}$ with $2\leq t\leq |P|-1$, choose one $a$-entry whose unique block containing $v_P$ covers $t$ at $P$.
Write the chosen entry as $a_{(X\setminus Y)\cup\{v_P\},v_P,\ell,r}$.
Since its unique block containing $v_P$ covers $t$ at $P$, we have $\ell\leq t\leq r$.
The remaining entries using $Y\setminus\{v_P\}$ must cover all indices in $\{2,\dots,\ell-1\}$ and are counted by $b^{\mathrm{left}}_{Y,v_P,\ell-1}$; the equality $b^{\mathrm{left}}_{Y,v_P,1}=b^{\mathrm{free}}_{Y,v_P}$ covers the case $\ell=2$.
This gives the recurrence for $b^{\mathrm{left}}$.
The recurrence for $b^{\mathrm{right}}$ is symmetric.
Conversely, any choice in these recurrences gives a family satisfying the corresponding left or right covering requirement by gluing the chosen $a$-entry to the remaining entries at $v_P$.

For the final formulas, $b^{\mathrm{left}}_{X,v_P,|P|-1}$ covers every internal index of $P$; this gives the formula for $b_{X,v_P,1,1}$.
For $2\leq i\leq j\leq |P|-1$, the union over $Y$ covers exactly the two required outside ranges $\{2,\dots,i-1\}$ and $\{j+1,\dots,|P|-1\}$.
The entry $b_{X,v_P,1,1}$ is also included because it covers all internal indices of $P$.

Conversely, take a feasible solution counted by $b_{X,v_P,i,j}$ and split the blocks containing $v_P$ into their corresponding $a$-entries.
If these $a$-entries cover every internal index of $P$, the solution is counted by $b_{X,v_P,1,1}$.
Otherwise, choose an index $t\in\{2,\dots,|P|-1\}$ that is not covered.
Then $i\leq t\leq j$, and this case does not occur for $b_{X,v_P,1,1}$.
No $a$-entry covers both an index smaller than $i$ and an index larger than $j$, since such an $a$-entry would also cover $t$.
Thus the $a$-entries covering indices smaller than $i$ and the $a$-entries covering indices larger than $j$ are disjoint.
Assign all remaining $a$-entries to either side.
This gives one term in the union for $b_{X,v_P,i,j}$.

The updates require convolutions over $Y\subseteq X$ and over pairs $(\alpha,\beta)\in\Omega$.
Consider one such convolution term $\bigcup_{Y\subseteq X}(R_Y\boxplus S_{X\setminus Y})$, where $R_Z$ and $S_Z$ are entries indexed by subsets $Z\subseteq X$.
For an entry $T$, define $f_T(\alpha):=\min\{\beta:(\alpha,\beta)\in T\}$ for $\alpha\in\{0,\dots,k\}$, where the minimum is $\infty$ if the set is empty.
Then $(\alpha,\beta)\in T$ exactly when $f_T(\alpha)\leq\beta$.
After fixing the first coordinates $\alpha_1,\alpha_2$, the second coordinate is handled by min-plus subset convolution: define $r_{\alpha_1}(Y):=f_{R_Y}(\alpha_1)$ and $s_{\alpha_2}(Y):=f_{S_Y}(\alpha_2)$, and compute $\min_{Y\subseteq X}(r_{\alpha_1}(Y)+s_{\alpha_2}(X\setminus Y))$ with costs capped at $kW$.
By \Cref{lem:min-plus-subset-convolution}, this takes $O^*(2^{|X|}W)$ time for each fixed pair $(\alpha_1,\alpha_2)$, and the number of such pairs is polynomial in $k$.
The choices of interval indices are polynomially many.
\end{proof}

\paragraph*{Oracle for computing $a$ from $b$.}

Now consider computing $a_{X,v_P,i,j}$, assuming that all entries $b_{Y,v_Q,\ell,r}$ with $Y\subsetneq X$ are already known.
For this purpose, we use the following auxiliary problem as an oracle; its algorithm is given in \Cref{sec:2VCSSBP,sec:2VCSSBP-evaluation}.

\boxproblemnoparam{$2$-Vertex-Connected Spanning Subgraph with Boundary-Pair Costs (2VCSS-BP)}{
    A multigraph $G$ with no self-loops and $|V(G)|\geq 2$, 
    integers $k$ and $W\geq 0$,
    an edge weight vector $\mathbf w_e=(\chi_e,c_e)\in \{0,1\}\times\{0,\dots,W\}$ for each $e\in E(G)$,
    a finite set $U$, and
    a finite set $\mathcal C_v$ of triples $(F,S,\mathbf s)$ with $F\in\binom{\delta_G(v)}{2}$, $S\subseteq U$, and $\mathbf s\in\{0,\dots,k\}\times\{0,\dots,kW\}$, for each $v\in V(G)$.
}{
    Select $A\subseteq E(G)$. 
    Moreover, for each $v\in V(G)$, select $(F_v,S_v,\mathbf s_v)\in\mathcal C_v$.
    Among $(A,(F_v,S_v,\mathbf s_v)_{v\in V(G)})$ with
    \begin{enumerate}
        \item[(E1)] $\bigcup_{v\in V(G)}F_v\subseteq A$,
        \item[(E2)] $\{S_v\}_{v\in V(G)}$ is a partition of $U$, and
        \item[(E3)] $(V(G),A)$ is $2$-vertex-connected,
    \end{enumerate}
    Output the set of all pairs $(\alpha,\beta)\in\{0,\dots,k\}\times\{0,\dots,kW\}$ for which some such $(A,(F_v,S_v,\mathbf s_v)_{v\in V(G)})$ satisfies
    $\sum_{e\in A}\mathbf w_e+\sum_{v\in V(G)}\mathbf s_v\leq(\alpha,\beta)$.
}

In condition (E3), we use the convention that a two-vertex multigraph with at least two parallel edges is admissible as \(2\)-vertex-connected.
A two-vertex graph with a single edge is irrelevant here, because each vertex must choose a boundary pair of two incident edges.
The output entry is upward closed by definition.
Here \(A\) is a set of edges; hence, even if the same edge belongs to \(F_v\) and \(F_w\), its edge weight is included only once in the sums over \(A\).

Before going into the details, we explain the intuition.
We guess the vertex set $Z$ of the unique block containing the root $v_P$.
The block itself has to be selected as a $2$-vertex-connected subgraph on $Z$.
Each vertex $v_Q\in Z$ may have components hanging outside the block, whose feasible weight vectors are represented by an already computed entry of the form $b_{Y,v_Q,i,j}$ with $Y\subsetneq X$.
Since the indices $i,\dots,j$ of $Q$ are not required to be covered inside the attached part, they must instead be covered by the block on $Z$.
The oracle explicitly chooses the pair of edges $F_{v_Q}$ incident with $v_Q$, which ensures that $i,\dots,j$ is covered by the block on $Z$; the set $S_{v_Q}$ represents the hanging vertex set $Y\setminus\{v_Q\}$, and $\mathbf s_{v_Q}$ records its weight vector.

\paragraph*{Computing $a$ from $b$.}

Fix $v_P$ with $|P|>1$ and consider computing an entry $a_{X,v_P,i,j}$.
The case $X=\{v_P\}$ is handled by \Cref{lem:vertex_a_singlevertex}.
For the remaining case, $|X|\geq 2$, and the unique block containing $v_P$ has vertex set different from $\{v_P\}$.
We guess its vertex set $Z\subseteq X$ with $v_P\in Z$ and $|Z|\geq 2$, and use \textsc{2VCSS-BP} on $H^\star[Z]$, where $H^\star[Z]$ is obtained from $H[Z]$ by deleting self-loops.

For an edge $e\in\delta_{H^\star[Z]}(v_Q)$, let $\mathrm{pos}_Q(e)$ be the index $t$ such that $u_{Q,t}$ is an endpoint of $\pi^{-1}(e)$.
This index is unique because $H^\star[Z]$ has no self-loops.
For $F=\{e,e'\}\in\binom{\delta_{H^\star[Z]}(v_Q)}{2}$, define
\[
    \ell^Q_F:=\min\{\mathrm{pos}_Q(e),\mathrm{pos}_Q(e')\},
    \qquad
    r^Q_F:=\max\{\mathrm{pos}_Q(e),\mathrm{pos}_Q(e')\}.
\]
For $2\leq i\leq j\leq |P|-1$, define the boundary-pair choices $\mathcal C^{P,Z,i,j}_{v_Q}$ as follows.
When $Q=P$, for every $F\in\binom{\delta_{H^\star[Z]}(v_P)}{2}$ with $\ell^P_F<i\leq j<r^P_F$, include the triple $(F,\emptyset,(0,0))$.
When $Q\neq P$ and $|Q|=1$, for every $F\in\binom{\delta_{H^\star[Z]}(v_Q)}{2}$, include the triple $(F,\emptyset,(0,0))$.
When $Q\neq P$ and $|Q|>1$, for every $S\subseteq X\setminus Z$, every $F\in\binom{\delta_{H^\star[Z]}(v_Q)}{2}$ with $\ell^Q_F<r^Q_F$, and every $\mathbf q\in b_{S\cup\{v_Q\},v_Q,\ell^Q_F+1,r^Q_F-1}$, include the triple $(F,S,\mathbf q)$.
For $\mathcal C^{P,Z,1,1}_{v_Q}$, use the same definition except that, when $Q=P$, the condition $\ell^P_F<i\leq j<r^P_F$ is replaced by $\ell^P_F<r^P_F$.
For each candidate $Z$, the graph $H^\star[Z]$, the inherited edge weights, the universe $U:=X\setminus Z$, and the boundary-pair choice sets $\mathcal C^{P,Z,i,j}_{v_Q}$ form a \textsc{2VCSS-BP} instance.
The next lemma says that $a_{X,v_P,i,j}$ is obtained by solving these instances for all choices of $Z\subseteq X$ with $v_P\in Z$ and $|Z|\geq 2$, and then taking the union of the returned entries.
\Cref{fig:vertex-a-from-b-blackbox} illustrates the oracle call for one guessed set $Z$.

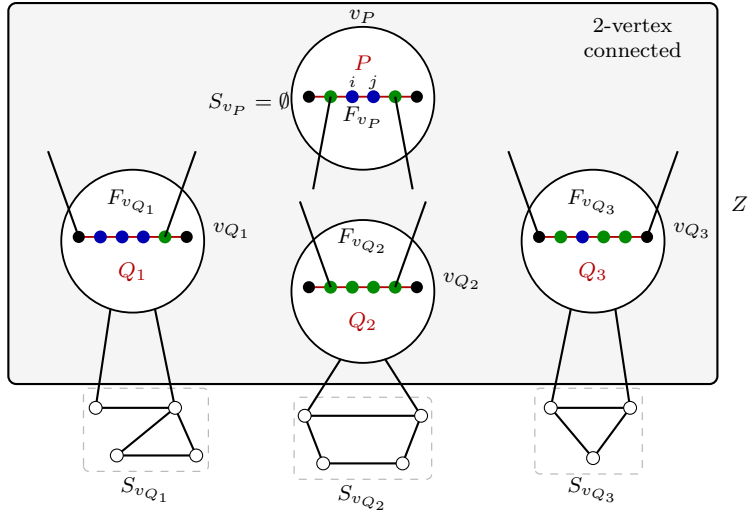
\begin{figure}[t]
\centering
\begin{tikzpicture}[
  x=0.70cm,y=0.70cm,
  rootbox/.style={draw=black,rounded corners=3pt,fill=gray!8,line width=0.8pt},
  vertex/.style={circle,draw=black,fill=white,inner sep=1pt,minimum size=7mm,font=\small},
  qvertex/.style={draw=black,fill=white,line width=0.7pt},
  path vertex/.style={circle,fill=black,inner sep=0pt,minimum size=1.6mm},
  covered path vertex/.style={circle,fill=blue!70!black,inner sep=0pt,minimum size=1.7mm},
  uncovered path vertex/.style={circle,fill=green!55!black,inner sep=0pt,minimum size=1.7mm},
  s vertex/.style={circle,draw=black,fill=white,inner sep=0pt,minimum size=1.7mm},
  path edge/.style={line width=0.85pt,red!70!black},
  small path edge/.style={line width=0.7pt,red!70!black},
  small cover/.style={line width=1.55pt,blue!65!black},
  cover/.style={line width=1.5pt,blue!65!black},
  block edge/.style={line width=0.8pt,black},
  incidence/.style={line width=0.8pt,black},
  attach/.style={line width=0.6pt,gray!65},
  setbox/.style={draw=gray!55,dashed,rounded corners=2pt,line width=0.45pt},
  note/.style={font=\scriptsize,align=center,inner sep=1pt},
  tiny note/.style={font=\tiny,align=center,inner sep=0.5pt},
]
\draw[rootbox] (-6.70,0.70) rectangle (6.70,7.90);
\coordinate (vp) at (0,6.10);
\coordinate (qone) at (-4.35,3.40);
\coordinate (qtwo) at (0,2.45);
\coordinate (qthree) at (4.35,3.40);
\draw[qvertex] (vp) circle [radius=1.35];
\draw[qvertex] (qone) circle [radius=1.35];
\draw[qvertex] (qtwo) circle [radius=1.35];
\draw[qvertex] (qthree) circle [radius=1.35];
\node[note] at ($(vp)+(0,1.55)$) {$v_P$};
\foreach \i/\dx in {1/-1.02,2/-0.61,3/-0.20,4/0.20,5/0.61,6/1.02}
  \coordinate (vpp\i) at ($(vp)+(\dx,0.03)$);
\draw[small path edge] (vpp1)--(vpp2)--(vpp3)--(vpp4)--(vpp5)--(vpp6);
\foreach \i in {1,6}
  \node[path vertex] at (vpp\i) {};
\foreach \i in {2,5}
  \node[uncovered path vertex] at (vpp\i) {};
\foreach \i in {3,4}
  \node[covered path vertex] at (vpp\i) {};
\draw[incidence] (vpp2) -- ++(-0.33,-1.75);
\draw[incidence] (vpp5) -- ++(0.33,-1.75);
\node[note] at ($(vp)+(0,-0.36)$) {$F_{v_P}$};
\node[tiny note] at ($(vpp3)+(0,0.28)$) {$i$};
\node[tiny note] at ($(vpp4)+(0,0.28)$) {$j$};
\node[note,red!70!black] at ($(vp)+(0,0.68)$) {$P$};
\foreach \name/\lab in {qone/Q_1,qtwo/Q_2,qthree/Q_3}
  \foreach \i/\dx in {1/-1.02,2/-0.61,3/-0.20,4/0.20,5/0.61,6/1.02}
    \coordinate (\name p\i) at ($(\name)+(\dx,0.08)$);
\foreach \name in {qone,qtwo,qthree}
  \draw[small path edge] (\name p1)--(\name p2)--(\name p3)--(\name p4)--(\name p5)--(\name p6);
\foreach \i in {1,6}
  \node[path vertex] at (qonep\i) {};
\node[uncovered path vertex] at (qonep5) {};
\foreach \i in {2,3,4}
  \node[covered path vertex] at (qonep\i) {};
\foreach \i in {1,6}
  \node[path vertex] at (qtwop\i) {};
\foreach \i in {2,3,4,5}
  \node[uncovered path vertex] at (qtwop\i) {};
\foreach \i in {1,6}
  \node[path vertex] at (qthreep\i) {};
\foreach \i in {2,4,5}
  \node[uncovered path vertex] at (qthreep\i) {};
\node[covered path vertex] at (qthreep3) {};
\node[note,red!70!black] at ($(qone)+(0,-0.58)$) {$Q_1$};
\node[note,red!70!black] at ($(qtwo)+(0,-0.58)$) {$Q_2$};
\node[note,red!70!black] at ($(qthree)+(0,-0.58)$) {$Q_3$};
\node[note] at ($(qone)+(1.88,0.20)$) {$v_{Q_1}$};
\node[note] at ($(qtwo)+(1.88,0.20)$) {$v_{Q_2}$};
\node[note] at ($(qthree)+(1.88,0.20)$) {$v_{Q_3}$};
\node[note] at (-2.15,6.00) {$S_{v_P}=\emptyset$};
\node[note] at (5.10,7.25) {2-vertex\\connected};
\node[note] at (7.12,4.10) {$Z$};

\draw[incidence] (qonep1) -- ++(-0.58,1.62);
\draw[incidence] (qonep5) -- ++(0.58,1.62);
\node[note] at (-4.35,4.15) {$F_{v_{Q_1}}$};
\draw[incidence] (qtwop2) -- ++(-0.58,1.62);
\draw[incidence] (qtwop5) -- ++(0.58,1.62);
\node[note] at (0,3.40) {$F_{v_{Q_2}}$};
\draw[incidence] (qthreep1) -- ++(-0.58,1.62);
\draw[incidence] (qthreep6) -- ++(0.58,1.62);
\node[note] at (4.35,4.15) {$F_{v_{Q_3}}$};

\coordinate (qoneL) at ($(qone)+(-0.42,-1.28)$);
\coordinate (qoneR) at ($(qone)+(0.42,-1.28)$);
\coordinate (a1) at (-5.05,0.25);
\coordinate (b1) at (-3.55,0.25);
\coordinate (c1) at (-4.65,-0.65);
\coordinate (d1) at (-3.15,-0.65);
\draw[setbox] (-5.28,-0.95) rectangle (-2.92,0.62);
\draw[block edge] (qoneL)--(a1)--(b1)--(qoneR);
\draw[block edge] (b1)--(c1)--(d1)--(b1);
\foreach \v/\lab in {a1/{},b1/{},c1/{},d1/{}}
  \node[s vertex] at (\v) {};
\node[note] at (-4.10,-1.30) {$S_{v_{Q_1}}$};

\coordinate (qtwoL) at ($(qtwo)+(-0.42,-1.28)$);
\coordinate (qtwoR) at ($(qtwo)+(0.42,-1.28)$);
\coordinate (a2) at (-1.10,0.10);
\coordinate (b2) at (1.10,0.10);
\coordinate (c2) at (0.75,-0.80);
\coordinate (d2) at (-0.75,-0.80);
\draw[setbox] (-1.30,-1.10) rectangle (1.30,0.45);
\draw[block edge] (qtwoL)--(a2)--(b2)--(qtwoR);
\draw[block edge] (a2)--(d2)--(c2)--(b2);
\foreach \v in {a2,b2,c2,d2}
  \node[s vertex] at (\v) {};
\node[note] at (0,-1.45) {$S_{v_{Q_2}}$};

\coordinate (qthreeL) at ($(qthree)+(-0.42,-1.28)$);
\coordinate (qthreeR) at ($(qthree)+(0.42,-1.28)$);
\coordinate (a3) at (3.55,0.25);
\coordinate (b3) at (5.05,0.25);
\coordinate (c3) at (4.35,-0.70);
\draw[setbox] (3.25,-1.00) rectangle (5.28,0.62);
\draw[block edge] (qthreeL)--(a3)--(b3)--(qthreeR);
\draw[block edge] (a3)--(c3)--(b3);
\foreach \v in {a3,b3,c3}
  \node[s vertex] at (\v) {};
\node[note] at (4.35,-1.30) {$S_{v_{Q_3}}$};
\end{tikzpicture}
\caption{The structure behind the oracle call in \Cref{lem:vertex_a_from_b}. After guessing $Z=\{v_P,v_{Q_1},v_{Q_2},v_{Q_3}\}$, we ask the \textsc{2VCSS-BP} oracle to find a $2$-vertex-connected solution on the vertices of $Z$. In the figure this solution is drawn as the large gray box, and it chooses the incident pairs $F_{v_P}$ and $F_{v_{Q_i}}$ shown by the dangling edges; these edges are not truly dangling, but the figure only indicates the chosen incidences. The blue vertices are the vertices that must be covered by the component on $Z$. The green vertices are already covered outside $Z$, or, for $P$, will be covered outside $Z$ after the oracle call; the black vertices are path endpoints. The red edges indicate the underlying paths. The vertices enclosed below each $v_{Q_i}$ are $S_{v_{Q_i}}$, which together with $v_{Q_i}$ form the smaller subproblem counted by the corresponding $b$-value.}
\label{fig:vertex-a-from-b-blackbox}
\end{figure}

\begin{lemma}\label{lem:vertex_a_from_b}
Assume $|P|>1$ and $|X|\geq 2$.
For each $Z\subseteq X$ with $v_P\in Z$ and $|Z|\geq 2$, let $a'_{X,v_P,Z,i,j}\subseteq\Omega$ be the entry returned by the \textsc{2VCSS-BP} oracle on $H^\star[Z]$ with edge weights inherited from $H$, boundary-pair choice sets $\mathcal C^{P,Z,i,j}_{v_Q}$, and $U:=X\setminus Z$.
Then
\[
    a_{X,v_P,i,j}
    =
    \bigcup_{\substack{Z\subseteq X\\ v_P\in Z,\ |Z|\geq 2}}
    a'_{X,v_P,Z,i,j}.
\]
\end{lemma}
\begin{proof}
Fix $Z$ with $v_P\in Z$ and $|Z|\geq 2$.
We first show that every edge set counted by $a_{X,v_P,i,j}$ whose unique block containing $v_P$ has vertex set $Z$ gives a feasible oracle solution with no larger weight vector.
Take such an edge set $A$, and let $B$ be its unique block containing $v_P$.
Since $|Z|\geq 2$, the block $B$ is not a self-loop block, and hence $E(B)\subseteq E(H^\star[Z])$.
Use $E(B)$ as the selected edge set of the oracle on $H^\star[Z]$.
Since $(X,A)$ is bridgeless, \(B\) cannot be a block consisting of a single non-loop edge; such an edge would be a bridge.
Hence every vertex of $B$ is incident with at least two edges of $B$.
For each $v_Q\in Z$, define $F_{v_Q}$ by choosing two incident edges of $B$ as follows.
If $|Q|>1$, choose edges whose preimage endpoints on $Q$ have minimum and maximum indices among all edges of $B$ incident with $v_Q$.
If $|Q|=1$, choose any two incident edges.

For each $v_Q\in Z$, let $S_{v_Q}$ be the set of vertices of $X\setminus Z$ lying in components of $(X,A)-Z$ adjacent to $v_Q$.
Each component of $(X,A)-Z$ is adjacent to exactly one vertex of $Z$.
Thus the sets $S_{v_Q}$ partition $X\setminus Z$.

We now verify that, for each $v_Q\in Z$, the data $F_{v_Q}$ and $S_{v_Q}$ extend to a member of $\mathcal C^{P,Z,i,j}_{v_Q}$.
For $Q=P$, condition (V3) at $P$ gives $\ell^P_{F_{v_P}}<r^P_{F_{v_P}}$.
If $2\leq i\leq j\leq |P|-1$, the covering requirement at $P$ and the choice of the minimum and maximum indices give $\ell^P_{F_{v_P}}<i\leq j<r^P_{F_{v_P}}$; if $(i,j)=(1,1)$, there is no covering requirement at $P$.
Because the block containing $v_P$ is unique, $S_{v_P}=\emptyset$.
Thus $\mathcal C^{P,Z,i,j}_{v_P}$ contains a choice with pair $F_{v_P}$, set $S_{v_P}$, and weight vector $(0,0)$.
For a vertex $v_Q\in Z$ with $Q\neq P$ and $|Q|=1$, (V3) allows at most one block containing $v_Q$; since $B$ is already such a block, no component can be attached at $v_Q$, and hence $S_{v_Q}=\emptyset$.
Thus $\mathcal C^{P,Z,i,j}_{v_Q}$ contains a choice with pair $F_{v_Q}$, set $S_{v_Q}$, and weight vector $(0,0)$.
For a vertex $v_Q\in Z$ with $Q\neq P$ and $|Q|>1$, (V3) applied to $B$ at $Q$ gives two incident edges of $B$ whose preimage endpoints on $Q$ have distinct indices, and hence $\ell^Q_{F_{v_Q}}<r^Q_{F_{v_Q}}$.
By the choice of the minimum and maximum indices, $B$ covers every index $t\in\{\ell^Q_{F_{v_Q}}+1,\dots,r^Q_{F_{v_Q}}-1\}$ at $Q$.
Every index $t$ of $Q$ with $2\leq t\leq \ell^Q_{F_{v_Q}}$ or $r^Q_{F_{v_Q}}\leq t\leq |Q|-1$ must therefore be covered by the attached subgraph on $S_{v_Q}\cup\{v_Q\}$.
Thus this attached subgraph is counted by $b_{S_{v_Q}\cup\{v_Q\},v_Q,\ell^Q_{F_{v_Q}}+1,r^Q_{F_{v_Q}}-1}$.
Choose the triple $(F_{v_Q},S_{v_Q},\mathbf q_{v_Q})\in\mathcal C^{P,Z,i,j}_{v_Q}$ whose weight vector is the pair given by the attached subgraph.
Summing $\mathbf w(E(B))$ and the chosen weight vectors gives a pair bounded componentwise by $\mathbf w(A)$.
Hence every pair contributed by $A$ to $a_{X,v_P,i,j}$ also belongs to $a'_{X,v_P,Z,i,j}$.

Conversely, take a pair $\mathbf q\in a'_{X,v_P,Z,i,j}$ and a feasible oracle solution on $H^\star[Z]$ witnessing $\mathbf q$.
Let $A_Z$ be the oracle edge set on $Z$.
For $v_Q\in Z$, let $(F_{v_Q},S_{v_Q},\mathbf s_{v_Q})\in\mathcal C^{P,Z,i,j}_{v_Q}$ be selected at $v_Q$.
For each $v_Q\in Z$ with $Q\neq P$ and $|Q|>1$, choose an edge set counted by the entry $b_{S_{v_Q}\cup\{v_Q\},v_Q,\ell^Q_{F_{v_Q}}+1,r^Q_{F_{v_Q}}-1}$ and glue it at $v_Q$.
Let $A$ be the union of $A_Z$ and the glued edge sets.
We show that $A$ is counted by $a_{X,v_P,i,j}$.
For vertices $v_Q$ with $Q\neq P$ and $|Q|=1$, the definition of $\mathcal C^{P,Z,i,j}_{v_Q}$ forces $S_{v_Q}=\emptyset$.
For $v_P$, the definition of $\mathcal C^{P,Z,i,j}_{v_P}$ forces $S_{v_P}=\emptyset$ and $\ell^P_{F_{v_P}}<r^P_{F_{v_P}}$.
If $2\leq i\leq j\leq |P|-1$, it also makes the block on $Z$ cover every index $t\in\{i,\dots,j\}$; for $(i,j)=(1,1)$, it imposes no covering requirement.
The oracle edge set is $2$-vertex-connected on $Z$, so it forms the unique block containing $v_P$ after the glued subgraphs are attached at single vertices.
For each $Q\neq P$ with $|Q|>1$, the chosen edges at $v_Q$ make the block on $Z$ cover every index $t\in\{\ell^Q_{F_{v_Q}}+1,\dots,r^Q_{F_{v_Q}}-1\}$.
The glued $b$-solution covers every index $t$ of $Q$ with $2\leq t\leq \ell^Q_{F_{v_Q}}$ or $r^Q_{F_{v_Q}}\leq t\leq |Q|-1$, and it also ensures the condition in (V3) for the other blocks containing $v_Q$.
For $|Q|=1$, no additional block containing $v_Q$ is glued.
Therefore $A$ is counted by $a_{X,v_P,i,j}$.
For each glued $b$-solution at $v_Q$, its weight vector is componentwise at most the selected vector $\mathbf s_{v_Q}$, because $\mathbf s_{v_Q}$ belongs to the corresponding $b$-entry.
Hence $\mathbf w(A)\leq \mathbf w(A_Z)+\sum_{v_Q\in Z}\mathbf s_{v_Q}\leq \mathbf q$.
Thus $\mathbf q\in a_{X,v_P,i,j}$, and so $a'_{X,v_P,Z,i,j}\subseteq a_{X,v_P,i,j}$.
\end{proof}

\paragraph*{Runtime analysis.}

By \Cref{lem:vertex_a_singlevertex,lem:vertex_b_from_a,lem:vertex_a_from_b} and a deterministic algorithm for \textsc{2VCSS-BP}, the dynamic programming tables can already be computed by calling \textsc{2VCSS-BP} separately for each candidate set \(X\) and each guessed set \(Z\).
Indeed, for a fixed \(Z\subseteq X\), the \textsc{2VCSS-BP} instance on \(H^\star[Z]\) with universe \(X\setminus Z\) gives the candidate entry \(a'_{X,v_P,Z,i,j}\), and taking the union over \(Z\) gives \(a_{X,v_P,i,j}\).
For a \textsc{2VCSS-BP} instance with universe \(U_0\), its \emph{all-subuniverse version} asks for the output entry of every variant indexed by \(U\subseteq U_0\), where the variant for \(U\) replaces (E2) by the condition that \(\{S_v\}_{v\in V(G)}\) is a partition of \(U\).
We further accelerate the dynamic programming by fixing \(Z\) first and using such an all-subuniverse algorithm; this returns all entries \(a'_{X,v_P,Z,i,j}\) with \(X\supseteq Z\) at once.

\begin{lemma}\label{lem:runtime_twovca_via_oracle}
Suppose that the all-subuniverse version of \textsc{2VCSS-BP} described above can be solved in $O^*(\rho^{|V(G)|}\nu^{|U_0|}W)$ time for every instance with graph $G$, universe $U_0$, and maximum edge cost $W$, suppressing polynomial factors in $k$.
Then \textsc{PFA} with $p$ paths and maximum link cost $W$ can be solved in $O^*(\max\{3,\rho+\nu\}^pW)$ time.
Consequently, \textsc{2VCA} with parameter $k$ and maximum link cost $W$ can be solved in $O^*(\max\{3,\rho+\nu\}^{2k}W)$ time.
\end{lemma}
\begin{proof}
We compute the values in increasing order of $|X|$.
For $|X|=1$, the $a$-values are initialized by \Cref{lem:vertex_a_singlevertex}.
For $h\geq 1$, assume that all $b$-values on vertex sets of size at most $h$ have been computed, and consider the computation of the $a$-values on vertex sets of size $h+1$.
Fix a candidate set $Z\subseteq V(H)$ with $|Z|\geq 2$, a root $v_P\in Z$, and interval indices $i,j$.
As explained above, one \textsc{2VCSS-BP} computation on $H^\star[Z]$ with universe $V(H)\setminus Z$ gives all entries $a'_{X,v_P,Z,i,j}$ with $|X|=h+1$ and $Z\subseteq X$.
By the assumed all-subuniverse oracle, this costs $O^*(\rho^{|Z|}\nu^{|V(H)\setminus Z|}W)$ time.
Summing over all choices of $Z$ gives the exponential part $\sum_{Z\subseteq V(H)}\rho^{|Z|}\nu^{|V(H)\setminus Z|}=(\rho+\nu)^{|V(H)|}$; the additional choices of $h$, the root, and the interval indices add only a polynomial factor.
After these oracle computations, \Cref{lem:vertex_a_from_b} gives each $a_{X,v_P,i,j}$ by taking the union over the already computed candidates $a'_{X,v_P,Z,i,j}$.
After the $a$-entries for $X$ are known, \Cref{lem:vertex_b_from_a} computes the $b$-entries for $X$ in $O^*(2^{|X|}W)$ time.
Over all $X\subseteq V(H)$, these $b$-computations take $O^*(\sum_{X\subseteq V(H)}2^{|X|}W)=O^*(3^{|V(H)|}W)\leq O^*(3^pW)$ time.
The oracle calls take $O^*((\rho+\nu)^{|V(H)|}W)\leq O^*((\rho+\nu)^pW)$ time.
Together these bounds give $O^*(\max\{3,\rho+\nu\}^pW)$ time for \textsc{PFA}.
The statement for \textsc{2VCA} follows from \Cref{lem:twovca_to_twovca_pf}.
\end{proof}

\section{Reducing \textsc{2VCSS-BP} to a Pseudo-Ear-Packing Polynomial}\label{sec:2VCSSBP}

We first reduce \textsc{2VCSS-BP} to evaluating a polynomial whose terms encode $2$-vertex-connected pseudo-ear packings.
\Cref{sec:2VCSSBP-evaluation} evaluates this polynomial in single-exponential time.
We begin by slightly rephrasing \textsc{2VCSS-BP}.

\paragraph*{Prescribed subuniverses.}
We actually construct an algorithm that, for all prescribed subsets $U\subseteq U_0$, computes the corresponding output entries simultaneously; this stronger form is used in \Cref{lem:runtime_twovca_via_oracle}.
The ordinary \textsc{2VCSS-BP} instance corresponds to the prescribed subset $U=U_0$.



\paragraph*{Modifying the definition of $A$ to delete disjointness constraint.}

By replacing $A$ by $A\setminus \bigcup_{v\in V(G)}F_v$, we may first rephrase the problem as finding an edge set $A$ and a triple $(F_v,S_v,\mathbf s_v)\in\mathcal C_v$ for each $v\in V(G)$ with
\begin{enumerate}
    \item[(E1)] $\bigcup_{v\in V(G)}F_v\cap A=\emptyset$,
    \item[(E2)] $\{S_v\}_{v\in V(G)}$ is a partition of $U_0$, and
    \item[(E3)] $\left(V,\bigcup_{v\in V(G)}F_v\cup A\right)$ is $2$-vertex-connected,
\end{enumerate}
and with the weight obtained from the edge set $A\cup\bigcup_vF_v$ and from the vectors $\mathbf s_v$.
We can safely remove the constraint (E1): once the pairs $F_v$ are fixed, selecting an edge of $\bigcup_{v\in V(G)}F_v$ also in $A$ does not help connectivity and can only increase the weight.

\paragraph*{Gluing $F$ into walks and cycles.}

A graph $D$ is a \emph{pseudo-ear} if
\begin{itemize}
    \item $D=(v_0,e_0,v_1,\dots, v_{\ell},e_{\ell},v_{\ell+1})$ is a walk with $\ell\geq 1$ such that its inner vertices $(v_1,\dots, v_{\ell})$ are pairwise distinct, $e_0\neq e_1$, and $e_{\ell-1}\neq e_{\ell}$, or
    \item $D$ is a cycle.
\end{itemize}
Note that, in the walk case, the endpoints need not be distinct from the inner vertices.
We also allow $e_0=e_{\ell}$.
For a walk pseudo-ear $D$ as above, recall that $V_{\mathrm{in}}(D)=\{v_1,\dots,v_{\ell}\}$; define $F_D(v_i):=\{e_{i-1},e_i\}$ for each $i\in\{1,\dots,\ell\}$.
Thus, for a walk pseudo-ear, $F_D(v_i)$ may differ from $\delta_D(v_i)$; this can occur when $v_i$ also appears as an endpoint of the walk, because $\delta_D(v_i)$ also sees the endpoint incidence.
For example, in \Cref{fig:choice-f-as-pseudo-ears}, $D_3=(f,g,h,a,g)$, and hence $F_{D_3}(g)$ consists of the edges $fg$ and $gh$, while $\delta_{D_3}(g)$ also contains the edge $ag$.
For a cycle $D$, recall that $V_{\mathrm{in}}(D)=V(D)$; define $F_D(v):=\delta_D(v)$ for every $v\in V(D)$.
A set $\mathcal{D}$ of pseudo-ears is a \emph{pseudo-ear packing}.
It \emph{spans} a vertex set $X$ if $\biguplus_{D\in \mathcal{D}}V_{\mathrm{in}}(D)$ is a partition of $X$; in this case, we call $\mathcal D$ a \emph{spanning pseudo-ear packing of $X$}.

For a choice of $F$, a spanning pseudo-ear packing $\mathcal{D}$ of $V(G)$ is \emph{consistent} with $F$ if, for each $v\in V(G)$, $F_D(v)=F_v$ for the unique $D\in \mathcal{D}$ with $v\in V_{\mathrm{in}}(D)$.
A choice of $F$ naturally defines such a packing as follows.
Each $F_v=\{e,f\}$ can naturally be regarded as a $3$-vertex walk pseudo-ear with inner vertex $v$ and edge set $\{e,f\}$.
We glue these walk pseudo-ears along shared end edges as follows.
If two current walks can be oriented as $D=(x_0,e_0,\dots,x_p,e,x_{p+1})$ and $D'=(x_p,e,x_{p+1},e'_1,y_1,\dots,e'_q,y_q)$, then gluing them along $e$ gives the walk $(x_0,e_0,\dots,x_p,e,x_{p+1},e'_1,y_1,\dots,e'_q,y_q)$.
This is again a walk pseudo-ear, since we glue along an end edge and the two current walks have disjoint sets of inner vertices.
If the two occurrences of the shared end edge belong to the same current walk, the same operation closes the walk into a cycle pseudo-ear.
Gluing exhaustively gives a set $\mathcal{D}$ of maximal walks and cycles.
By construction, $\mathcal{D}$ is a spanning pseudo-ear packing of $V(G)$ with $F_D(v)=F_v$ for all $v\in V_{\mathrm{in}}(D)$.
Moreover, every edge appears in at most one member of $\mathcal{D}$, and for this maximal packing we have
\[
    \biguplus_{D\in \mathcal{D}}E(D)=\bigcup_{v\in V(G)}F_v.
\]
See \Cref{fig:f-to-pseudo-ear-decomposition} for an example of this gluing.

\begin{figure}[t]
\centering
\begin{subfigure}[t]{0.47\linewidth}
\centering
\begin{tikzpicture}[
  x=0.95cm,y=0.95cm,
  >={Latex[length=1.8mm,width=1.2mm]},
  vertex/.style={circle,fill=black,inner sep=0pt,minimum size=2mm},
  vertex label/.style={font=\scriptsize,inner sep=1pt},
  chosen/.style={
    line width=0.6pt,
    postaction={decorate},
    decoration={markings,mark=at position 0.58 with {\arrow{Latex[length=1.8mm,width=1.2mm]}}}
  },
  note/.style={font=\scriptsize},
]
\coordinate (a) at (0,0);
\coordinate (b) at (1.1,0);
\coordinate (c) at (2.2,0);
\coordinate (d) at (3.3,0);
\coordinate (e) at (3.3,-1.15);
\coordinate (f) at (4.4,-1.15);
\coordinate (g) at (2.35,-2.1);
\coordinate (h) at (0.45,-1.85);
\coordinate (i) at (4.45,0.9);

\draw[chosen] (b) -- (a);
\draw[chosen] (b) to[bend left=10] (c);
\draw[chosen] (c) to[bend left=10] (b);
\draw[chosen] (c) -- (d);
\draw[chosen] (e) -- (d);
\draw[chosen] (e) -- (f);
\draw[chosen] (g) -- (f);
\draw[chosen] (g) to[bend left=10] (h);
\draw[chosen] (h) to[bend left=10] (g);
\draw[chosen] (h) to[bend left=10] (a);
\draw[chosen] (a) to[bend left=10] (h);
\draw[chosen] (a) -- (g);
\draw[chosen] (d) to[bend left=10] (i);
\draw[chosen] (i) to[bend left=10] (d);
\draw[chosen] (i) to[bend left=10] (f);
\draw[chosen] (f) to[bend left=10] (i);
\draw[chosen] (f) to[bend left=10] (d);
\draw[chosen] (d) to[bend left=10] (f);

\foreach \v in {a,b,c,d,e,f,g,h,i}
  \node[vertex] at (\v) {};
\node[vertex label,below=2pt] at (a) {$a$};
\node[vertex label,below=2pt] at (b) {$b$};
\node[vertex label,below=2pt] at (c) {$c$};
\node[vertex label,above left=1pt] at (d) {$d$};
\node[vertex label,left=2pt] at (e) {$e$};
\node[vertex label,right=2pt] at (f) {$f$};
\node[vertex label,below=2pt] at (g) {$g$};
\node[vertex label,left=2pt] at (h) {$h$};
\node[vertex label,right=2pt] at (i) {$i$};
\end{tikzpicture}
\caption{}
\label{fig:choice-f-as-directed-edges}
\end{subfigure}\hfill
\begin{subfigure}[t]{0.47\linewidth}
\centering
\begin{tikzpicture}[
  x=0.95cm,y=0.95cm,
  >={Latex[length=1.8mm,width=1.2mm]},
  vertex/.style={circle,fill=black,inner sep=0pt,minimum size=2mm},
  inner one/.style={circle,fill=blue!65!black,inner sep=0pt,minimum size=2mm},
  inner two/.style={circle,fill=teal!70!black,inner sep=0pt,minimum size=2mm},
  inner three/.style={circle,fill=violet!80!black,inner sep=0pt,minimum size=2mm},
  inner four/.style={circle,fill=orange!80!black,inner sep=0pt,minimum size=2mm},
  vertex label/.style={font=\scriptsize,inner sep=1pt},
  walk edge/.style={
    line width=1pt,blue!65!black
  },
  walk edge two/.style={
    line width=1pt,teal!70!black
  },
  walk edge three/.style={
    line width=1pt,violet!80!black
  },
  cycle edge/.style={
    line width=1pt,orange!80!black
  },
  note/.style={font=\scriptsize},
]
\coordinate (a) at (0,0);
\coordinate (b) at (1.1,0);
\coordinate (c) at (2.2,0);
\coordinate (d) at (3.3,0);
\coordinate (e) at (3.3,-1.15);
\coordinate (f) at (4.4,-1.15);
\coordinate (g) at (2.35,-2.1);
\coordinate (h) at (0.45,-1.85);
\coordinate (i) at (4.45,0.9);

\draw[walk edge] (a) -- (b);
\draw[walk edge] (b) -- (c);
\draw[walk edge] (c) -- (d);
\draw[walk edge two] (d) -- (e);
\draw[walk edge two] (e) -- (f);
\draw[walk edge three] (f) -- (g);
\draw[walk edge three] (g) -- (h);
\draw[walk edge three] (h) -- (a);
\draw[walk edge three] (a) -- (g);
\draw[cycle edge] (d) -- (i);
\draw[cycle edge] (i) -- (f);
\draw[cycle edge] (f) -- (d);

\foreach \v in {b,c}
  \node[inner one] at (\v) {};
\node[inner two] at (e) {};
\foreach \v in {a,g,h}
  \node[inner three] at (\v) {};
\foreach \v in {d,f,i}
  \node[inner four] at (\v) {};
\node[vertex label,below=2pt] at (a) {$a$};
\node[vertex label,below=2pt] at (b) {$b$};
\node[vertex label,below=2pt] at (c) {$c$};
\node[vertex label,above left=1pt] at (d) {$d$};
\node[vertex label,left=2pt] at (e) {$e$};
\node[vertex label,right=2pt] at (f) {$f$};
\node[vertex label,below=2pt] at (g) {$g$};
\node[vertex label,left=2pt] at (h) {$h$};
\node[vertex label,right=2pt] at (i) {$i$};
\node[note,blue!65!black] at (1.65,0.42) {$D_1$};
\node[note,teal!70!black] at (2.88,-0.62) {$D_2$};
\node[note,violet!80!black] at (1.25,-2.45) {$D_3$};
\node[note,orange!80!black] at (4.95,0.45) {$D_4$};
\end{tikzpicture}
\caption{}
\label{fig:choice-f-as-pseudo-ears}
\end{subfigure}
\caption{How a choice of $F$ induces pseudo-ears. (a) We display the $F_v$ part of a \textsc{2VCSS-BP} solution: an arrow leaving $v$ means that the corresponding edge belongs to $F_v$, so every vertex has exactly two outgoing edges; if both endpoints select the same edge, we draw the edge twice with opposite directions. (b) Gluing along such bidirected edges gives four pseudo-ears: $D_1$ is a walk from $a$ to $d$, $D_2$ is a walk from $d$ to $f$, $D_3=(f,g,h,a,g)$ is a walk from $f$ to $g$ whose endpoint $g$ is also an inner vertex, and $D_4$ is a cycle on $\{d,i,f\}$. The inner vertices are colored according to the pseudo-ear containing them.}
\label{fig:f-to-pseudo-ear-decomposition}
\end{figure}
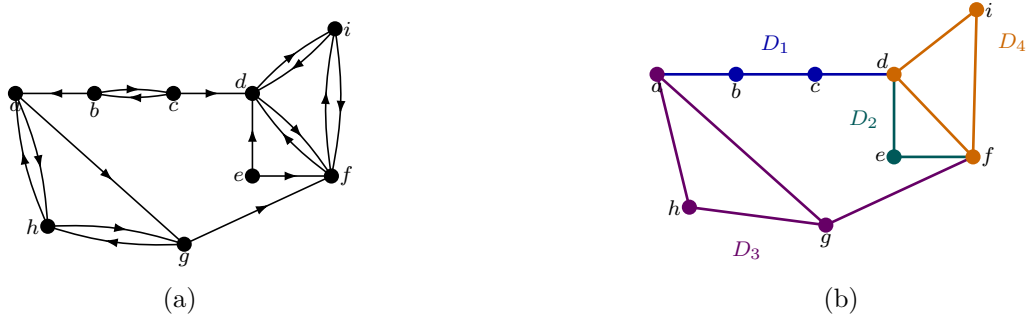

Thus, instead of selecting $F$, we can select $\mathcal{D}$.
Here, we can safely ignore the condition that each $D\in \mathcal{D}$ is maximal.
Since the selected packing need not be obtained by the exhaustive gluing above, the same edge may occur more than once in $\bigcup_{D\in\mathcal D}E(D)$.
Such an edge is counted each time it appears in pseudo-ears.
In particular, if an edge $e$ appears as end edges of two walk pseudo-ears that may be the same and remains unglued, then the weight of $e$ is counted twice in the total weight.


\paragraph*{Explicit new formulation.}

Now, the new formulation is:
\boxproblemnoparam{$2$-Vertex-Connected Spanning Subgraph with Boundary-Pair Costs, Modified Form (2VCSS-BP-M)}{
    a multigraph $G$ with no self-loops and $|V(G)|\geq 2$,
    integers $k$ and $W\geq 0$,
    an edge weight vector $\mathbf w_e=(\chi_e,c_e)\in \{0,1\}\times\{0,\dots,W\}$ for each $e\in E(G)$,
    a finite set $U_0$, and
    a finite set $\mathcal C_v$ of triples $(F,S,\mathbf s)$ with $F\in\binom{\delta_G(v)}{2}$, $S\subseteq U_0$, and $\mathbf s\in\{0,\dots,k\}\times\{0,\dots,kW\}$, for each $v\in V(G)$.
}{
    For each $U\subseteq U_0$, do the following:

    Select $A\subseteq E(G)$.
    Furthermore, select a spanning pseudo-ear packing $\mathcal{D}$ of $V(G)$.
    For each $D\in\mathcal D$ and $v\in V_{\mathrm{in}}(D)$, select $(S_{D,v},\mathbf s_{D,v})$ such that $(F_D(v),S_{D,v},\mathbf s_{D,v})\in\mathcal C_v$.
    Among $(A,\mathcal{D},(S_{D,v},\mathbf s_{D,v}))$ with
    \begin{enumerate}
        \item[(E2)] $\{S_{D,v}\}_{D\in\mathcal D,\,v\in V_{\mathrm{in}}(D)}$ is a partition of $U$, and
        \item[(E3)] $\left(V(G),\bigcup_{D\in \mathcal{D}}E(D)\cup A\right)$ is $2$-vertex-connected,
    \end{enumerate}
    Output the set of all pairs $(\alpha,\beta)\in\Omega$ for which some such $(A,\mathcal D,(S_{D,v},\mathbf s_{D,v}))$ satisfies
    $\sum_{e\in A}\mathbf w_e+\sum_{D\in\mathcal D}\left(\sum_{e\in E(D)}\mathbf w_e+\sum_{v\in V_{\mathrm{in}}(D)}\mathbf s_{D,v}\right)\leq(\alpha,\beta)$.
}

In computing the total weight, $E(D)$ is read as the multiset of edge occurrences along $D$.
Thus, if the same edge $e$ appears twice in a single pseudo-ear $D$, then $\mathbf w_e$ is added twice.

\begin{lemma}\label{lem:bp_to_bpm}
Fix $U\subseteq U_0$.
The output entry of the prescribed-$U$ variant of \textsc{2VCSS-BP} equals the output entry of the prescribed-$U$ variant of \textsc{2VCSS-BP-M}.
\end{lemma}
\begin{proof}
We prove that every feasible \textsc{2VCSS-BP} solution maps to a feasible \textsc{2VCSS-BP-M} solution of the same weight, and that every \textsc{2VCSS-BP-M} solution maps back to a \textsc{2VCSS-BP} solution of componentwise no larger weight.
Let $(A,(F_v,S_v,\mathbf s_v)_{v\in V(G)})$ be a feasible \textsc{2VCSS-BP} solution, and let $\mathcal D$ be the maximal pseudo-ear packing induced by the pairs $F_v$.
For the unique $D\in\mathcal D$ with $v\in V_{\mathrm{in}}(D)$, set $S_{D,v}:=S_v$ and $\mathbf s_{D,v}:=\mathbf s_v$.
Then $(A\setminus\bigcup_vF_v,\mathcal D,(S_{D,v},\mathbf s_{D,v}))$ is a feasible \textsc{2VCSS-BP-M} solution.
The selected edge set and the pseudo-ear occurrences together count every edge of the original set $A$ exactly once, so the weight vector is unchanged.
Thus every pair in the original output entry also belongs to the modified output entry.

Conversely, from $(A,\mathcal D,(S_{D,v},\mathbf s_{D,v}))$, select the edge set $A\cup\bigcup_{D\in\mathcal D}E(D)$ and, for the unique $D\in\mathcal D$ with $v\in V_{\mathrm{in}}(D)$, select the triple $(F_D(v),S_{D,v},\mathbf s_{D,v})\in\mathcal C_v$.
This gives a feasible \textsc{2VCSS-BP} solution.
Its weight vector is at most that of the \textsc{2VCSS-BP-M} solution, because repeated edge occurrences in the pseudo-ears are counted only once in the selected edge set.
Thus, if the \textsc{2VCSS-BP-M} solution witnesses a pair \((\alpha,\beta)\), then the corresponding \textsc{2VCSS-BP} solution witnesses some pair \((\alpha',\beta')\leq(\alpha,\beta)\).
Since the original output entry is upward closed, \((\alpha,\beta)\) also belongs to the original output entry.
Hence the two output entries are equal.
\end{proof}


\paragraph*{The weight polynomial.}
For $U\subseteq U_0$, let $\mathcal F_U$ be the set of tuples $(A,\mathcal D,(S_{D,v},\mathbf s_{D,v})_{D,v})$ satisfying (E2) and (E3), where $(F_D(v),S_{D,v},\mathbf s_{D,v})\in\mathcal C_v$ for each $D\in\mathcal D$ and $v\in V_{\mathrm{in}}(D)$.
For such a tuple, define its weight vector by
\[
    \mathbf w(A,\mathcal D,(S_{D,v},\mathbf s_{D,v})_{D,v})
    :=
    \sum_{e\in A}\mathbf w_e+
    \sum_{D\in\mathcal D}\left(\sum_{e\in E(D)}\mathbf w_e+\sum_{v\in V_{\mathrm{in}}(D)}\mathbf s_{D,v}\right).
\]
Write $\mathbf w(A,\mathcal D,(S_{D,v},\mathbf s_{D,v})_{D,v})=(\chi(A,\mathcal D,(S_{D,v},\mathbf s_{D,v})_{D,v}),c(A,\mathcal D,(S_{D,v},\mathbf s_{D,v})_{D,v}))$.
Define the polynomial \(P_U(\xi,\zeta)\in\ZZ[\xi,\zeta]\) by
\[
    P_U(\xi,\zeta)
    :=
    \sum_{(A,\mathcal D,(S_{D,v},\mathbf s_{D,v})_{D,v})\in\mathcal F_U}
    \xi^{\chi(A,\mathcal D,(S_{D,v},\mathbf s_{D,v})_{D,v})}\zeta^{c(A,\mathcal D,(S_{D,v},\mathbf s_{D,v})_{D,v})}.
\]
The prescribed-$U$ output entry consists of the pairs $(\alpha,\beta)\in\Omega$ for which $P_U(\xi,\zeta)$ has a nonzero monomial $\xi^{\alpha'}\zeta^{\beta'}$ with $(\alpha',\beta')\leq(\alpha,\beta)$.
The $\xi$-degree of $P_U(\xi,\zeta)$ is bounded by a polynomial in the input size and $k$, and the $\zeta$-degree is bounded by that polynomial times $W$.
Thus polynomially many evaluations in $\xi$ and $O^*(W)$ evaluations in $\zeta$ recover $P_U(\xi,\zeta)$ by interpolation.
Afterwards, we inspect only the coefficients $\xi^{\alpha}\zeta^{\beta}$ with $(\alpha,\beta)\in\Omega$.

\paragraph*{Removing the \(S\)-partition constraint.}
We first eliminate the constraint that the sets $S_v$ form a partition of the prescribed universe.
Let $\mathcal M$ be the set of pairs $(A,\mathcal D)$ such that $A\subseteq E(G)$ and $\mathcal D$ is a spanning pseudo-ear packing of $V(G)$.
For $(A,\mathcal D)\in\mathcal M$, let
\[
    H_{A,\mathcal D}:=\left(V(G),A\cup\bigcup_{D\in\mathcal D}E(D)\right).
\]
Let $\mathcal M^{\mathrm{2vcon}}:=\{(A,\mathcal D)\in\mathcal M:H_{A,\mathcal D}\text{ is }2\text{-vertex-connected}\}$.
For edge weights $\kappa_e$ and boundary-pair weights $\eta_{v,F}$, define
\[
    \mathrm{wt}_{A,\mathcal D}(\kappa,\eta):=
    \prod_{e\in A}\kappa_e\cdot
    \prod_{D\in\mathcal D}
    \left(
        \prod_{e\in E(D)}\kappa_e
        \prod_{v\in V_{\mathrm{in}}(D)}\eta_{v,F_D(v)}
    \right).
\]
Define
\[
    Q(\kappa,\eta)
    :=
    \sum_{(A,\mathcal D)\in\mathcal M^{\mathrm{2vcon}}}
    \mathrm{wt}_{A,\mathcal D}(\kappa,\eta).
\]

Introduce variables $\omega$ and $\gamma_p$ for $p\in U_0$, and let
\[
    \kappa_e:=\xi^{\chi_e}\zeta^{c_e}
    \quad\text{and}\quad
    \eta_{v,F}(\gamma):=\sum_{\substack{S\subseteq U_0,\ \mathbf s=(\alpha,\beta)\\ (F,S,\mathbf s)\in\mathcal C_v}}
    \omega^{|S|}\xi^{\alpha}\zeta^{\beta}\prod_{p\in S}\gamma_p .
\]
Thus the edge weight vector $\mathbf w_e$ is encoded in $\kappa_e$, while the boundary-pair weight vector $\mathbf s$ is encoded in the corresponding $\eta$-term.
For $Y\subseteq U_0$, let $\mathbf{1}_Y$ be the Boolean assignment that sends $\gamma_p$ to $1$ if $p\in Y$ and to $0$ otherwise.
We have the following inversion.

\begin{lemma}\label{lem:app-cost-from-q}
For every $U\subseteq U_0$, we have
\[
    P_U(\xi,\zeta)
    =
    [\omega^{|U|}]
    \sum_{Y\subseteq U}(-1)^{|U|-|Y|}
    Q(\kappa,\eta(\mathbf{1}_{Y})).
\]
\end{lemma}
\begin{proof}
Fix $U\subseteq U_0$.
Expand the right-hand side.
For fixed choices $(A,\mathcal D,(S_{D,v},\mathbf s_{D,v})_{D,v})$ satisfying (E3), the term corresponding to $Y\subseteq U$ survives in $Q(\kappa,\eta(\mathbf{1}_Y))$ exactly when $\bigcup_{D,v} S_{D,v}\subseteq Y$.
Thus these fixed choices are multiplied by the sum of $(-1)^{|U|-|Y|}$ over all $Y$ such that $\bigcup_{D,v} S_{D,v}\subseteq Y\subseteq U$.
By \Cref{lem:app-all-subuniverse-ie}, this sum is $1$ if $\bigcup_{D,v} S_{D,v}=U$ and is $0$ otherwise.
The coefficient of $\omega^{|U|}$ then keeps exactly the surviving choices satisfying $\sum_{D,v} |S_{D,v}|=|U|$.
Together with $\bigcup_{D,v} S_{D,v}=U$, this is equivalent to saying that $\{S_{D,v}\}_{D\in\mathcal D,\,v\in V_{\mathrm{in}}(D)}$ is a partition of $U$.
The vector $\mathbf s_{D,v}=(\alpha_{D,v},\beta_{D,v})$ is represented by the monomial $\xi^{\alpha_{D,v}}\zeta^{\beta_{D,v}}$ in the corresponding $\eta$-factor.
Thus the right-hand side sums $\xi^{\chi(A,\mathcal D,(S_{D,v},\mathbf s_{D,v})_{D,v})}\zeta^{c(A,\mathcal D,(S_{D,v},\mathbf s_{D,v})_{D,v})}$ exactly over tuples $(A,\mathcal D,(S_{D,v},\mathbf s_{D,v})_{D,v})\in\mathcal F_U$, and hence equals $P_U(\xi,\zeta)$.
\end{proof}

\Cref{lem:app-cost-from-q} implies the following.
For fixed values of $\xi$ and $\zeta$, once the values $Q(\kappa,\eta(\mathbf{1}_{Y}))$ are known for all $Y\subseteq U_0$, the right-hand sides in \Cref{lem:app-cost-from-q} for all $U\subseteq U_0$ can be computed by one M\"obius transform using \Cref{lem:app-all-subuniverse-ie}.
The coefficient of $\omega^{|U|}$ is obtained by interpolation in $\omega$, and interpolation in $\xi$ and $\zeta$ recovers all polynomials $P_U(\xi,\zeta)$ within the same $O^*(2^{|U_0|}W)$ evaluation bound, suppressing polynomial factors in the input size and in $k$.

\section{Evaluating the Pseudo-Ear-Packing Polynomial}\label{sec:2VCSSBP-evaluation}

It remains to evaluate $Q(\kappa,\eta)$ efficiently for fixed values of the edge weights $\kappa$ and boundary-pair weights $\eta$.
\Cref{subsec:evaluation-dp-onesum} reduces this task to evaluating polynomials \(R(\kappa,\tau)\), which are defined by the same kind of summation over edge choices and pseudo-ear packings as \(Q(\kappa,\eta)\), but with an ordinary connectedness constraint in place of the 2-vertex-connectivity constraint.
\Cref{subsec:evaluation-r-polynomials} computes these \(R\)-polynomials.
Finally, \Cref{subsec:evaluation-runtime} analyzes the running time.

\subsection{M\"obius inversion and dynamic programming on 1-sum trees}\label{subsec:evaluation-dp-onesum}
In this subsection, we reduce the evaluation of \(Q(\kappa,\eta)\) to the evaluation of \(R(\kappa,\tau)\).
We first replace the 2-vertex-connectivity requirement in \(Q\) by a M\"obius-weighted sum over 1-sum trees: in this sum, exactly the 2-vertex-connected terms survive, and all other terms cancel; see \Cref{lem:app-bag-tree-mobius-global}.
This sum is evaluated by dynamic programming with auxiliary polynomials \(Q^\circ\) and \(Q^\bullet\).
A transition for \(Q^\circ\) fixes the top bag \(Z\) and reduces the remaining task to evaluating an \(R\)-polynomial on \(Z\), where the ordinary connectedness constraint remains and the 1-sum-tree pieces attached outside \(Z\) are encoded by an attachment table \(\tau\).

\paragraph*{Treating 2-vertex connectivity by M\"obius inversion on 1-sum trees.}
A \emph{1-sum tree} \(T\) on a finite vertex set \(X\) is a bipartite tree whose nodes are the original vertices in \(X\) and a family \(\mathcal B_T\) of bag nodes.
Each bag node has degree at least two and is identified with its neighborhood in \(X\), called a bag.
For a connected graph \(H\) on \(X\), a 1-sum tree is \emph{consistent with \(H\)} if each bag is the vertex union of a nonempty collection of blocks of \(H\) whose union is connected in \(H\), and these collections partition the set of blocks of \(H\).
In particular, every edge of \(H\) is contained in a unique bag of \(T\), and \(H[B]\) is connected for every bag \(B\in\mathcal B_T\).
For an original vertex \(v\), let \(d_T(v)\) be its degree in \(T\), equivalently the number of bags containing \(v\).

Let $\mathcal M^{\mathrm{con}}$ be the set of pairs $(A,\mathcal D)\in\mathcal M$ for which $H_{A,\mathcal D}$ is connected.
Let $\mathcal T$ be the set of 1-sum trees on \(V(G)\) that are consistent with \(H_{A,\mathcal D}\) for at least one \((A,\mathcal D)\in\mathcal M^{\mathrm{con}}\), and let $\mathcal M_T^{\mathrm{con}}:=\{(A,\mathcal D)\in\mathcal M^{\mathrm{con}}:T\text{ is consistent with }H_{A,\mathcal D}\}$.

\begin{lemma}\label{lem:app-bag-tree-mobius-global}
We have
\[
    Q(\kappa,\eta)
    =
    \sum_{T\in\mathcal T}
    \left(
    \prod_{v\in V(G)}(-1)^{d_T(v)-1}(d_T(v)-1)!
    \cdot
    \sum_{(A,\mathcal D)\in\mathcal M_T^{\mathrm{con}}}
    \mathrm{wt}_{A,\mathcal D}(\kappa,\eta)
    \right).
\]
\end{lemma}
\begin{proof}
Fix $(A,\mathcal D)\in\mathcal M^{\mathrm{con}}$ and let $H:=H_{A,\mathcal D}$.
Let $\mathcal B_H$ be the set of blocks of $H$, and let $\mathcal B_H(v):=\{B\in\mathcal B_H:v\in B\}$.
We claim that 1-sum trees $T$ consistent with $H$ are in bijection with choices, for every $v\in V(G)$, of a partition of $\mathcal B_H(v)$.

First fix such partitions.
Start with the block-cut tree of $H$.
For each original vertex $v$ and each part $\mathcal P$ of the chosen partition of $\mathcal B_H(v)$, identify all block nodes in $\mathcal P$ and delete parallel edges created at $v$.
Identifying several neighbors of one vertex in a tree, followed by deleting parallel edges, keeps the graph connected and creates no cycle; repeating this operation for all $v$ therefore gives a tree.
The vertices obtained by identifying block nodes are bags whose vertex sets are the unions of the original blocks identified into them.
For each such bag, the original blocks identified into it are connected through shared original vertices, and hence their union is connected in \(H\).
These collections partition the set of blocks of \(H\), so the resulting tree is consistent with \(H\).

Conversely, let $T$ be a 1-sum tree consistent with $H$.
By consistency, the blocks of \(H\) are partitioned among the bags of \(T\).
For each vertex $v$, put two blocks of $\mathcal B_H(v)$ in the same part exactly when they are contained in the same bag of $T$.
This defines a partition of $\mathcal B_H(v)$.
It remains to see that the construction recovers each bag of $T$.
Fix a bag $B$ of $T$.
By consistency, the blocks contained in \(B\) form a connected collection.
Hence any two of these blocks are joined by a sequence of blocks contained in \(B\), where consecutive blocks share an original vertex.
At each shared vertex, the corresponding two blocks are put in the same part of the partition recovered from \(T\), so the vertex-wise identifications merge all blocks contained in \(B\) into one bag.
Blocks contained in distinct bags of \(T\) are never put in the same part at any original vertex.
Applying this to every bag of $T$, the construction recovers exactly the bags of $T$, and hence recovers $T$.
Under this bijection, $d_T(v)$ is the number of parts in the partition of $\mathcal B_H(v)$.
Thus the multiplier of the fixed weight term $\mathrm{wt}_{A,\mathcal D}(\kappa,\eta)$ in the right-hand side is
\[
    \prod_{v\in V(G)}
    \sum_{\pi\in\Pi(\mathcal B_H(v))}
    (-1)^{|\pi|-1}(|\pi|-1)! .
\]
By \Cref{lem:partition-lattice-mobius}, the inner sum is $1$ if $|\mathcal B_H(v)|=1$ and $0$ otherwise.
Hence this multiplier is $1$ exactly when \(H\) satisfies the \(2\)-vertex-connectivity condition in the definition of \(\mathcal M^{\mathrm{2vcon}}\), and is $0$ otherwise.
Thus the right-hand side keeps precisely the pairs in $\mathcal M^{\mathrm{2vcon}}$, with their original weights.
\end{proof}

\paragraph*{Dynamic programming on 1-sum trees.}
For dynamic programming, we root 1-sum trees at an original vertex.
Fix an original vertex $r$ of a 1-sum tree $T$, and view $T$ as rooted at $r$.
For an original vertex $v$ of $T$, let $d^+_{T,r}(v)$ be the number of connected components of $T-v$ that do not contain $r$.
Equivalently, $d^+_{T,r}(v)=d_T(v)-1$ for $v\neq r$, while $d^+_{T,r}(r)=d_T(r)$ unless $T$ is the one-vertex graph.
When the root is clear, we write $d_T^+(v)$.

Let $X\subseteq V(G)$, $u\in X$, and $F\in\binom{\delta_G(u)}{2}$.
We say that a spanning pseudo-ear packing $\mathcal D$ of $X$ is \emph{$(u,F)$-open} if $F_{D_u}(u)=F$ for the unique pseudo-ear $D_u$ with $u\in V_{\mathrm{in}}(D_u)$, and $F_D(v)\subseteq E(G[X])$ for every $v\in X\setminus\{u\}$.
Thus $F$ is the only boundary pair allowed to use edges outside $G[X]$.
\Cref{fig:uf-open-packing} illustrates this condition.
In the figure, exactly one edge of $F$ leaves $G[X]$; in general, zero, one, or two edges of $F$ may leave $G[X]$.

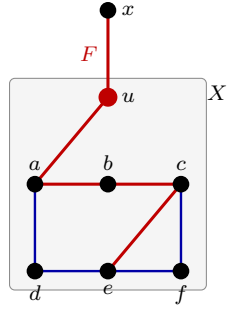
\begin{figure}[t]
\centering
\begin{tikzpicture}[
  x=0.84cm,y=1cm,
  vertex/.style={circle,fill=black,inner sep=0pt,minimum size=2.2mm},
  root/.style={circle,fill=red!75!black,inner sep=0pt,minimum size=2.5mm},
  vertex label/.style={font=\scriptsize,inner sep=1pt},
  note/.style={font=\scriptsize},
  root edge/.style={line width=1.15pt,red!75!black},
  walk edge/.style={line width=0.9pt,blue!65!black},
  box/.style={draw=black!45,fill=black!4,rounded corners=2pt},
]
\draw[box] (0.95,-0.85) rectangle (4.05,1.95);
\node[note] at (4.22,1.76) {$X$};

\coordinate (x) at (2.5,2.85);
\coordinate (u) at (2.5,1.7);
\coordinate (a) at (1.35,0.55);
\coordinate (b) at (2.5,0.55);
\coordinate (c) at (3.65,0.55);
\coordinate (d) at (1.35,-0.6);
\coordinate (e) at (2.5,-0.6);
\coordinate (f) at (3.65,-0.6);

\draw[root edge] (x) -- (u) -- (a) -- (b) -- (c) -- (e);
\draw[walk edge] (c) -- (f) -- (e) -- (d) -- (a);

\node[vertex] at (x) {};
\node[root] at (u) {};
\foreach \v in {a,b,c,d,e,f}
  \node[vertex] at (\v) {};

\node[vertex label,right=4pt] at (x) {$x$};
\node[vertex label,right=4pt] at (u) {$u$};
\node[vertex label,above=4pt] at (a) {$a$};
\node[vertex label,above=4pt] at (b) {$b$};
\node[vertex label,above=4pt] at (c) {$c$};
\node[vertex label,below=4pt] at (d) {$d$};
\node[vertex label,below=4pt] at (e) {$e$};
\node[vertex label,below=4pt] at (f) {$f$};

\node[note,red!75!black] at (2.2,2.28) {$F$};
\end{tikzpicture}
\caption{A $(u,F)$-open spanning pseudo-ear packing of $X$, where $F=F_{D_u}(u)=\{ux,ua\}$.}
\label{fig:uf-open-packing}
\end{figure}

For $A\subseteq E(G[X])$ and a $(u,F)$-open spanning pseudo-ear packing $\mathcal D$ of $X$, let
\[
    H_X(A,\mathcal D):=
    \left(X,A\cup\bigcup_{D\in\mathcal D}(E(D)\cap E(G[X]))\right).
\]
Let $\mathcal C^\circ_{X,u,F}$ be the set of triples $(A,\mathcal D,T)$ satisfying the following conditions.
\begin{itemize}
    \item $A\subseteq E(G[X])$ and $\mathcal D$ is a $(u,F)$-open spanning pseudo-ear packing of $X$.
    \item The graph $H_X(A,\mathcal D)$ is connected.
    \item $T$ is a 1-sum tree of $X$ that is consistent with $H_X(A,\mathcal D)$, the vertex $u$ belongs to exactly one bag of $T$, and $T$ is rooted at $u$.
\end{itemize}
For $(A,\mathcal D,T)\in\mathcal C^\circ_{X,u,F}$, define
\[
    \mathrm{wt}^{-}_{A,\mathcal D,u}(\kappa,\eta):=
    \prod_{e\in A}\kappa_e
    \prod_{D\in\mathcal D}
    \left(
        \prod_{e\in E(D)\cap E(G[X])}\kappa_e
        \prod_{v\in V_{\mathrm{in}}(D)\setminus\{u\}}\eta_{v,F_D(v)}
    \right).
\]
Compared with $\mathrm{wt}_{A,\mathcal D}$, $\mathrm{wt}^{-}_{A,\mathcal D,u}$ omits the boundary-pair factor $\eta_{u,F}$ at $u$ and keeps only pseudo-ear edge occurrences in $E(D)\cap E(G[X])$.
By $(u,F)$-openness, the only possible uncharged edge occurrences are edges of the boundary pair $F$ at $u$ that lie outside $G[X]$.
Define
\[
    Q^\circ_{X,u,F}(\kappa,\eta):=
    \sum_{(A,\mathcal D,T)\in\mathcal C^\circ_{X,u,F}}
    \left(\prod_{v\in X\setminus\{u\}}(-1)^{d_T^+(v)}d_T^+(v)!\right)
    \mathrm{wt}^{-}_{A,\mathcal D,u}(\kappa,\eta).
\]
By definition, if $X=\{u\}$, then $Q^{\circ}_{\{u\},u,F}(\kappa,\eta)=1$.

For $X\subseteq V(G)$ with $v\in X$ and $d\geq 0$, define
\[
    Q^\bullet_{X,v,F,d}(\kappa,\eta):=
    \sum_{\{Y_1,\dots,Y_d\}\in\Pi_d(X\setminus\{v\})}
    \prod_{i=1}^d Q^\circ_{Y_i\cup\{v\},v,F}(\kappa,\eta).
\]
Since $\Pi_0(\emptyset)=\{\emptyset\}$ and $\Pi_d(\emptyset)=\emptyset$ for $d>0$, we have $Q^\bullet_{\{v\},v,F,0}=1$ and $Q^\bullet_{\{v\},v,F,d}=0$ for $d>0$.
The definition is evaluated by subset convolution.
\begin{lemma}\label{lem:app-attachment}
For $v\in V(G)$, $F\in\binom{\delta_G(v)}{2}$, and $X\subseteq V(G)\setminus\{v\}$, assume that the values $Q^\circ_{Y\cup\{v\},v,F}(\kappa,\eta)$ are known for all $Y\subseteq X$.
Then all values $Q^\bullet_{X\cup\{v\},v,F,d}(\kappa,\eta)$ with $d\geq0$ can be computed in $O^*(2^{|X|})$ arithmetic operations.
\end{lemma}
\begin{proof}
Apply \Cref{lem:app-partition-convolution} to the ground set $X$ with $x_\emptyset:=0$ and $x_Y:=Q^{\circ}_{Y\cup\{v\},v,F}(\kappa,\eta)$ for nonempty $Y\subseteq X$.
For each $d\geq0$, the resulting partition sum for $X$ is exactly the sum in the definition of $Q^\bullet_{X\cup\{v\},v,F,d}(\kappa,\eta)$.
\end{proof}

The following lemma shows that it is enough to compute the values $Q^\bullet_{V(G),r,F,d}(\kappa,\eta)$ for the fixed vertex $r$.
\begin{lemma}\label{lem:app-unrooted-recovery}
We have
\[
    Q(\kappa,\eta)
    =
    \sum_{F\in\binom{\delta_G(r)}{2}}
    \sum_{d\geq 1}
    \eta_{r,F}(-1)^{d-1}(d-1)!
    Q^\bullet_{V(G),r,F,d}(\kappa,\eta).
\]
\end{lemma}
\begin{proof}
We show that each summand in \Cref{lem:app-bag-tree-mobius-global} appears exactly once on the right-hand side with the same weight, including the same M\"obius factor.
Fix $F\in\binom{\delta_G(r)}2$ and $d\geq1$.
It suffices to show that the contribution
\[
    \eta_{r,F}(-1)^{d-1}(d-1)!
    \sum_{\{Y_1,\dots,Y_d\}\in\Pi_d(V(G)\setminus\{r\})}
    \prod_{i=1}^d Q^\circ_{Y_i\cup\{r\},r,F}(\kappa,\eta).
\]
is exactly the part of \Cref{lem:app-bag-tree-mobius-global} indexed by triples $(A,\mathcal D,T)$ such that $F_{D_r}(r)=F$ and $d_T(r)=d$, where $D_r$ is the pseudo-ear containing $r$.

Expand the expression above.
A summand is specified by a partition $\{Y_1,\dots,Y_d\}$ of $V(G)\setminus\{r\}$ and triples $(A_i,\mathcal D_i,T_i)\in\mathcal C^\circ_{Y_i\cup\{r\},r,F}$ for $i\in\{1,\dots,d\}$.
Identify the $d$ copies of $r$ in the rooted 1-sum trees $T_1,\dots,T_d$.
The resulting 1-sum tree $T$ satisfies $d_T(r)=d$, and the components of $T-r$ have vertex sets $Y_1,\dots,Y_d$.
Let $A:=\bigcup_{i=1}^d A_i$.

It remains to define the spanning pseudo-ear packing $\mathcal D$.
All pseudo-ears not containing a copy of $r$ are kept.
For each $i$, let $D_i^r$ be the pseudo-ear in $\mathcal D_i$ containing $r$ as an inner vertex, and write $F=\{e_1,e_2\}$.
For each $a\in\{1,2\}$, let $i_a$ be the unique index with $e_a\in E(G[Y_{i_a}\cup\{r\}])$.
If $i_1=i_2$, keep $D_{i_1}^r$ as the pseudo-ear through $r$.
If $i_1\neq i_2$, orient $D_{i_1}^r$ and $D_{i_2}^r$ as $(x,\dots,e_1,r,e_2,y)$ and $(z,e_1,r,e_2,\dots,w)$, and replace them by the walk pseudo-ear $(x,\dots,e_1,r,e_2,\dots,w)$.
For each $i\notin\{i_1,i_2\}$, the pseudo-ear $D_i^r$ has no inner vertex in $Y_i$ and no edge occurrence in $E(G[Y_i\cup\{r\}])$, by $(r,F)$-openness, so $D_i^r$ disappears after the copies of $r$ are identified.
This defines a spanning pseudo-ear packing $\mathcal D$ of $V(G)$.

The construction is reversible.
Given a triple $(A,\mathcal D,T)$ in \Cref{lem:app-bag-tree-mobius-global} with $F_{D_r}(r)=F$ and $d_T(r)=d$, deleting $r$ from $T$ gives the partition $\{Y_1,\dots,Y_d\}$.
For each $i$, the pseudo-ear counted in the factor $Q^\circ_{Y_i\cup\{r\},r,F}$ that contains $r$ as an inner vertex is obtained from $D_r$ by keeping the inner vertices $\{r\}\cup(V_{\mathrm{in}}(D_r)\cap Y_i)$ and using the same boundary pair $F$ at $r$.
If $V_{\mathrm{in}}(D_r)\cap Y_i=\emptyset$, this is the pseudo-ear whose only inner vertex is $r$.
Thus the expansion contains every summand of \Cref{lem:app-bag-tree-mobius-global} with $F_{D_r}(r)=F$ and $d_T(r)=d$ exactly once.

Finally, the weights are preserved.
The $\kappa$-factor is the product of the $\kappa$-factors from the $d$ rooted components, and the gluing above does not change the edge occurrences charged inside those components.
The $\eta$-factors and the M\"obius factors are matched in the same way.
Away from $r$, both types of factors come from the $Q^\circ$-terms.
At $r$, the multiplier supplies the missing factors: $\eta_{r,F}$ for the boundary-pair factor and $(-1)^{d-1}(d-1)!$ for the M\"obius factor, since $d_T(r)=d$.
Therefore the right-hand side equals the expression in \Cref{lem:app-bag-tree-mobius-global}.
\end{proof}

Fix an integer $h\geq 1$, and assume that all values $Q^\bullet_{X',v,F',d}(\kappa,\eta)$ with $v\in X'$ and $|X'|\leq h$ are already known.
We compute $Q^{\circ}_{X,u,F}(\kappa,\eta)$ for sets $X$ with $|X|=h+1$.
We try all choices of the top bag $Z$ containing $u$, and evaluate the parts attached below vertices of $Z$ by the already computed $Q^\bullet$-values.

\paragraph*{Attachment substitutions.}
Fix $u\in V(G)$ and $F\in\binom{\delta_G(u)}2$.
For a pair $F'$ of edges, let $V(F')$ be the set of endpoints of the two edges in $F'$.
Fix a nonnegative integer $h$.
Define the table $\tau=\tau^{h,u,F}(\kappa,\eta)$, whose entries are denoted by $\tau_{Y,v,F',L}$, as follows.
For each $Y\subseteq V(G)$ and $L\subseteq V(G)\setminus Y$, the entries indexed by $Y$ and $L$ are defined as follows.
For $v=u$ and $F'\in\binom{\delta_G(u)}2$, let
\[
    \tau_{Y,u,F',L}:=
    \begin{cases}
        1 & \text{if }u\in Y,\ F'=F,\text{ and }L=\emptyset,\\
        0 & \text{otherwise.}
    \end{cases}
\]
For $v\in Y\setminus\{u\}$ and $F'\in\binom{\delta_G(v)}2$, set $\tau_{Y,v,F',L}:=0$ unless $V(F')\setminus Y\subseteq L$, $L\cap V(F)=\emptyset$, and $|L\cup\{v\}|\leq h$.
If these conditions hold, let
\[
    \tau_{Y,v,F',L}
    :=
    \eta_{v,F'}
    \sum_{d\geq 0}(-1)^d d! Q^\bullet_{L\cup\{v\},v,F',d}(\kappa,\eta).
\]
For $Z\subseteq V(G)$ with $u\in Z$, let $\mathcal R_{Z,u,F}$ be the set of pairs $(A,\mathcal D)$ satisfying the following conditions.
\begin{itemize}
    \item $A\subseteq E(G[Z])$, and $\mathcal D$ is a spanning pseudo-ear packing of $Z$.
    \item If $D_u$ is the pseudo-ear with $u\in V_{\mathrm{in}}(D_u)$, then $F_{D_u}(u)=F$.
    \item The graph $\left(Z,A\cup\bigcup_{D\in\mathcal D}\left(E(D)\cap E(G[Z])\right)\right)$ is connected.
\end{itemize}
For $L\subseteq V(G)\setminus Z$, define
\[
    \begin{aligned}
    &R_{Z,u,F,L}(\kappa,\tau):=\\
    &
    \sum_{(A,\mathcal D)\in\mathcal R_{Z,u,F}}
    \left(
    \prod_{e\in A}\kappa_e
    \cdot
    \sum_{\substack{(L_v)_{v\in Z}\\ \biguplus_{v\in Z}L_v=L}}
    \prod_{D\in\mathcal D}
    \left(
        \prod_{e\in E(D)\cap E(G[Z])}\kappa_e
        \cdot
        \prod_{v\in V_{\mathrm{in}}(D)}
        \tau_{Z,v,F_D(v),L_v}
    \right)
    \right).
    \end{aligned}
\]
The following lemma shows that evaluating \(R_{Z,u,F,L}\) with this table \(\tau\) gives the desired value \(Q^\circ_{X,u,F}\).
Before proceeding, we spell out the role of the $\tau$-factors in \(R_{Z,u,F,L}\).
The factor at \(u\) is nonzero only for \(F'=F\) and \(L_u=\emptyset\), which ensures \(F_{D_u}(u)=F\) and that no bag other than \(Z\) contains \(u\).
For \(v\in Z\setminus\{u\}\), the condition \(V(F')\setminus Z\subseteq L_v\) ensures that \(F'\) does not leave \(Z\cup L_v\), and hence does not leave \(X=Z\cup L=Z\cup\bigcup_{v\in Z}L_v\).
Thus every non-root boundary pair lies in \(E(G[X])\), as required by \((u,F)\)-openness.
The condition \(L_v\cap V(F)=\emptyset\) prevents endpoints of \(F\) outside \(Z\) from being assigned to a set \(L_v\) with \(v\neq u\); such endpoints must remain outside \(X=Z\cup\bigcup_{v\in Z}L_v\) as part of the boundary pair \(F\).

\begin{figure}[t]
\centering
\begin{tikzpicture}[
  scale=1.08,
  transform shape,
  x=1cm,y=1cm,
  vertex/.style={circle,fill=black,inner sep=0pt,minimum size=2.0mm},
  inner vertex/.style={circle,fill=black,inner sep=0pt,minimum size=2.0mm},
  bag/.style={circle,draw=black,fill=white,line width=0.6pt,inner sep=0pt,minimum size=9mm,font=\small,align=center},
  tree edge/.style={line width=0.55pt,gray!75!black},
  f edge/.style={line width=1.75pt,line cap=round},
  f one/.style={f edge,red!70!black},
  boundary f one/.style={f one},
  f two/.style={f edge,blue!65!black},
  boundary f two/.style={f two},
  f three/.style={f edge,green!45!black},
  boundary f three/.style={f three},
  f four/.style={f edge,purple!70!black},
  boundary f four/.style={f four},
  f five/.style={f edge,cyan!60!black},
  boundary f five/.style={f five},
  f six/.style={f edge,teal!70!black},
  boundary f six/.style={f six},
  block region/.style={draw=black,rounded corners=5pt,fill=gray!8,line width=0.55pt},
  subtree box/.style={draw=black,rounded corners=3pt,line width=0.55pt},
  ear edge/.style={line width=0.75pt,line cap=round,line join=round},
  ear one/.style={ear edge,red!70!black},
  ear two/.style={ear edge,blue!65!black},
  ear three/.style={ear edge,green!45!black},
  ear four/.style={ear edge,purple!70!black},
  ear five/.style={ear edge,cyan!60!black},
  ear six/.style={ear edge,teal!70!black},
  note/.style={font=\scriptsize,inner sep=1pt},
  state label/.style={font=\small,align=center},
  every label/.style={font=\scriptsize,inner sep=1pt}
]
\node[state label] at (-4.55,-3.02) {(a)};

\draw[block region] (-6.10,1.90) rectangle (-3.00,4.24);
\node[note,anchor=west] at (-2.90,3.98) {$Z$};
\coordinate (zu) at (-4.55,3.88);
\coordinate (zuTop) at (-4.55,4.56);
\coordinate (zx) at (-4.95,3.00);
\coordinate (zy) at (-4.15,3.00);
\coordinate (za) at (-5.50,2.23);
\coordinate (zb) at (-3.60,2.23);
\coordinate (zaOut) at (-5.50,1.51);
\coordinate (zbOneOut) at (-4.05,1.51);
\coordinate (zbTwoOut) at (-3.15,1.51);
\draw[ear one] (zuTop) -- (zu) -- (zx) -- (za) -- (zaOut);
\draw[ear two] (zb) -- (zy) -- (zx);
\draw[ear three] (zb) -- (zbOneOut);
\draw[ear three] (zb) -- (zbTwoOut);
\draw[f one] (zx) -- (za);
\draw[f one] (za) -- (zaOut);
\draw[f three] (zb) -- (zbOneOut);
\draw[f three] (zb) -- (zbTwoOut);
\node[vertex,label=right:{$u$}] at (zu) {};
\node[inner vertex] at (zuTop) {};
\node[inner vertex,label=left:{$x$}] at (zx) {};
\node[inner vertex,label=right:{$y$}] at (zy) {};
\node[vertex,label=above:{$a$}] at (za) {};
\node[vertex,label=above:{$b$}] at (zb) {};
\node[inner vertex,label=below:{$s$}] at (zaOut) {};
\node[inner vertex,label=below left:{$p$}] at (zbOneOut) {};
\node[inner vertex,label=below right:{$q$}] at (zbTwoOut) {};

\begin{scope}[yshift=-1.25cm]
\coordinate (pa) at (-6.55,1.24);
\coordinate (paOut) at (-6.55,1.86);
\coordinate (paTL) at (-6.88,0.62);
\coordinate (paTR) at (-6.22,0.62);
\coordinate (paL) at (-6.88,0.00);
\coordinate (paR) at (-6.22,0.00);
\coordinate (paBL) at (-6.88,-0.62);
\coordinate (paBR) at (-6.22,-0.62);
\draw[subtree box] (-7.20,-1.08) rectangle (-5.90,1.42);
\node[note] at (-6.55,-1.30) {$\scriptstyle X_{a,1}$};
\draw[ear one] (paOut) -- (pa) -- (paTL) -- (paL) -- (paR) -- (paTR) -- (paTL);
\draw[ear six] (paL) -- (paBL) -- (paBR) -- (paR);
\draw[f one] (pa) -- (paOut);
\draw[f one] (pa) -- (paTL);
\node[vertex,label=right:{$a$}] at (pa) {};
\node[inner vertex,label=above:{$x$}] at (paOut) {};
\node[note,anchor=east,xshift=-4pt] at (paTL) {$s$};
\foreach \p in {paTL,paTR,paL,paR,paBL,paBR} {
  \node[inner vertex] at (\p) {};
}

\coordinate (qa) at (-4.53,1.24);
\coordinate (qbTL) at (-4.84,0.62);
\coordinate (qbTR) at (-4.22,0.62);
\coordinate (qbL) at (-4.84,0.00);
\coordinate (qbR) at (-4.22,0.00);
\coordinate (qbBL) at (-4.84,-0.62);
\coordinate (qbBR) at (-4.22,-0.62);
\draw[subtree box] (-5.13,-1.08) rectangle (-3.93,1.42);
\node[note] at (-4.53,-1.30) {$\scriptstyle X_{b,1}$};
\draw[ear three] (qbTL) -- (qa) -- (qbTR) -- (qbL) -- (qbTL);
\draw[ear four] (qbTR) -- (qbR) -- (qbBR) -- (qbBL) -- (qbR);
\draw[f three] (qa) -- (qbTL);
\draw[f three] (qa) -- (qbTR);
\node[vertex,label=right:{$b$}] at (qa) {};
\node[note,anchor=east,xshift=-4pt] at (qbTL) {$p$};
\node[note,anchor=west,xshift=4pt] at (qbTR) {$q$};
\foreach \p in {qbTL,qbTR,qbL,qbR,qbBL,qbBR} {
  \node[inner vertex] at (\p) {};
}

\coordinate (rb) at (-2.60,1.24);
\coordinate (rbOutL) at (-2.93,1.86);
\coordinate (rbOutR) at (-2.27,1.86);
\coordinate (rcTL) at (-2.92,0.62);
\coordinate (rcTR) at (-2.28,0.62);
\coordinate (rcL) at (-2.92,0.00);
\coordinate (rcR) at (-2.28,0.00);
\coordinate (rcBL) at (-2.92,-0.62);
\coordinate (rcBR) at (-2.28,-0.62);
\draw[subtree box] (-3.22,-1.08) rectangle (-1.98,1.42);
\node[note] at (-2.60,-1.30) {$\scriptstyle X_{b,2}$};
\draw[ear five] (rb) -- (rcTL) -- (rcTR);
\draw[ear five] (rcTR) -- (rcR) -- (rcBR) -- (rcBL) -- (rcL) -- (rcR);
\draw[f three] (rb) -- (rbOutL);
\draw[f three] (rb) -- (rbOutR);
\node[vertex,label=right:{$b$}] at (rb) {};
\node[inner vertex] at (rbOutL) {};
\node[inner vertex] at (rbOutR) {};
\node[note,anchor=south east,xshift=-5pt,yshift=4pt] at (rbOutL) {$p$};
\node[note,anchor=south west,xshift=5pt,yshift=4pt] at (rbOutR) {$q$};
\foreach \p in {rcTL,rcTR,rcL,rcR,rcBL,rcBR} {
  \node[inner vertex] at (\p) {};
}
\end{scope}

\begin{scope}[shift={(7.55,1.18)}]
\draw[block region] (-5.30,-0.33) rectangle (-2.20,2.01);
\node[note,anchor=west] at (-2.10,1.75) {$Z$};

\coordinate (cu) at (-3.75,1.65);
\coordinate (cuTop) at (-3.75,2.33);
\coordinate (cx) at (-4.15,0.77);
\coordinate (cy) at (-3.35,0.77);
\coordinate (ca) at (-4.70,0.00);
\coordinate (cb) at (-2.80,0.00);
\coordinate (cout) at (-4.70,2.16);
\coordinate (waTL) at (-5.95,-0.85);
\coordinate (waTR) at (-5.30,-0.85);
\coordinate (waL) at (-5.95,-1.45);
\coordinate (waR) at (-5.30,-1.45);
\coordinate (waBL) at (-5.95,-2.05);
\coordinate (waBR) at (-5.30,-2.05);
\coordinate (wbTL) at (-4.05,-0.85);
\coordinate (wbTR) at (-3.45,-0.85);
\coordinate (wbL) at (-4.05,-1.45);
\coordinate (wbR) at (-3.45,-1.45);
\coordinate (wbBL) at (-4.05,-2.05);
\coordinate (wbBR) at (-3.45,-2.05);
\coordinate (wcTL) at (-2.22,-0.85);
\coordinate (wcTR) at (-1.62,-0.85);
\coordinate (wcL) at (-2.22,-1.45);
\coordinate (wcR) at (-1.62,-1.45);
\coordinate (wcBL) at (-2.22,-2.05);
\coordinate (wcBR) at (-1.62,-2.05);

\draw[ear one] (cuTop) -- (cu) -- (cx) -- (ca) -- (waTL) -- (waL) -- (waR) -- (waTR) -- (waTL);
\draw[ear six] (waL) -- (waBL) -- (waBR) -- (waR);
\draw[ear two] (cb) -- (cy) -- (cx);
\draw[ear three] (wbTL) -- (cb) -- (wbTR) -- (wbL) -- (wbTL);
\draw[ear four] (wbTR) -- (wbR) -- (wbBR) -- (wbBL) -- (wbR);
\draw[ear five] (cb) -- (wcTL) -- (wcTR);
\draw[ear five] (wcTR) -- (wcR) -- (wcBR) -- (wcBL) -- (wcL) -- (wcR);

\node[vertex,label=right:{$u$}] at (cu) {};
\node[inner vertex] at (cuTop) {};
\node[inner vertex,label=left:{$x$}] at (cx) {};
\node[inner vertex,label=right:{$y$}] at (cy) {};
\node[vertex,label=above:{$a$}] at (ca) {};
\node[vertex,label=above:{$b$}] at (cb) {};
\node[note,anchor=east,xshift=-4pt] at (waTL) {$s$};
\node[note,anchor=east,xshift=-4pt] at (wbTL) {$p$};
\node[note,anchor=west,xshift=4pt] at (wbTR) {$q$};
\foreach \p in {waTL,waTR,waR,waBR,waBL,waL,wbTL,wbTR,wbR,wbBR,wbBL,wbL,wcTL,wcTR,wcR,wcBR,wcBL,wcL} {
  \node[inner vertex] at (\p) {};
}
\draw[subtree box] (-6.25,-2.35) rectangle (-4.92,-0.55);
\node[note] at (-5.59,-2.57) {$\scriptstyle X_{a,1}\setminus\{a\}$};
\draw[subtree box] (-4.40,-2.35) rectangle (-3.10,-0.55);
\node[note] at (-3.75,-2.57) {$\scriptstyle X_{b,1}\setminus\{b\}$};
\draw[subtree box] (-2.57,-2.35) rectangle (-1.27,-0.55);
\node[note] at (-1.92,-2.57) {$\scriptstyle X_{b,2}\setminus\{b\}$};
\node[state label] at (-3.95,-4.20) {(b)};
\end{scope}
\end{tikzpicture}
\caption{Obtaining a \((u,F)\)-open spanning pseudo-ear decomposition by gluing the \(Z\)-part and the \(X_{v,i}\)-parts in \Cref{lem:app-top-bag-substitution}. The parts before and after gluing are shown in (a) and (b), respectively. In (a), the thick colored segments indicate the edges in the pairs \(F_{D_v}(v)\). The \(Z\)-part and the \(X_{a,1}\)-part are glued at the common boundary pair \(F_{D_a}(a)=\{xa,as\}\), while the \(Z\)-part and the two \(X_{b,i}\)-parts are glued at the common boundary pair \(F_{D_b}(b)=\{bp,bq\}\).}
\label{fig:app-top-bag-substitution}
\end{figure}
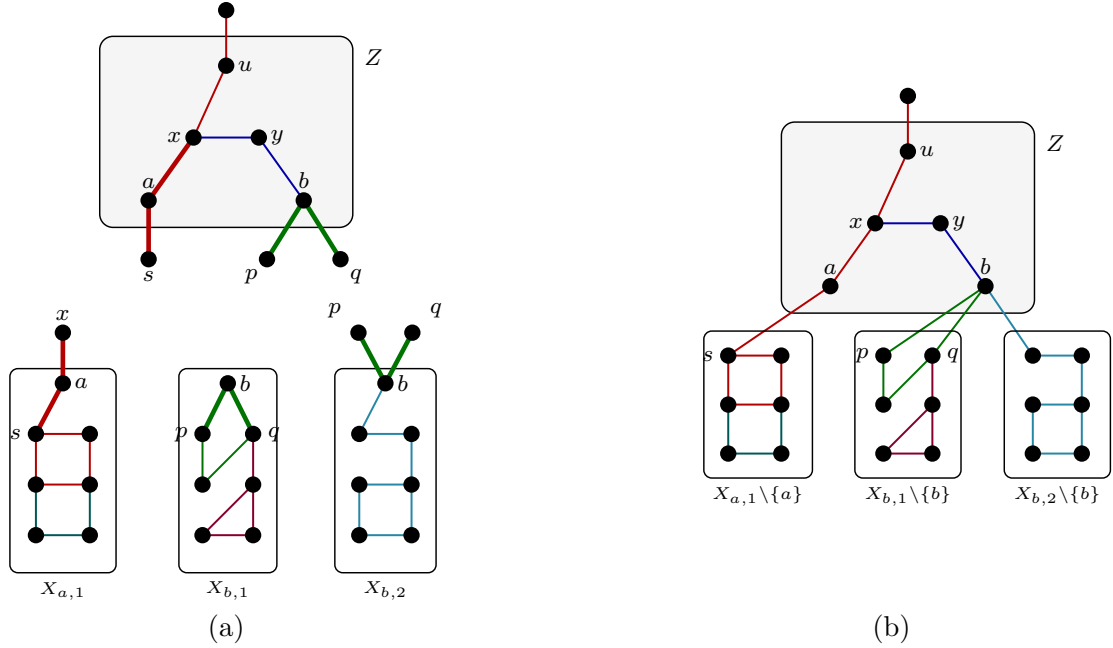

\begin{lemma}\label{lem:app-top-bag-substitution}
For $X\subseteq V(G)$ with $u\in X$, $F\in\binom{\delta_G(u)}2$, and $|X|\geq2$, we have
\[
    Q^{\circ}_{X,u,F}(\kappa,\eta)
    =
    \sum_{\substack{Z\subseteq X\\ u\in Z,\ |Z|\geq 2}}
    R_{Z,u,F,X\setminus Z}(\kappa,\tau^{|X|-1,u,F}(\kappa,\eta)).
\]
\end{lemma}
\begin{proof}
Expand the right-hand side.
Since $(A_0,\mathcal D_0)\in\mathcal R_{Z,u,F}$ has $F_{D_u}(u)=F$, the definition of the entry $\tau$ at $u$ implies that $\tau_{Z,u,F_{D_u}(u),L_u}$ is nonzero exactly when $L_u=\emptyset$.
A nonzero term consists of the following data:
\begin{itemize}
    \item a set $Z$ with $u\in Z$ and $|Z|\geq2$, and a local pair $(A_0,\mathcal D_0)\in\mathcal R_{Z,u,F}$;
    \item a disjoint union $\biguplus_{v\in Z}L_v=X\setminus Z$ with $L_u=\emptyset$;
    \item for each $v\in Z\setminus\{u\}$, an integer $d_v\geq0$, a partition $\{X_{v,1}\setminus\{v\},\dots,X_{v,d_v}\setminus\{v\}\}$ of $L_v$, and, for each $j$, a term of $Q^\circ_{X_{v,j},v,F_{D_v}(v)}(\kappa,\eta)$.
\end{itemize}
Here $D_v$ denotes the unique pseudo-ear in $\mathcal D_0$ with $v\in V_{\mathrm{in}}(D_v)$.

We first construct, from a nonzero term described above, a triple in $\mathcal C^\circ_{X,u,F}$ whose unique bag containing $u$ has vertex set $Z$; see \Cref{fig:app-top-bag-substitution}.
For each $v\in Z\setminus\{u\}$ and $j\in\{1,\dots,d_v\}$, write the triple corresponding to the child part indexed by $X_{v,j}$ as $(A_{v,j},\mathcal D_{v,j},T_{v,j})$.
\begin{itemize}
    \item $A$ is $A_0\cup\bigcup_{v\in Z\setminus\{u\}}\bigcup_{j=1}^{d_v}A_{v,j}$.
    \item $T$ is obtained by taking $Z$ as the top bag and, for each $v\in Z\setminus\{u\}$ and each $j\in\{1,\dots,d_v\}$, attaching $T_{v,j}$ at its root $v$.
    \item $\mathcal D$ is obtained from the union of $\mathcal D_0$ and all $\mathcal D_{v,j}$ by the same gluing as in the proof of \Cref{lem:app-unrooted-recovery}, applied separately at each $v\in Z\setminus\{u\}$.
    Namely, treat the top bag $Z$ as the $0$-th piece: let $X_{v,0}:=Z$ and $D_{v,0}:=D_v$, and let $D_{v,j}$ be the pseudo-ear containing the root vertex $v$ in the $j$-th child part.
    These pseudo-ears all have the same boundary pair $F_{D_v}(v)$ at $v$, and we glue along this pair.
    The nonzero $\tau$-term gives $V(F_{D_v}(v))\setminus Z\subseteq L_v$, so each edge of $F_{D_v}(v)$ lies in exactly one of the pieces $X_{v,0},X_{v,1},\dots,X_{v,d_v}$.
    As in \Cref{lem:app-unrooted-recovery}, any pseudo-ear $D_{v,j}$ not using one of the two edges of $F_{D_v}(v)$ contributes nothing after the copies of $v$ are identified.
    All pseudo-ears not involved in these gluings are kept unchanged.
\end{itemize}
The following checks show that the constructed object belongs to $\mathcal C^\circ_{X,u,F}$ and that the construction is reversible.
\begin{itemize}
    \item The vertex set is exactly $X$: the top bag uses $Z$, and the attached child parts add the disjoint sets $X_{v,j}\setminus\{v\}$ whose union is $\biguplus_{v\in Z}L_v=X\setminus Z$.
    In the merge above, the only pseudo-ears that disappear without contributing to the new pseudo-ear have the single inner vertex $v$.
    Hence no vertex in $X\setminus Z$ is lost.
    \item The prescribed boundary pair at $u$ is $F$: by the definition of $\mathcal R_{Z,u,F}$, we have $F_{D_u}(u)=F$.
    \item The resulting pseudo-ear packing is $(u,F)$-open.  Indeed, for a child part on $X_{v,j}$, every boundary pair in that part other than $F_{D_v}(v)$ at $v$ is contained in $E(G[X_{v,j}])\subseteq E(G[X])$.  For a top-bag vertex $v\in Z\setminus\{u\}$, the $\tau$-entry used at $v$ is nonzero only if $V(F_{D_v}(v))\setminus Z\subseteq L_v$, and then $F_{D_v}(v)\subseteq E(G[Z\cup L_v])\subseteq E(G[X])$.  Thus the only boundary pair allowed to use an edge of $E(G)\setminus E(G[X])$ is $F$ at $u$.
    \item The graph $H_X(A,\mathcal D)$ is connected.  Indeed, the top graph on $Z$ is connected by the definition of $\mathcal R_{Z,u,F}$, and each child part gives a connected graph containing its root $v$.  Attaching these graphs at the corresponding vertices $v\in Z\setminus\{u\}$ preserves connectedness.  The merge of pseudo-ears described above does not change the set of edge occurrences in $E(G[X])$.
    \item The construction is reversible.  Given a triple in $\mathcal C^\circ_{X,u,F}$ whose unique bag containing $u$ has vertex set $Z$, delete this top bag.  The child subtrees attached at each $v\in Z\setminus\{u\}$ give the sets $X_{v,j}$.  As in \Cref{lem:app-unrooted-recovery}, the root pseudo-ear in the child part on $X_{v,j}$ is obtained from the pseudo-ear $D\in\mathcal D$ containing $v$ by keeping the inner vertices $\{v\}\cup(V_{\mathrm{in}}(D)\cap X_{v,j})$ and using the same boundary pair $F_D(v)$ at $v$; if $V_{\mathrm{in}}(D)\cap X_{v,j}=\emptyset$, this is the pseudo-ear whose only inner vertex is $v$.  The top bag gives $(A_0,\mathcal D_0)$.  Since the 1-sum tree is a tree, the sets $L_v=\biguplus_j(X_{v,j}\setminus\{v\})$ form a disjoint union of $X\setminus Z$.  Moreover, $V(F_D(v))\setminus Z\subseteq L_v$, because every endpoint of $F_D(v)$ outside $Z$ lies in the part below $v$; hence the corresponding $\tau$-entry is nonzero.
\end{itemize}
Since $|Z|\geq2$, each child part has vertex set size at most $|X|-1$, and hence is included by the threshold $h=|X|-1$.

The weights match under this bijection, as follows.
The set $A$ is the union of $A_0$ and the sets $A_{v,j}$, and the merge of pseudo-ears does not change the set of edge occurrences in $E(G[X])$.
Indeed, each pseudo-ear that disappears has no edge occurrence in $E(G[X])$.
Thus the product of the $\kappa$-variables is the product of the $\kappa$-variables from the top part and the child parts.
For the $\eta$-variables, the factor $\eta_{x,F_D(x)}$ for each $x\in X\setminus Z$ is supplied by the unique child term containing $x$.
The factor $\eta_{v,F_{D_v}(v)}$ for each $v\in Z\setminus\{u\}$ is supplied by $\tau_{Z,v,F_{D_v}(v),L_v}$.
The child terms $Q^\circ_{X_{v,j},v,F_{D_v}(v)}(\kappa,\eta)$ do not supply $\eta_{v,F_{D_v}(v)}$, since $Q^\circ$ omits the $\eta$-factor at its root.
For $u$, the entry $\tau_{Z,u,F,\emptyset}=1$ supplies no $\eta$-variable, which agrees with $\mathrm{wt}^{-}_{A,\mathcal D,u}$.
Thus every $\eta$-variable in $\mathrm{wt}^{-}_{A,\mathcal D,u}$ is supplied exactly once.
It remains to compare the M\"obius coefficients.
For $v\in Z\setminus\{u\}$, attaching $d_v$ child 1-sum trees below the top bag gives $d_T^+(v)=d_v$, and hence the M\"obius factor of $v$ in $Q^\circ_{X,u,F}(\kappa,\eta)$ is $(-1)^{d_v}d_v!$; this is the factor in the summand of $\tau_{Z,v,F_{D_v}(v),L_v}$ with this value of $d_v$.
For $x\in X\setminus Z$, the vertex $x$ lies in a unique child part, and its value of $d_T^+(x)$ is the same as in that child 1-sum tree; thus the M\"obius factor $(-1)^{d_T^+(x)}d_T^+(x)!$ is already present in the corresponding child term.
The root $u$ has no M\"obius factor in the definition of $Q^\circ_{X,u,F}(\kappa,\eta)$.
\end{proof}

Thus, to compute the values $Q^\circ_{X,u,F}(\kappa,\eta)$, it remains to compute $R_{Z,u,F,L}(\kappa,\tau)$ for $Z\subseteq X$ with $u\in Z$ and $|Z|\geq2$, and for $L=X\setminus Z\subseteq V(G)\setminus Z$, where $\tau$ is the table $\tau^{|X|-1,u,F}(\kappa,\eta)$.

\subsection{Computing \texorpdfstring{\(R\)}{R}-polynomials}\label{subsec:evaluation-r-polynomials}
We now compute the \(R\)-polynomials.
We use two auxiliary tables, \(M\) and \(R^{\mathrm{all}}\).
The table \(M\) counts one pseudo-ear with a prescribed set of inner vertices, while \(R^{\mathrm{all}}\) is the version of \(R\) without the connectedness constraint.
We compute \(R^{\mathrm{all}}\) from \(M\) in \Cref{lem:app-rall-expansion}, and then compute \(R\) from \(R^{\mathrm{all}}\) in \Cref{lem:app-local-z-substitution}.

Fix $Y\subseteq V(G)$.
For nonempty $I\subseteq Y$, let $\mathcal E_{Y,I}$ be the set of pseudo-ears $D$ of $G$ with $V_{\mathrm{in}}(D)=I$.
For $L\subseteq V(G)\setminus Y$, define
\[
    M_{Y,I,L}(\kappa,\tau):=
    \sum_{D\in\mathcal E_{Y,I}}
    \left(
        \prod_{e\in E(D)\cap E(G[Y])}\kappa_e
        \cdot
        \sum_{\substack{(L_v)_{v\in I}\\ \biguplus_{v\in I}L_v=L}}
        \prod_{v\in I}
        \tau_{Y,v,F_D(v),L_v}
    \right).
\]
A summand of $M_{Y,I,L}(\kappa,\tau)$ chooses $D$ with $V_{\mathrm{in}}(D)=I$ and sets $L_v$ for $v\in I$ with $\biguplus_{v\in I}L_v=L$; intuitively, $L_v$ is the set of vertices outside $Y$ assigned to the part attached at $v$.
The $\tau$-factor is used as follows.
If $u\in I$, then the factor at $u$ is nonzero only when $F_D(u)=F$ and $L_u=\emptyset$, and then it equals $1$.
For $v\in I\setminus\{u\}$, the factor $\tau_{Y,v,F_D(v),L_v}$ is nonzero only when $V(F_D(v))\setminus Y\subseteq L_v$, $L_v\cap V(F)=\emptyset$, and $|L_v\cup\{v\}|\leq h$.
Thus every endpoint of $D$ outside $Y$ that is incident with an edge of $F_D(v)$ for some $v\in I\setminus\{u\}$ lies in $\biguplus_{v\in I}L_v=L$.
The subscript $Y$ also records which edge occurrences are charged by $\kappa$, namely those inside $G[Y]$.
\Cref{fig:m-polynomial-path-ear} illustrates one summand of $M_{Y,I,L}(\kappa,\tau)$ for a walk pseudo-ear.
The non-inner vertices of the walk pseudo-ear are either in $Y$ or in the set $L_v$ corresponding to the inner vertex $v$ whose $F_D(v)$-edge reaches them; in the figure, $a\in Y$ and $b\in L_{v_3}\subseteq L$.

\begin{figure}[t]
\centering
\begin{tikzpicture}[
  x=1cm,y=1cm,
  vertex/.style={circle,fill=black,inner sep=0pt,minimum size=2mm},
  inner/.style={circle,fill=blue!65!black,inner sep=0pt,minimum size=2.2mm},
  ybox/.style={draw=black!45,rounded corners=2pt,line width=0.5pt},
  ibox/.style={draw=black!45,rounded corners=2pt,line width=0.55pt},
  lbox/.style={draw=black!45,densely dashed,rounded corners=2pt,line width=0.55pt},
  lpart/.style={draw=black!45,rounded corners=1.5pt,line width=0.45pt},
  branch/.style={densely dotted,line width=0.45pt,black!45},
  ear/.style={line width=1pt,blue!65!black},
  note/.style={font=\scriptsize,inner sep=1pt},
]
\draw[ybox] (-0.25,0.25) rectangle (5.25,3.05);
\node[note] at (5.02,2.82) {$Y$};
\draw[ibox] (0.35,0.55) rectangle (4.95,1.75);
\node[note,anchor=west] at (5.0,1.2) {$I=V_{\mathrm{in}}(D)$};
\draw[lbox] (-0.15,-1.8) rectangle (5.15,-0.45);
\draw[lpart] (0.3,-1.55) rectangle (1.5,-0.75);
\draw[lpart] (1.95,-1.55) rectangle (3.15,-0.75);
\draw[lpart] (3.6,-1.55) rectangle (4.8,-0.75);
\node[note] at (4.95,-0.68) {$L$};
\node[note] at (0.9,-1.32) {$L_{v_1}$};
\node[note] at (2.55,-1.32) {$L_{v_2}$};
\node[note] at (4.2,-1.32) {$L_{v_3}$};

\coordinate (a) at (0.9,2.45);
\coordinate (v1) at (0.9,1.1);
\coordinate (v2) at (2.55,1.1);
\coordinate (v3) at (4.2,1.1);
\coordinate (b) at (4.2,-1.1);

\draw[branch] (v1) -- (0.3,-0.75);
\draw[branch] (v1) -- (1.5,-0.75);
\draw[branch] (v2) -- (1.95,-0.75);
\draw[branch] (v2) -- (3.15,-0.75);
\draw[branch] (v3) -- (3.6,-0.75);
\draw[branch] (v3) -- (4.8,-0.75);
\draw[ear] (a) -- (v1) -- (v2) -- (v3) -- (b);
\node[vertex] at (a) {};
\foreach \p in {v1,v2,v3}
  \node[inner] at (\p) {};
\node[vertex] at (b) {};

\node[note,left=2pt] at (a) {$a$};
\node[note,below=2pt] at (v1) {$v_1$};
\node[note,below=2pt] at (v2) {$v_2$};
\node[note,below=2pt] at (v3) {$v_3$};
\node[note,right=2pt] at (b) {$b$};
\end{tikzpicture}
\caption{A walk pseudo-ear contributing to \(M_{Y,I,L}\).  The inner vertices \(I=V_{\mathrm{in}}(D)\) lie in \(Y\), and are drawn in the \(I\)-box.  The edge occurrences \(av_1\), \(v_1v_2\), and \(v_2v_3\) lie in \(G[Y]\) and contribute their \(\kappa\)-factors, while the edge \(v_3b\) leaves \(Y\) and is not counted by the \(\kappa\)-product.}
\label{fig:m-polynomial-path-ear}
\end{figure}
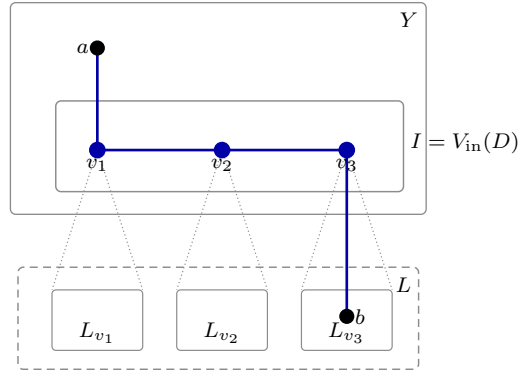

\begin{lemma}\label{lem:app-one-ear-held-karp}
For fixed $Y$ and $\tau$, all values $M_{Y,I,L}(\kappa,\tau)$ with $\emptyset\neq I\subseteq Y$ and $L\subseteq V(G)\setminus Y$ can be computed together in $O^*(2^{|V(G)|})$ arithmetic operations.
\end{lemma}
\begin{proof}
We use a Held--Karp type dynamic programming along the walk or cycle underlying a pseudo-ear.
For walk pseudo-ears, we use states $H_{x_1}[J,K,x,e]$.
The state counts partial oriented walks of the form $x_0-e_0-x_1-\cdots-x-e-y$.
Here $J\subseteq Y$ is the set of used inner vertices, and $K\subseteq V(G)\setminus Y$ is the exact outside set already used by completed inner vertices.
The analogous cycle states are denoted by $C_{x_1,e_0}[J,K,x,e]$.
They count partial oriented cycle walks of the form $x_0-e_0-x_1-\cdots-x-e-y$ with the same meaning of $J$ and $K$.
Let $E_G(x,y)$ denote the set of edges with endpoints $x$ and $y$.

A walk pseudo-ear has two orientations, and a cycle pseudo-ear with inner-vertex set $I$ has $2|I|$ possible traversals.
We count only one canonical traversal to avoid double-counting.
Fix a linear order of $V(G)$.
For a walk pseudo-ear with at least two inner vertices, the canonical orientation is the one whose first inner vertex is smaller than its last inner vertex; the one-inner-vertex case is handled directly and added once.
For a cycle pseudo-ear, the canonical traversal starts at the smallest vertex of $I$ and takes the direction in which the next inner vertex is smaller than the other neighbor of the start vertex; the two-inner-vertex case is handled directly and added once.
To summarize, for $x_1\in Y$, $e_1\in\delta_G(x_1)$, and $K\subseteq V(G)\setminus Y$, the initial walk states are
\[
H_{x_1}[\{x_1\},K,x_1,e_1]
:=
\sum_{\substack{e_0\in\delta_G(x_1)\\ e_0\neq e_1}}
\tau_{Y,x_1,\{e_0,e_1\},K}
\prod_{f\in\{e_0,e_1\}\cap E(G[Y])}\kappa_f .
\]
For $x\in J\setminus\{x_1\}$, $e\in\delta_G(x)$, and $K\subseteq V(G)\setminus Y$, the walk transition is
\[
H_{x_1}[J,K,x,e]
:=
\sum_{\substack{x'\in J\setminus\{x\},\, f\in E_G(x',x)\\ f\neq e}}
\sum_{K'\subseteq K}
H_{x_1}[J\setminus\{x\},K',x',f]\,
\tau_{Y,x,\{f,e\},K\setminus K'}
\prod_{g\in\{e\}\cap E(G[Y])}\kappa_g .
\]
Then, inductively, $H_{x_1}[J,K,x,e]$ equals the sum of $\prod_{z\in J}\tau_{Y,z,F_D(z),K_z}\prod_{f\in E(D)\cap E(G[Y])}\kappa_f$ over oriented partial walks $D=x_0-e_0-x_1-\cdots-x-e-y$ whose inner-vertex set is exactly $J$, whose first and last inner vertices are $x_1$ and $x$, and over disjoint sets $(K_z)_{z\in J}$ with union $K$; here $F_D(z)$ is the pair of edges of $D$ incident with $z$.

The cycle table uses the analogous formulas with the remembered edge $e_0$.
For $x_1,x_0,x_2\in Y$ with $x_2<x_0$, distinct edges $e_0\in E_G(x_0,x_1)$ and $e_1\in E_G(x_1,x_2)$, and $K\subseteq V(G)\setminus Y$, the initial cycle states are
\[
C_{x_1,e_0}[\{x_1\},K,x_1,e_1]
:=
\tau_{Y,x_1,\{e_0,e_1\},K}
\kappa_{e_1}.
\]
All other initial cycle states are set to $0$.
For $x\in J\setminus\{x_1\}$, $e\in\delta_{G[Y]}(x)$, and $K\subseteq V(G)\setminus Y$, the cycle transition is
\[
C_{x_1,e_0}[J,K,x,e]
:=
\sum_{\substack{x'\in J\setminus\{x\},\, f\in E_G(x',x)\\ f\neq e}}
\sum_{K'\subseteq K}
C_{x_1,e_0}[J\setminus\{x\},K',x',f]\,
\tau_{Y,x,\{f,e\},K\setminus K'}
\kappa_e.
\]
Since a cycle pseudo-ear uses only vertices of $Y$, every edge occurrence in the cycle table lies in $E(G[Y])$.
Then, inductively, $C_{x_1,e_0}[J,K,x,e]$ equals the sum of $\prod_{z\in J}\tau_{Y,z,F_D(z),K_z}\prod_{f\in E(D)\setminus\{e_0\}}\kappa_f$ over oriented partial cycle walks $D=x_0-e_0-x_1-e_1-x_2-\cdots-x-e-y$ whose inner-vertex set is exactly $J$, whose first and last inner vertices are $x_1$ and $x$, with $x_0,x_2\in Y$ and $x_2<x_0$, and over disjoint sets $(K_z)_{z\in J}$ with union $K$.
Here \(E(D)\) is regarded as a multiset of edge occurrences.
In particular, when a cycle state is completed with last edge $e=e_0$, the final edge occurrence contributes the factor $\kappa_e$ and supplies the missing charge for $e_0$.

Finally, $M_{Y,I,L}(\kappa,\tau)$ is obtained as the sum of the following terms.
\begin{itemize}
    \item If $I=\{x\}$, then for each two-edge subset $\{e_0,e_1\}$ of $\delta_G(x)$, add $\tau_{Y,x,\{e_0,e_1\},L}$, multiplied by $\kappa_f$ for each $f\in\{e_0,e_1\}$ that lies in $E(G[Y])$.
    \item If $I=\{x,y\}$ with $x<y$, then for each two-edge subset $\{e_0,e_1\}$ of $E_G(x,y)$ and each subset $K$ of $L$, add $\tau_{Y,x,\{e_0,e_1\},K}\tau_{Y,y,\{e_0,e_1\},L\setminus K}\kappa_{e_0}\kappa_{e_1}$.
    \item Add $H_{x_1}[I,L,x,e]$ over all $x_1,x\in I$ with $x_1<x$ and all $e\in\delta_G(x)$.
    \item If $|I|\ge 3$, add $C_{x_1,e_0}[I,L,x,e_0]$, where $x_1=\min I$, over all $x\in I\setminus\{x_1\}$ and all $e_0\in E_G(x_1,x)$.
\end{itemize}
The four cases above are disjoint and exhaustive.
For a walk pseudo-ear, the case \(|I|=1\) is handled by the first item, and the case \(|I|\geq2\) is handled by the third item.
In the latter case, \(x_1\) and \(x\) are the first and last inner vertices of the oriented walk.
The condition \(x_1<x\) selects exactly one of the two orientations and ensures that the walk is counted in exactly one direction.
For a cycle pseudo-ear, self-loops are absent, so \(|I|\geq2\).
The case \(|I|=2\) is handled by the second item, and the case \(|I|\geq3\) is handled by the fourth item.
In the latter case, \(x_1\), \(x_2\), and \(x_0\) are the first, second, and last vertices along the traversal of the cycle.
The condition $x_1=\min I$ uniquely selects the starting vertex, and the requirement $x_2<x_0$ in the initial cycle states selects exactly one of the two possible directions from $x_1$, ensuring that the cycle is counted from exactly one starting point and in exactly one direction.
Hence every pseudo-ear with $V_{\mathrm{in}}(D)=I$ contributes exactly once to $M_{Y,I,L}(\kappa,\tau)$.

It remains to analyze the running time.
There are $2^{|Y|}\cdot 2^{|V(G)\setminus Y|}=2^{|V(G)|}$ choices of the pair $(J,K)$.
For each fixed $J\subseteq Y$, the sums over $K'\subseteq K$ in the transitions are subset convolutions on the universe $V(G)\setminus Y$.
The sum over \(K\subseteq L\) in the two-inner-vertex cycle case is also a subset convolution on the same universe.
Hence these convolutions over all $J\subseteq Y$ cost $O^*(2^{|Y|}\cdot 2^{|V(G)\setminus Y|})=O^*(2^{|V(G)|})$ arithmetic operations.
The remaining endpoint and edge indices contribute only polynomial factors, so all tables for the fixed $Y$ are filled within the same bound.
\end{proof}

For fixed $u$ and $F$, let $\mathcal R^{\mathrm{all}}_{Y,u,F}$ be the set of pairs $(A,\mathcal D)$ satisfying the following conditions.
\begin{itemize}
    \item $A\subseteq E(G[Y])$, and $\mathcal D$ is a spanning pseudo-ear packing of $Y$.
    \item If $u\in Y$ and $D_u$ is the pseudo-ear with $u\in V_{\mathrm{in}}(D_u)$, then $F_{D_u}(u)=F$.
\end{itemize}
The choice of \(A\) is independent of the pseudo-ear packing \(\mathcal D\); all restrictions involving \(u\) and \(F\) concern only \(\mathcal D\).
For $L\subseteq V(G)\setminus Y$, define
\[
    \begin{aligned}
    &R^{\mathrm{all}}_{Y,u,F,L}(\kappa,\tau):=\\
    &\sum_{(A,\mathcal D)\in\mathcal R^{\mathrm{all}}_{Y,u,F}}
    \left(
    \prod_{e\in A}\kappa_e
    \cdot
    \sum_{\substack{(L_v)_{v\in Y}\\ \biguplus_{v\in Y}L_v=L}}
    \prod_{D\in\mathcal D}
    \left(
        \prod_{e\in E(D)\cap E(G[Y])}\kappa_e
        \cdot
        \prod_{v\in V_{\mathrm{in}}(D)}
        \tau_{Y,v,F_D(v),L_v}
    \right)
    \right).
    \end{aligned}
\]
Thus $R^{\mathrm{all}}$ is the version of $R$ in which the connectedness condition is not imposed.
The next lemma shows that these values are obtained from the one-pseudo-ear values $M_{Y,I,L}$.
\Cref{fig:rall-from-m} illustrates how several $M$-terms form one summand of $R^{\mathrm{all}}$.

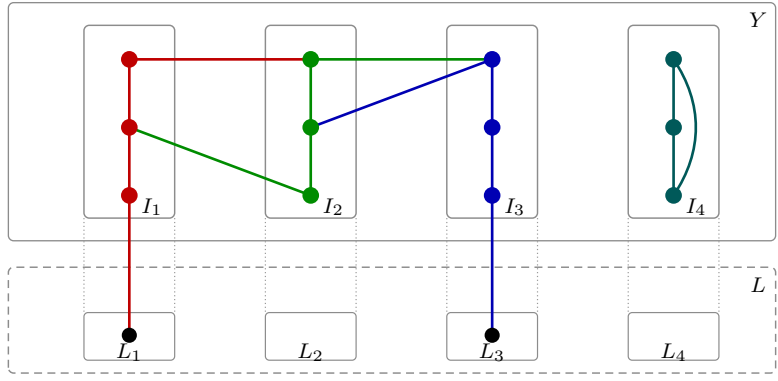
\begin{figure}[t]
\centering
\begin{tikzpicture}[
  x=1cm,y=1cm,
  vertex/.style={circle,fill=black,inner sep=0pt,minimum size=2mm},
  innerA/.style={circle,fill=red!75!black,inner sep=0pt,minimum size=2.2mm},
  innerB/.style={circle,fill=green!55!black,inner sep=0pt,minimum size=2.2mm},
  innerC/.style={circle,fill=blue!70!black,inner sep=0pt,minimum size=2.2mm},
  innerD/.style={circle,fill=teal!70!black,inner sep=0pt,minimum size=2.2mm},
  ybox/.style={draw=black!45,rounded corners=2pt,line width=0.5pt},
  ibox/.style={draw=black!45,rounded corners=2pt,line width=0.55pt},
  lbox/.style={draw=black!45,densely dashed,rounded corners=2pt,line width=0.55pt},
  lpart/.style={draw=black!45,rounded corners=1.5pt,line width=0.45pt},
  guide/.style={densely dotted,line width=0.45pt,black!40},
  earA/.style={line width=1pt,red!75!black},
  earB/.style={line width=1pt,green!55!black},
  earC/.style={line width=1pt,blue!70!black},
  earD/.style={line width=1pt,teal!70!black},
  note/.style={font=\scriptsize,inner sep=1pt},
]
\draw[ybox] (-0.25,0.6) rectangle (9.9,3.75);
\node[note] at (9.67,3.52) {$Y$};
\draw[lbox] (-0.25,-1.15) rectangle (9.9,0.25);
\node[note] at (9.67,0.02) {$L$};
\draw[lpart] (0.75,-0.98) rectangle (1.95,-0.35);
\draw[lpart] (3.15,-0.98) rectangle (4.35,-0.35);
\draw[lpart] (5.55,-0.98) rectangle (6.75,-0.35);
\draw[lpart] (7.95,-0.98) rectangle (9.15,-0.35);
\node[note] at (1.35,-0.87) {$L_1$};
\node[note] at (3.75,-0.87) {$L_2$};
\node[note] at (6.15,-0.87) {$L_3$};
\node[note] at (8.55,-0.87) {$L_4$};

\draw[ibox] (0.75,0.9) rectangle (1.95,3.45);
\draw[ibox] (3.15,0.9) rectangle (4.35,3.45);
\draw[ibox] (5.55,0.9) rectangle (6.75,3.45);
\draw[ibox] (7.95,0.9) rectangle (9.15,3.45);
\node[note] at (1.65,1.05) {$I_1$};
\node[note] at (4.05,1.05) {$I_2$};
\node[note] at (6.45,1.05) {$I_3$};
\node[note] at (8.85,1.05) {$I_4$};
\foreach \x/\y in {0.75/1.95,3.15/4.35,5.55/6.75,7.95/9.15} {
  \draw[guide] (\x,0.9) -- (\x,-0.35);
  \draw[guide] (\y,0.9) -- (\y,-0.35);
}

\coordinate (a) at (1.35,-0.65);
\coordinate (v1) at (1.35,1.2);
\coordinate (v2) at (1.35,2.1);
\coordinate (v3) at (1.35,3.0);
\coordinate (w1) at (3.75,3.0);
\coordinate (w2) at (3.75,2.1);
\coordinate (w3) at (3.75,1.2);
\coordinate (b) at (6.15,-0.65);
\coordinate (z1) at (6.15,1.2);
\coordinate (z2) at (6.15,2.1);
\coordinate (z3) at (6.15,3.0);
\coordinate (r1) at (8.55,1.2);
\coordinate (r2) at (8.55,2.1);
\coordinate (r3) at (8.55,3.0);

\draw[earA] (a) -- (v1) -- (v2) -- (v3) -- (w1);
\draw[earB] (v2) -- (w3) -- (w2) -- (w1) -- (z3);
\draw[earC] (b) -- (z1) -- (z2) -- (z3) -- (w2);
\draw[earD] (r1) -- (r2) -- (r3) to[bend left=34] (r1);

\foreach \p in {a,b}
  \node[vertex] at (\p) {};
\foreach \p in {v1,v2,v3}
  \node[innerA] at (\p) {};
\foreach \p in {w1,w2,w3}
  \node[innerB] at (\p) {};
\foreach \p in {z1,z2,z3}
  \node[innerC] at (\p) {};
\foreach \p in {r1,r2,r3}
  \node[innerD] at (\p) {};
\end{tikzpicture}
\caption{Several \(M\)-terms forming one summand of \(R^{\mathrm{all}}\).  The sets \(I_1,I_2,I_3,I_4\) partition \(Y\), and the colored objects represent the corresponding pseudo-ears; the edges of \(A\) are omitted.  The set \(L\) is split as \(L_1\uplus L_2\uplus L_3\uplus L_4\), where \(L_j\) is the outside set assigned to the \(M\)-term on \(I_j\).  An end edge of one \(M\)-term may go to another \(M\)-term: for example, the red and blue pseudo-ears have endpoints in \(I_2\), and the green pseudo-ear has endpoints in \(I_1\) and \(I_3\).  At this stage, the resulting graph on \(Y\) need not be connected.}
\label{fig:rall-from-m}
\end{figure}

\begin{lemma}\label{lem:app-rall-expansion}
For nonempty $Y$ and $L\subseteq V(G)\setminus Y$, we have
\[
    R^{\mathrm{all}}_{Y,u,F,L}(\kappa,\tau)
    =
    \left(\prod_{e\in E(G[Y])}(1+\kappa_e)\right)
    \left(
    \sum_{t=1}^{|Y|}
    \sum_{\{I_1,\dots,I_t\}\in\Pi_t(Y)}
    \sum_{\substack{L_1,\dots,L_t\\ \biguplus_{j=1}^tL_j=L}}
    \prod_{j=1}^{t}M_{Y,I_j,L_j}(\kappa,\tau)
    \right).
\]
Moreover, for fixed $Y,u,F$ and $\tau$, all values $R^{\mathrm{all}}_{Y,u,F,L}(\kappa,\tau)$ with $L\subseteq V(G)\setminus Y$ can be computed from the values $M_{Y,I,L'}(\kappa,\tau)$ in $O^*(2^{|V(G)|})$ arithmetic operations.
\end{lemma}
\begin{proof}
First,
\[
    \prod_{e\in E(G[Y])}(1+\kappa_e)
    =
    \sum_{A\subseteq E(G[Y])}\prod_{e\in A}\kappa_e .
\]
Thus the first factor is exactly the sum over choices of $A\subseteq E(G[Y])$, with weight $\prod_{e\in A}\kappa_e$, appearing in the definition of $R^{\mathrm{all}}_{Y,u,F,L}$.
Next fix $t$ and a partition $\{I_1,\dots,I_t\}\in\Pi_t(Y)$.
Expanding the $M$-terms and summing over the disjoint choices of $L_1,\dots,L_t$ gives
\[
    \begin{aligned}
    &\sum_{\substack{L_1,\dots,L_t\\ \biguplus_{j=1}^tL_j=L}}
    \prod_{j=1}^{t}M_{Y,I_j,L_j}(\kappa,\tau)\\
    &\quad =
    \sum_{(D_1,\dots,D_t)\in\mathcal E_{Y,I_1}\times\cdots\times\mathcal E_{Y,I_t}}
    \sum_{\substack{(L_v)_{v\in Y}\\ \biguplus_{v\in Y}L_v=L}}
    \prod_{j=1}^{t}
    \left(
        \prod_{e\in E(D_j)\cap E(G[Y])}\kappa_e
        \cdot \prod_{v\in I_j}\tau_{Y,v,F_{D_j}(v),L_v}
    \right).
    \end{aligned}
\]
Indeed, for each $j$, expanding $M_{Y,I_j,L_j}$ chooses a pseudo-ear $D_j\in\mathcal E_{Y,I_j}$ and a disjoint union $\biguplus_{v\in I_j}L_v=L_j$.
Since $I_1,\dots,I_t$ partition $Y$ and $L_1,\dots,L_t$ disjointly union to $L$, this is the same as choosing $(L_v)_{v\in Y}$ with $\biguplus_{v\in Y}L_v=L$.
Thus, after summing over $t$ and over the partitions of $Y$, the second factor in the statement is the sum over spanning pseudo-ear packings of $Y$:
\[
    \sum_{\mathcal D}
    \sum_{\substack{(L_v)_{v\in Y}\\ \biguplus_{v\in Y}L_v=L}}
    \prod_{D\in\mathcal D}
    \left(
        \prod_{e\in E(D)\cap E(G[Y])}\kappa_e
        \cdot \prod_{v\in V_{\mathrm{in}}(D)}\tau_{Y,v,F_D(v),L_v}
    \right),
\]
where $\mathcal D$ ranges over all spanning pseudo-ear packings of $Y$.
Each such packing appears once, because it is produced only from the partition \(\{V_{\mathrm{in}}(D):D\in\mathcal D\}\) of \(Y\).
If \(u\in Y\) and \(D_u\) is the pseudo-ear through \(u\), then the factor \(\tau_{Y,u,F_{D_u}(u),L_u}\) can be nonzero only when \(F_{D_u}(u)=F\), which is the condition imposed in \(\mathcal R^{\mathrm{all}}_{Y,u,F}\).
Thus the surviving pseudo-ear packings are exactly those allowed by \(\mathcal R^{\mathrm{all}}_{Y,u,F}\).
Because \(A\) and \(\mathcal D\) are chosen independently, multiplying this expansion by \(\prod_{e\in E(G[Y])}(1+\kappa_e)\) yields \(R^{\mathrm{all}}_{Y,u,F,L}(\kappa,\tau)\).

For the running-time assertion, define a value $x_J$ for each $J\subseteq V(G)$ by letting $x_J:=M_{Y,J\cap Y,J\setminus Y}(\kappa,\tau)$ if $J\cap Y\neq\emptyset$, and $x_J:=0$ otherwise.
Apply \Cref{lem:app-partition-convolution} to the universe $V(G)$ with weight $x_J$ assigned to each part $J\subseteq V(G)$.
For each $t$, \Cref{lem:app-partition-convolution} computes the left-hand side of
\[
    \sum_{\{J_1,\dots,J_t\}\in\Pi_t(Y\cup L)}
    \prod_{j=1}^t x_{J_j}
    =
    \sum_{\{I_1,\dots,I_t\}\in\Pi_t(Y)}
    \sum_{\substack{L_1,\dots,L_t\\ \biguplus_{j=1}^t L_j=L}}
    \prod_{j=1}^{t}M_{Y,I_j,L_j}(\kappa,\tau).
\]
Hence summing these values over $t=1,\dots,|Y|$ gives the second factor in the statement.
Thus all $L$-values are obtained in $O^*(2^{|V(G)|})$ arithmetic operations.
The factor $\prod_{e\in E(G[Y])}(1+\kappa_e)$ is then multiplied into every entry.
\end{proof}

The next lemma computes $R$ from the values $R^{\mathrm{all}}$ by imposing the connectedness condition.

\begin{lemma}\label{lem:app-local-z-substitution}
For a partition $\mathcal P$ of $Z$, let $Y_u(\mathcal P)$ be the part containing $u$.
Fix $h,Z,u,F$, let $\tau=\tau^{h,u,F}(\kappa,\eta)$, and let $L\subseteq V(G)\setminus Z$.
Then
\[
    R_{Z,u,F,L}(\kappa,\tau)
    =
    \sum_{t=1}^{|Z|}
    (-1)^{t-1}(t-1)!
    \sum_{\substack{\mathcal P=\{Y_1,\dots,Y_t\}\in\Pi_t(Z)\\ V(F)\cap Z\subseteq Y_u(\mathcal P)}}
    \left(
    \sum_{\substack{L_1,\dots,L_t\\ \biguplus_{i=1}^tL_i=L}}
    \prod_{i=1}^{t}
    R^{\mathrm{all}}_{Y_i,u,F,L_i}(\kappa,\tau)
    \right).
\]
Moreover, for fixed $Z,u,F$ and $\tau$, all values $R_{Z,u,F,L}(\kappa,\tau)$ with $L\subseteq V(G)\setminus Z$ can be computed from the values $R^{\mathrm{all}}_{Y,u,F,L'}(\kappa,\tau)$, where $\emptyset\neq Y\subseteq Z$ and $L'\subseteq V(G)\setminus Z$, in $O^*(2^{|V(G)|})$ arithmetic operations.
\end{lemma}
\begin{proof}
Fix $(A,\mathcal D)\in\mathcal R^{\mathrm{all}}_{Z,u,F}$ and a tuple $(L_v)_{v\in Z}$ with $\biguplus_{v\in Z}L_v=L$.
Let $H_Z(A,\mathcal D):=(Z,A\cup\bigcup_{D\in\mathcal D}(E(D)\cap E(G[Z])))$.
We show that the right-hand side adds
\[
    \prod_{e\in A}\kappa_e
    \cdot
    \prod_{D\in\mathcal D}
    \left(
        \prod_{e\in E(D)\cap E(G[Z])}\kappa_e
        \cdot
        \prod_{v\in V_{\mathrm{in}}(D)}
        \tau_{Z,v,F_D(v),L_v}
    \right)
\]
with coefficient $1$ if $H_Z(A,\mathcal D)$ is connected and with coefficient $0$ otherwise.
The equality in the lemma follows from this coefficient statement, since \(R_{Z,u,F,L}\) is obtained from \(R^{\mathrm{all}}_{Z,u,F,L}\) by imposing this connectedness condition.
Let $\mathcal K$ be the partition of $Z$ into the connected components of $H_Z(A,\mathcal D)$.

Fix a partition $\mathcal P=\{Y_1,\dots,Y_t\}$ of $Z$, and let \(L_i=\biguplus_{v\in Y_i}L_v\).
We prove that the contribution of $(A,\mathcal D,(L_v)_{v\in Z})$ is counted by $\prod_{i=1}^{t}R^{\mathrm{all}}_{Y_i,u,F,L_i}(\kappa,\tau)$ exactly when $\mathcal K$ is a refinement of $\mathcal P$.
Suppose first that $\prod_{i=1}^{t}R^{\mathrm{all}}_{Y_i,u,F,L_i}(\kappa,\tau)$ counts this contribution.
Edges of $A$ counted in the $i$-th factor lie in $E(G[Y_i])$.
Consider a pseudo-ear \(D\) counted in the $i$-th factor and an edge occurrence of \(H_Z(A,\mathcal D)\) leaving \(Y_i\).
Let \(v\in V_{\mathrm{in}}(D)\cap Y_i\) be the endpoint of this occurrence in \(Y_i\), and let \(w\in Z\setminus Y_i\) be the other endpoint.
If \(v=u\), then this occurrence is an edge of \(F\), but edges in \(F\) cannot join two parts of \(\mathcal P\), because \(V(F)\cap Z\subseteq Y_u(\mathcal P)\).
If \(v\neq u\), then \(\tau_{Y_i,v,F_D(v),L_v}\) is nonzero, so the condition \(V(F_D(v))\setminus Y_i\subseteq L_v\) forces \(w\in L_v\subseteq L_i\subseteq V(G)\setminus Z\), contradicting \(w\in Z\).
Thus $H_Z(A,\mathcal D)$ has no edge joining two parts of $\mathcal P$, so every connected component of $H_Z(A,\mathcal D)$ is contained in one part of $\mathcal P$.

Conversely, assume that $\mathcal K$ is a refinement of $\mathcal P$.
Then every edge of $H_Z(A,\mathcal D)$ lies inside one part $Y_i$, and each pseudo-ear of $\mathcal D$ has all its inner vertices in one part $Y_i$.
Since $F_{D_u}(u)=F$, every vertex of $V(F)\cap Z$ lies in the component of $H_Z(A,\mathcal D)$ containing $u$, and hence \(V(F)\cap Z\subseteq Y_u(\mathcal P)\).
Therefore $\prod_{i=1}^{t}R^{\mathrm{all}}_{Y_i,u,F,L_i}(\kappa,\tau)$ counts the contribution of $(A,\mathcal D,(L_v)_{v\in Z})$.

It remains to show that the \(\tau\)-factor in \(R^{\mathrm{all}}_{Y_i,u,F,L_i}\) is equal to the corresponding \(\tau\)-factor in \(R^{\mathrm{all}}_{Z,u,F,L}\).
By the refinement assumption, for each \(D\in\mathcal D\) with \(V_{\mathrm{in}}(D)\subseteq Y_i\) and each \(v\in V_{\mathrm{in}}(D)\setminus\{u\}\), no endpoint of an edge in \(F_D(v)\) lies in \(Z\setminus Y_i\).
Thus the condition \(V(F_D(v))\setminus Y_i\subseteq L_v\) holds if and only if \(V(F_D(v))\setminus Z\subseteq L_v\) holds.
All other conditions in the definition of \(\tau\) are the same for \(Y_i\) and \(Z\), so \(\tau_{Y_i,v,F_D(v),L_v}=\tau_{Z,v,F_D(v),L_v}\).

It follows that the contribution of $(A,\mathcal D,(L_v)_{v\in Z})$ receives the coefficient $\sum_{\mathcal P}(-1)^{|\mathcal P|-1}(|\mathcal P|-1)!$, where the sum ranges over the partitions of $Z$ obtained by merging parts of $\mathcal K$.
By \Cref{lem:partition-lattice-mobius}, this coefficient is $1$ when $\mathcal K$ has one part and $0$ otherwise.
Hence exactly the tuples with connected $H_Z(A,\mathcal D)$ remain, which is precisely the definition of $R_{Z,u,F,L}(\kappa,\tau)$.

For the running-time assertion, fix $t\in\{1,\dots,|Z|\}$.
For this fixed $t$, we compute the sum over all partitions $\mathcal P\in\Pi_t(Z)$ in the $t$-th summand of the statement, before multiplying by $(-1)^{t-1}(t-1)!$.
Define a table $x$ on $2^{V(G)}$ as follows.
For $J\subseteq V(G)$, let $x_J:=R^{\mathrm{all}}_{J\cap Z,u,F,J\setminus Z}(\kappa,\tau)$ if $J\cap Z\neq\emptyset$ and either $u\notin J\cap Z$ or $V(F)\cap Z\subseteq J\cap Z$; otherwise let $x_J:=0$.
For every $L\subseteq V(G)\setminus Z$, \Cref{lem:app-partition-convolution} computes the left-hand side of
\[
    \sum_{\{J_1,\dots,J_t\}\in\Pi_t(Z\cup L)}
    \prod_{j=1}^{t}x_{J_j}
    =
    \sum_{\substack{\mathcal P=\{Y_1,\dots,Y_t\}\in\Pi_t(Z)\\ V(F)\cap Z\subseteq Y_u(\mathcal P)}}
    \sum_{\substack{L_1,\dots,L_t\\ \biguplus_{i=1}^tL_i=L}}
    \prod_{i=1}^tR^{\mathrm{all}}_{Y_i,u,F,L_i}(\kappa,\tau),
\]
for all choices of $L$ in $O^*(2^{|V(G)|})$ arithmetic operations.
The condition $V(F)\cap Z\subseteq Y_u(\mathcal P)$ comes from the extra requirement $V(F)\cap Z\subseteq J\cap Z$ imposed in the definition of $x_J$ only when $u\in J\cap Z$: for this unique part $J$, we have $J\cap Z=Y_u(\mathcal P)$.
The right-hand side is the sum over $\mathcal P$ in the $t$-th summand of the statement, without the factor $(-1)^{t-1}(t-1)!$.
We multiply these values by $(-1)^{t-1}(t-1)!$ and add them to the table for $R_{Z,u,F,\cdot}(\kappa,\tau)$.
Summing over $t$ gives all $L$-values in $O^*(2^{|V(G)|})$ arithmetic operations.
\end{proof}

\subsection{Runtime analysis}\label{subsec:evaluation-runtime}
We now summarize the formulas obtained so far.
\Cref{lem:app-one-ear-held-karp} computes the one-pseudo-ear values $M$.
\Cref{lem:app-rall-expansion} computes $R^{\mathrm{all}}$ from these $M$-values, and \Cref{lem:app-local-z-substitution} computes the connected values $R$ from $R^{\mathrm{all}}$.
\Cref{lem:app-top-bag-substitution} computes $Q^\circ$ from $R$ using $Q^\bullet$-values on smaller vertex sets through the substitution $\tau$, while \Cref{lem:app-attachment} computes $Q^\bullet$ from $Q^\circ$.
Thus $Q^\circ$ and $Q^\bullet$ are filled together in increasing order of the vertex set size.
\Cref{lem:app-unrooted-recovery} then recovers $Q$ from $Q^\bullet$.
Finally, \Cref{lem:app-cost-from-q} expresses the desired weight polynomials $P_U$ in terms of evaluations of $Q$.
Thus the remaining task is to make this layered computation explicit and to bound the total running time.

\begin{lemma}\label{lem:app-local-bag}
Fix a nonnegative integer $h$.
Assume that all values $Q^\bullet_{X',v,F',d}(\kappa,\eta)$ with $v\in X'$ and $|X'|\leq h$ are known.
Then, for this fixed $h$, all values $R_{Z,u,F,L}(\kappa,\tau^{h,u,F}(\kappa,\eta))$ with $u\in Z$, $|Z|\geq2$, $L\subseteq V(G)\setminus Z$, and $|Z|+|L|=h+1$ can be computed in $O^*(4^{|V(G)|})$ time.
\end{lemma}
\begin{proof}
Fix $u$ and $F$, and let $\tau$ be the table $\tau^{h,u,F}(\kappa,\eta)$.
\Cref{lem:app-one-ear-held-karp} computes all $M_{Y,I,L}(\kappa,\tau)$ for a fixed $Y$ in $O^*(2^{|V(G)|})$ time.
Over all $Y\subseteq V(G)$, this costs $2^{|V(G)|}\cdot O^*(2^{|V(G)|})=O^*(4^{|V(G)|})$.
From the $M$-tables, the last assertion of \Cref{lem:app-rall-expansion} computes all $R^{\mathrm{all}}_{Y,u,F,L}(\kappa,\tau)$ for a fixed $Y$ in $O^*(2^{|V(G)|})$ time.
Over all $Y\subseteq V(G)$, this costs $2^{|V(G)|}\cdot O^*(2^{|V(G)|})=O^*(4^{|V(G)|})$.
From the $R^{\mathrm{all}}$-tables, the last assertion of \Cref{lem:app-local-z-substitution} computes all $R_{Z,u,F,L}(\kappa,\tau)$ for a fixed $Z$ in $O^*(2^{|V(G)|})$ time.
Over all $Z\subseteq V(G)$, this costs $2^{|V(G)|}\cdot O^*(2^{|V(G)|})=O^*(4^{|V(G)|})$.
The entries with $|Z|+|L|=h+1$ are among the computed $R$-values.
The choices of $u$ and $F$ only contribute polynomial factors.
\end{proof}

\begin{lemma}\label{lem:app-deterministic-q}
Given values of $\kappa_e$ and $\eta_{v,F}$, the value $Q(\kappa,\eta)$ can be computed in $O^*(4^{|V(G)|})$ arithmetic operations.
\end{lemma}
\begin{proof}
We compute $Q^\circ$ and $Q^\bullet$ in increasing order of the vertex-set size.
The initial values are $Q^\circ_{\{u\},u,F}=1$, and \Cref{lem:app-attachment} gives the corresponding singleton $Q^\bullet$-values.
Assume that all $Q^\bullet$-values on vertex sets of size at most $h$ are known.
Then \Cref{lem:app-local-bag} computes all $R_{Z,u,F,L}(\kappa,\tau^{h,u,F}(\kappa,\eta))$ with $|Z|+|L|=h+1$ in $O^*(4^{|V(G)|})$ time.
By \Cref{lem:app-top-bag-substitution}, these $R$-values give all $Q^\circ_{X,u,F}$ with $|X|=h+1$, which costs $\sum_{|X|=h+1}O^*(2^{|X|})=O^*(3^{|V(G)|})$ time.
After this layer is filled, \Cref{lem:app-attachment} computes the $Q^\bullet$-values on vertex sets of size $h+1$.

Over all layers, the applications of \Cref{lem:app-local-bag} cost $O^*(4^{|V(G)|})$ time.
For fixed $v$ and $F$, computing all values $Q^\bullet_{X,v,F,d}$ from the corresponding $Q^\circ$-values by \Cref{lem:app-attachment} costs only $\sum_{X\subseteq V(G)}O^*(2^{|X|})=O^*(3^{|V(G)|})$ time.
The choices of $v$ and $F$ only contribute polynomial factors.
Finally, \Cref{lem:app-unrooted-recovery} gives $Q(\kappa,\eta)$ from the computed $Q^\bullet$-values in polynomial time.
\end{proof}

\begin{lemma}\label{lem:deterministic-2vcssbp-all-subuniverses}
For a \textsc{2VCSS-BP} instance with universe $U_0$, all output entries for prescribed subuniverses $U\subseteq U_0$ can be computed in $O^*(4^{|V(G)|}2^{|U_0|}W)$ time, suppressing polynomial factors in $k$, where the variant for $U$ replaces (E2) by the condition that $\{S_v\}_{v\in V(G)}$ is a partition of $U$.
\end{lemma}
\begin{proof}
By \Cref{lem:app-cost-from-q}, it is enough to evaluate $Q$ after the substitutions
\[
    \kappa_e=\xi^{\chi_e}\zeta^{c_e}
    \quad\text{and}\quad
    \eta_{v,F}=\sum_{\substack{S\subseteq U_0,\ \mathbf s=(\alpha,\beta)\\ (F,S,\mathbf s)\in\mathcal C_v}}
    \omega^{|S|}\xi^{\alpha}\zeta^{\beta}\prod_{p\in S}\gamma_p,
\]
where the variables $\gamma_p$ range over all $2^{|U_0|}$ Boolean substitutions and $(\omega,\xi,\zeta)$ ranges over tuples of interpolation points.
The \(\omega\)-degree and \(\xi\)-degree are bounded by a polynomial in the input size, and the \(\zeta\)-degree is bounded by \(W\) times such a polynomial.
Since interpolation uses one more point than the corresponding degree bound in each variable, it uses \(O^*(W)\) tuples \((\omega,\xi,\zeta)\) of interpolation points.
For fixed values of \(\omega,\xi,\zeta\) and fixed \(v,F\), all values \(\eta_{v,F}(\mathbf 1_Y)\) over \(Y\subseteq U_0\) are obtained by the fast zeta transform of \Cref{lem:fast-zeta-mobius} on \(U_0\).
Thus all Boolean substitutions of the \(\eta\)-variables are computed in \(O^*(2^{|U_0|})\) arithmetic operations for each interpolation tuple.
Since there are \(O^*(W)\) interpolation tuples, these fast zeta transforms contribute \(O^*(2^{|U_0|}W)\) arithmetic operations in total.
Thus the total cost of the $Q$-evaluations, including the construction of the substituted \(\eta\)-values, is $2^{|U_0|}\cdot O^*(W)\cdot O^*(4^{|V(G)|})=O^*(4^{|V(G)|}2^{|U_0|}W)$, because each specialization is evaluated in $O^*(4^{|V(G)|})$ time by \Cref{lem:app-deterministic-q}.
\Cref{lem:app-cost-from-q} recovers all polynomials $P_U(\xi,\zeta)$ for $U\subseteq U_0$ simultaneously, using one M\"obius transform over $U_0$ for each value of $\omega,\xi,\zeta$.
These M\"obius transforms cost $O^*(2^{|U_0|}W)$ arithmetic operations in total.
After recovering the polynomials $P_U(\xi,\zeta)$, we inspect the coefficients with exponents in $\Omega$ to determine the upward-closed output entry for each $U\subseteq U_0$.
By \Cref{lem:bp_to_bpm}, this is also the output entry for the original \textsc{2VCSS-BP} instance.

It remains to justify that the arithmetic operations above can be implemented within the same $O^*$-bound in bit complexity.
Every integer coefficient reconstructed by interpolation is a signed sum of at most $2^{\poly(|V(G)|+|E(G)|+|U_0|)}$ terms, each of which is a product of polynomially many M\"obius factors $(d-1)!$ with $d\leq |V(G)|$.
Hence the reconstructed coefficients have absolute value at most $2^{\poly(|V(G)|+|E(G)|+|U_0|)}$.
We perform the evaluations and interpolations modulo primes \(p\) larger than all interpolation degrees.
Then the whole computation is valid over \(\mathbb F_p\), and interpolation recovers the desired coefficients modulo \(p\).
Since the interpolation degrees for \(\omega\) and \(\xi\) are polynomial and the interpolation degree for \(\zeta\) is \(O^*(W)\), such primes have bit length polynomial in the input size and \(\log W\).
Polynomially many such primes suffice to reconstruct the integer coefficients by the Chinese remainder theorem.
Thus each arithmetic operation in the stated arithmetic-operation bound is on integers with polynomially many bits plus \(O(\log W)\) bits.
\end{proof}

\twovcamainthm*
\begin{proof}
By \Cref{lem:deterministic-2vcssbp-all-subuniverses}, the all-subuniverse oracle required in \Cref{lem:runtime_twovca_via_oracle} is available with $\rho=4$ and $\nu=2$.
Thus \textsc{2VCA} can be solved deterministically in $O^*(\max\{3,4+2\}^{2k}W)=O^*(36^kW)$ time.
\end{proof}

\section*{Acknowledgements}
Tomohiro Koana was supported in part by JST CREST Grant Number JPMJCR24Q2 and JST ERATO Grant Number JPMJER2301.

\section*{Acknowledgement of AI Assistance}
The authors used OpenAI's ChatGPT/Codex as an AI assistance tool in preparing this manuscript.
The overall algorithmic strategy, the informal formulations and intuition for the lemmas, and the proof plans were developed by the authors.
The tool was used to help turn the authors' informal materials into formal statements and proofs, draft and revise explanatory text, refine definitions and notation, and check consistency.
It was also used to assist with surveying related techniques and work, obtaining review-style comments on drafts, drafting the technical overview, and editing LaTeX/TikZ source for the figures.
All mathematical claims, proofs, algorithms, references, and figures were reviewed, verified, and substantially revised by the authors.
The authors assume responsibility for all content.

\bibliography{refs}

\end{document}